\documentclass[11pt,reqno]{amsart}

\usepackage{amsmath,amsfonts,amsbsy}
\numberwithin{equation}{section}

\usepackage{enumitem}
\setlist{leftmargin=*}

\usepackage[a4paper]{geometry}
\newgeometry{margin=3cm}

\usepackage{xy}
\xyoption{matrix} \xyoption{arrow} \xyoption{arc} \xyoption{color}

\usepackage{tikz}
\usetikzlibrary{arrows}
\usetikzlibrary{intersections}

\usepackage{graphicx}
\usepackage{caption}
\usepackage{subcaption}

\usepackage{thmtools}

\declaretheoremstyle[
spaceabove=\medskipamount, spacebelow=\medskipamount,
headfont=\bfseries,
notefont=\bfseries\boldmath, notebraces={(}{)},
bodyfont=\itshape,
postheadspace=.5em
]{cursive}

\declaretheorem[style=cursive,name=Theorem,numberwithin=section]{theorem}
\declaretheorem[style=cursive,name=Lemma,numberlike=theorem]{lemma}

\declaretheorem[style=cursive,name=Proposition,numberlike=theorem]{proposition}
\declaretheorem[style=cursive,name=Problem,numberlike=theorem]{problem}

\declaretheoremstyle[
spaceabove=\medskipamount, spacebelow=\medskipamount,
headfont=\bfseries,
notefont=\bfseries\boldmath, notebraces={(}{)},
bodyfont=\rmfamily,
postheadspace=.5em
]{upright}

\declaretheorem[style=upright,name=Definition,qed=$\Diamond$,numberlike=theorem]{definition}
\declaretheorem[style=upright,name=Remark,qed=$\Diamond$,numberlike=theorem]{remark}

\declaretheorem[style=upright,qed=$\Diamond$,numberlike=theorem]{assumption}

\newcommand{\eps}{\varepsilon}

\newcommand{\Id}{\mathbf{1}}
\newcommand{\B}{\mathbb{B}}
\newcommand{\Beff}{\mathbb{B}\sub{eff}}
\newcommand{\C}{\mathbb{C}}
\newcommand{\R}{\mathbb{R}}
\newcommand{\Z}{\mathbb{Z}}

\newcommand{\T}{\mathbb{T}}

\newcommand{\Hi}{\mathcal{H}}
\newcommand{\Hf}{\mathcal{H}\sub{f}}
\newcommand{\U}{\mathcal{U}}
\newcommand{\A}{\mathcal{A}}
\newcommand{\BH}{\mathcal{B}(\mathcal{H})}
\newcommand{\scal}[2]{\left\langle #1, #2 \right\rangle}
\newcommand{\norm}[1]{\left\| #1 \right\|}

\newcommand{\bra}[1]{\left\langle #1 \right|}
\newcommand{\ket}[1]{\left| #1 \right\rangle}
\newcommand{\eu}{\mathrm{e}}
\newcommand{\iu}{\mathrm{i}}
\newcommand{\di}{\mathrm{d}}
\newcommand{\act}{\triangleleft}
\newcommand{\sub}[1]{_{\mathrm{#1}}}

\newcommand{\ie}{{\sl i.\,e.\ }}
\newcommand{\eg}{{\sl e.\,g.\ }}
\newcommand{\set}[1]{ \left\{  #1 \right\}}

\DeclareMathOperator{\tr}{tr}
\DeclareMathOperator{\Tr}{Tr}
\DeclareMathOperator{\Ran}{Ran}

\DeclareMathOperator{\Fr}{Fr}

\DeclareMathOperator{\Pf}{Pf}
\DeclareMathOperator{\Span}{Span}

\newcommand{\kk}{{\mathbf k}}

\title[Wannier functions and $\Z_2$ invariants in TRS topological insulators]{Wannier functions and $\Z_2$ invariants in time-reversal symmetric topological insulators}
\author{Horia D. Cornean \and Domenico Monaco \and Stefan Teufel}
\date{\today, arXiv version 2, accepted for publication in {\it Rev. Math. Phys.}} 

\usepackage
{hyperref}

\begin{document}

\begin{abstract}
We provide a constructive proof of exponentially localized Wannier functions and related Bloch frames in $1$- and $2$-dimensional time-reversal symmetric (TRS) topological insulators. The construction is formulated in terms of periodic TRS families of projectors (corresponding, in applications, to the eigenprojectors on an arbitrary number of relevant energy bands), and is thus model-independent. The possibility to enforce also a TRS constraint on the frame is investigated. This leads to a topological obstruction in dimension $2$, related to $\Z_2$ topological phases. 

We review several proposals for $\Z_2$ indices that distinguish these topological phases, including the ones by Fu--Kane \cite{FuKane06}, Prodan \cite{Prodan11_PRB}, Graf--Porta \cite{GrafPorta13} and Fiorenza--Monaco--Panati \cite{FiorenzaMonacoPanati16}. We show that all these formulations are equivalent. In particular, this allows to prove a geometric formula for the  the $\Z_2$ invariant of $2$-dimensional TRS topological insulators, originally indicated in \cite{FuKane06}, which expresses it in terms of the Berry connection and the Berry curvature.

\bigskip

\end{abstract}

\maketitle

\tableofcontents


\section{Introduction} \label{sec:intro}

The last decade has witnessed the rise of a new class of materials, called \emph{topological insulators} \cite{HasanKane10, Fu11, Ando13, FruchartCarpentier13}. These materials, although insulating in the bulk, have the property of conducting (charge or spin) currents on their boundary, making them amenable to various types of applications in material science. The property which is more relevant to our analysis is that these conducting edge states are protected by the \emph{symmetries} of the quantum system under scrutiny. Furthermore, a principle known as the \emph{bulk-edge correspondence} states that the features of these modes should be accessible also through (topological) properties of the bulk of the material.

Elaborating on the seminal work by Altland and Zirnbauer \cite{AltlandZirnbauer97}, by now several classification schemes for these \emph{topological phases of quantum matter} have been proposed, using a rich variety of mathematical tools, ranging from $K$-theory \cite{Kitaev09, Ryu_et_al10, Thiang16, ProdanSchulz-Baldes16}, to homotopy theory \cite{KennedyGuggenheim15, KennedyZirnbauer16}, from functional analysis \cite{GrossmannSchulz-Baldes16}, to noncommutative geometry \cite{Prodan11_JPA, BourneCareyRennie15}. In this work, we focus on one particular class, which was among the first to be proposed and experimentally realized \cite{KaneMele05, FuKane06, BernevigHughesZhang06}, namely that of \emph{time-reversal symmetric} (or TRS, for short) topological insulators. These models for crystalline materials are characterized by the fact that, in the physically relevant energy interval, the system possess a time-reversal symmetry of \emph{fermionic} or \emph{odd} type, namely there exists an antiunitary operator $\theta$ which implements this symmetry and satisfies $\theta^2 = - \Id$ (as opposed to the \emph{bosonic} or \emph{even} case, where $\theta^2 = \Id$). In the Altland-Zirnbauer classification, this is known as the AII class.

The mathematical description of these materials is encoded in the features of the Hamiltonian $H$ of the quantum system. This operator is assumed to commute with a representation of a group of lattice translations $\Lambda$ on the one-particle Hilbert space, implementing the crystalline nature of the material: in $d$ dimensions, the group $\Lambda$ can be identified with the integer lattice $\Z^d$. This commutation relation leads to the well-known \emph{Bloch-Floquet-Zak reduction} (see \eg \cite{MonacoPanati15} for a concise review on the subject), which fibers the Hamiltonian over the crystal momentum $\kk \in \R^d / \Lambda^* \simeq \T^d$. In particular, the band-gap structure of the spectrum of $H$ can be reconstructed by the spectral properties of the fiber Hamiltonians $H(\kk)$. We focus on a physically relevant energy interval: for example, in solid state physics applications it is usually assumed that the Fermi energy $E\sub{F}$ lies in a gap of the Hamiltonian, and the relevant interval consists then of the bands below $E\sub{F}$. This allows to define the \emph{spectral projector} $P$ on this energy interval, and the associated fibre projectors $P(\kk)$.

Time-reversal symmetry is then implemented antiunitarily on the Hilbert space, and squares to the operator of multiplication by $-1$. Requiring that the Hamiltonian has this symmetry implies then the existence of an antiunitary operator $\theta$ on the fiber Hilbert space $\Hi$, on which the fiber Hamiltonian $H(\kk)$ acts, satisfying $\theta^2 = - \Id$ and
\[ H(-\kk) = \theta \, H(\kk) \, \theta^{-1}. \]
In particular, the time-reversal operator $\theta$ also intertwines the spectral projectors $P(\kk)$ and $P(-\kk)$. 

The study of effective models for the relevant quantum system in the selected energy interval requires the construction of a basis for the range of the spectral projection $P$. Since the Hamiltonian $H$ has absolutely continuous spectrum, as dictated by periodicity, one cannot rely on a set of eigenfunctions. However, an orthonormal set of ``orbital-like'' functions which can be used in their place was introduced by Wannier, and are now commonly used in theoretical and numerical solid state physics. The literature on \emph{Wannier functions} is by now very vast: we refer to \cite{CorneanHerbstNenciu15} and references therein for a recent discussion on the subject.

The construction of an orthonormal basis of exponentially localized Wannier functions is thus of uttermost importance in solid state physics problems. The problem can be reformulated in $\kk$-space as the existence of a set of \emph{real analytic} functions $\R^d \ni \kk \mapsto \Xi_a(\kk) \in \Hi$, $a \in \set{1, \ldots, \dim \Ran P(\kk)}$, which are \emph{periodic} and provide and orthonormal basis of the vector space $\Ran P(\kk)$. It was early realized \cite{Nenciu91} that the existence of such \emph{Bloch frames}%
\footnote{%
It should be emphasized that the terminology ``frame'' is used here as a synonym for ``moving reference systems'' rather than for ``linearly dependent generating system''. A Bloch frame provides indeed an \emph{orthonormal basis} in the range of the projector (compare Definition \ref{dfn:Bloch}).} %
is in general \emph{topologically obstructed}, as regularity and periodicity may compete against each other. This topological obstruction is encoded in the geometry of the \emph{Bloch bundle}, a vector bundle naturally associated to the family of projectors $\set{P(\kk)}_{\kk \in \R^d}$ \cite{Panati07}. From this geometric object one can extract a set of integers (\emph{Chern numbers}), whose vanishing allows the construction of exponentially localized Wannier functions. Time-reversal symmetry (either of bosonic or fermionic type) kills the topological obstruction whenever $d \le 3$, and gapped periodic quantum systems enjoying this symmetry indeed admit an orthonormal basis of exponentially localized Wannier functions. This was proved in \cite{Panati07, MonacoPanati15} by a careful and sophisticated analysis of the geometry of the Bloch bundle; it must be noticed, for the intent of the present work, that these existence result however fail to produce explicitly the required analytic and periodic Bloch frames, and consequently exponentially localized Wannier functions.

When TRS is of fermionic nature, however, a finer topological information can be extracted from these systems. This was first proposed by Fu and Kane \cite{FuKane06} for $2$-dimensional systems, with later generalizations to $3$-dimensional systems by Fu, Kane and Mele \cite{FuKaneMele07}. These authors associated \emph{$\Z_2$-valued indices} to TRS topological insulators, whose non-triviality characterizes in particular \emph{quantum spin Hall phases} \cite{KaneMele05, FuKane06}. After the proposal by Fu, Kane and Mele, the mathematical physics community embarked in the task of understanding the geometry of these invariants and their relation with observable currents present in these systems, also through the bulk-edge correspondence \cite{AvilaSchulz-BaldesVillegas-Blas12, GrafPorta13, Schulz-Baldes13, DeNittisGomi15}.

A geometric interpretation of the Fu--Kane $\Z_2$ index in terms of obstruction theory has recently been proposed in \cite{FiorenzaMonacoPanati16}. There it is recognized that the existence of continuous and periodic Bloch frames which are also compatible with TRS is in general topologically obstructed, and that this obstruction can be encoded in a $\Z_2$-valued topological invariant $\delta$ associated to the family of projectors $\set{P(\kk)}_{\kk \in \R^2}$; this invariant is then shown to coincide numerically with the Fu--Kane index. Although the Fiorenza--Monaco--Panati invariant from \cite{FiorenzaMonacoPanati16} characterizes up to unitary equivalence the family of projectors to which it is associated (inside the class of TRS Bloch bundles), its definition is rather involved and indirect, and an expression for $\delta$ which explicitly depends only on the family of projectors itself is still missing.

We are now in position to state the goals of the present analysis. We are first of all interested in giving a \emph{constructive proof} of the existence of exponentially localized Wannier functions in gapped periodic quantum systes with a \emph{fermionic} TRS in dimension $d \le 2$. We discuss also the possibility of imposing a natural compatibility condition with the time-reversal operator, and the relation of the latter with the $\Z_2$ invariants. We will prove that analytic and periodic frames can indeed be constructed in dimension $d \le 2$, and we will encounter a topological obstruction in $d=2$ when we require also a TRS property for the frame to hold (see Theorem \ref{thm:MainResults_BB}). No such obstruction is instead present in $d=1$, thus recovering and generalizing previous results formulated in the bosonic setting (see \eg \cite{NenciuNenciu98} and references therein). The $2$-dimensional construction procedure can go through provided the \emph{Graf--Porta $\Z_2$ index} \cite{GrafPorta13} vanishes: the latter was introduced as a ``bulk invariant'' for $2$-dimensional TRS topological insulators, as a fundamental step in the proof of the bulk-edge correspondence for this class of materials. As the topological obstruction was encoded in \cite{FiorenzaMonacoPanati16} in the $\Z_2$ invariant $\delta$, it is natural to ask how these two quantities relate. We prove that indeed they agree numerically (compare Theorem \ref{thm:MainResults_Z2}). The equality of the two $\Z_2$ indices produces a fruitful interaction, leading to a deeper understanding of their shared properties: the Graf-Porta index becomes manifestly a topological invariant of the quantum system, while in turn the topological obstruction $\delta$ is brought into contact with the actual process of detection of topological phases in real experimental setups in view of the bulk-edge correspondence proved in \cite{GrafPorta13}. While well beyond the scope of this paper, the investigation of a direct connection of the obstruction-theoretic invariant with surface states, which are moreover robust against disorder and impurities which break microscopically the symmetries of the system, remains a stimulating line of research for the future.

Moreover, from its equality with the Graf-Porta index we are able to deduce an expression for the $\Z_2$ invariant $\delta$ which depends explicitly on the family of projections $\set{P(\kk)}_{\kk \in \R^2}$ to which it is associated, via its \emph{Berry connection} and \emph{Berry curvature} (see Theorem \ref{thm:MainResults_delta}). This formula was found already by Fu and Kane in the original paper on the $\Z_2$ index \cite[Eqn.~(A8)]{FuKane06}. Our proof complies with the geometric character of the problem, in that it uses geometric objects naturally associated to the family of projections (mainly the \emph{parallel transport} relative to the Berry connection).

The paper is structured as follows. First, we formulate precisely our main results, which were sketched above, and set some notation in Section \ref{sec:main}. We then reformulate in Section \ref{sec:d=d} the problem of the construction of analytic, periodic, and possibly TRS Bloch frames for a $d$-dimensional family of projectors $\set{P(\kk)}_{\kk \in \R^d}$, whose properties are modelled after the ones of the eigenprojectors of the Hamiltonian of a gapped periodic quantum system with fermionic TRS, in terms of an equivalent problem for particular families of unitary matrices. The latter emerge from the construction of a $d$-dimensional frame as follows: Given an input frame at $k_1=0$, this is parallel transported along the direction $k_1$, and the failure of this procedure to produce a frame which is periodic with respect to $k_1$ is measured by a family of unitary matrices $\set{\alpha(\kk_\perp)}_{\kk_\perp \in \R^{d-1}}$ depending on the other $d-1$ coordinates. These matrices are again periodic and satisfy a TRS condition. The possibility to construct a periodic and TRS frame is then equivalent to the possibility of ``rotating'' this family $\alpha$ to the identity matrix, by preserving its properties.

We then specialize and solve the problem in dimension $d=1$ in Section \ref{sec:d=1}, and in dimension $d=2$ in Section \ref{sec:d=2}.  The general algorithm that we propose consists of the following steps:
\begin{itemize}
 \item if the family $\alpha$ admits a ``good logarithm'', \ie if it can be written as $\alpha(\kk_\perp) = \eu^{\iu \, h(\kk_\perp)}$ with a continuous, periodic, and possibly TRS logarithm $h(\kk_\perp) = h(\kk_\perp)^*$, then it is possible to rotate it to the identity via the procedure described in Proposition \ref{prop:GoodLog};
 \item when instead the family does not admit such a logarithm, then approximate it with a family which does, and has the same properties (continuity, periodic, TRS); then apply the procedure illustrated in Proposition \ref{prop:AlmostLog} to rotate it to the identity.
\end{itemize}
The problem is solved by constructing explicitly these ``good logarithms''; some technical results which are necessary for this construction are contained in Appendix \ref{sec:ExtraDegen}. The first of the above cases is realized in $d=1$, while the second occurs in $d=2$, and may require further topological constraints to be satisfied. We recognize in particular that the topological obstruction to the existence of a ``good logarithm'' in $d=2$ is encoded in the Graf--Porta $\Z_2$ index \cite{GrafPorta13}, for which we prove several useful properties, including the fact that it characterizes completely the homotopy class of a continuous, periodic and TRS family of unitary matrices. We stress once again that all the steps in this construction are explicit and operational.

The realization that the Graf--Porta index is a topological obstruction gives us the opportunity to discuss and review several approaches to the formulation of $\Z_2$ invariants distinguishing the different topological phases in $2$-dimensional TRS topological insulators, and in particular in quantum spin Hall systems. As was already mentioned above, we prove in Section \ref{sec:Z2} the equivalence between the Graf--Porta and the Fiorenza--Monaco--Panati invariants \cite{GrafPorta13, FiorenzaMonacoPanati16}. This allows us to show that the $\Z_2$ invariant can be expressed in geometric terms as a function of the family of projectors, as explained above.

\subsection{Comparison with the literature}

We conclude the Introduction with a comparison of our methods with similar approaches present in the literature. 

As a first observation, we would like to point out that constructions of Wannier functions are not new, also in the mathematical physics community. After the early attempts in dimension $d=1$, constructive proofs of their existence also in $d>1$ were recently proposed (see \cite{CorneanHerbstNenciu15, FiorenzaMonacoPanati16_B, CancesLevittPanatiStoltz16} and references therein), when \emph{bosonic} time-reversal symmetry is present. The situation is substantially different in this case, as no topological obstructions arise, and exponentially localized Wannier function which are moreover \emph{real-valued} (\ie invariant under $\theta := C$, the complex conjugation operator) can be constructed. 

The literature on the case of fermionic TRS, at least as the Wannier function issue is concerned, is less developed. There are of course a few exceptions, on which we would like to comment now. Apart from identifying the geometric nature of the $\Z_2$ invariants for $2$- and $3$-dimensional TRS topological insulators, the work by Fiorenza, Monaco and Panati \cite{FiorenzaMonacoPanati16} provides a construction for smooth $d$-dimensional Bloch frames which enjoy periodicity and TRS whenever $d \le 3$, when such obstructions vanish. However, a construction of a Bloch frame which does \emph{not} satisfy TRS, being still smooth and periodic, is missing in \cite{FiorenzaMonacoPanati16}. In contrast, the method employed \eg by Soluyanov and Vanderbilt \cite{SoluyanovVanderbilt12} (see also \cite{WinklerSoluyanovTroyer16} for more recent developments) constructs smooth periodic frames in $d=2$, also in the non-trivial topological phase, but fails to produce TRS frames when the topological obstruction vanishes. The parallel transport method that we employ is used also in \cite{SoluyanovVanderbilt12}, but their construction is limited to the two-band case (\ie in our language when $\dim \Ran P(\kk) = 2$), while our proofs work in arbitrary rank.

Finally, we would like to point out the work by Prodan \cite{Prodan11_PRB}, where the $\Z_2$ invariant is also expressed in terms of the matrices associated to parallel transport. While the formula presented in \cite{Prodan11_PRB} is manifestly gauge-independent, its geometric interpretation is more hidden, since it requires evaluation of certain quantities at high-symmetry points (in the spirit of the ``Pfaffian'' formulation of the Fu--Kane index \cite{FuKane06}). The formula we derive for the $\Z_2$ invariant is instead truly geometric, as it depends only on quantities (the integrals of the Berry connection and Berry curvature) which can be expressed directly in terms of the relevant family of projectors. The relation of the Prodan index with the other formulations will be discussed further in Section \ref{sec:TRIM}.

\bigskip 

\noindent \textbf{Acknowledgments.} The authors would like to thank D.~Fiorenza, G.~Panati, and H.~Schulz-Baldes for stimulating discussions. The financial support from the Danish Council for Independent Research $|$ Natural Sciences within Grant 4181-00042 and from the German Science Foundation (DFG) within the GRK 1838 ``Spectral theory and dynamics of quantum systems'' is gratefully acknowledged.



\section{Statement of the problem and main results} \label{sec:main}

\subsection{Statement of the problem} \label{sec:problem}

Inspired by gapped periodic quantum systems with fermionic time-reversal symmetry in the Bloch-Floquet-Zak representation%
\footnote{The reduction of the problem from periodic time-reversal symmetric Hamiltonians to families of projectors as in Assumption \ref{assum:proj} is given \eg in \cite[Sec.~1.2]{CorneanHerbstNenciu15}.}%
, we study families of projectors satisfying the hypotheses collected below in Assumption \ref{assum:proj}. In what follows, we let $\Hi$ be a separable Hilbert space with scalar product $\scal{\cdot}{\cdot}$, $\BH$ denote the algebra of bounded linear operators on $\Hi$, and $\U(\Hi)$ the group of unitary operators on $\Hi$.

\begin{assumption} \label{assum:proj}
The family of orthogonal projectors $\set{P(\kk)}_{\kk \in \R^d} \subset \BH$, $P(\kk) = P(\kk)^2 = P(\kk)^*$, enjoys the following properties:
\begin{enumerate}[label=$(\mathrm{P}_\arabic*)$,ref=$(\mathrm{P}_\arabic*)$]
 \item \label{item:smooth} \emph{analyticity}: the map $\R^d \ni \kk \mapsto P(\kk) \in \BH$ is real analytic;
\item \label{item:periodic} \emph{periodicity}: the map $\kk \mapsto P(\kk)$ is $\Z^d$-periodic, namely 
\[ P(\kk + \mathbf{n}) = P(\kk) \quad \text{for all } \mathbf{n} \in \Z^d; \]
\item \label{item:TRS} \emph{time-reversal symmetry}: the map $\kk \mapsto P(\kk)$ is time-reversal symmetric (TRS), \ie there exists an antiunitary
operator%
\footnote{Recall that a surjective antilinear operator $K \colon \Hi \to \Hi$ is called \emph{antiunitary} if $\scal{K \Xi_1}{K \Xi_2} = \scal{\Xi_2}{\Xi_1}$ for all $\Xi_1, \Xi_2 \in \Hi$.} %
$\theta \colon \Hi \to \Hi$, called the \emph{time-reversal operator}, such that 
\[ \theta^2 = - \Id_{\Hi} \quad \text{and} \quad P(-\kk) = \theta P(\kk) \theta^{-1}. \qedhere\] 
\end{enumerate}
\end{assumption}

For a family of projectors satisfying Assumption \ref{assum:proj}, it follows from \ref{item:smooth} that the rank $m$ of the projectors $P(\kk)$ is constant in $\kk$. We will assume that $m < + \infty$; property \ref{item:TRS} then gives that $m$ must be even. Indeed, the formula
\begin{equation} \label{eqn:symplectic}
(\Xi_1, \Xi_2) := \scal{\theta \Xi_1}{\Xi_2} \quad \text{for } \Xi_1, \Xi_2 \in \Hi
\end{equation}
defines a bilinear, skew-symmetric, non-degenerate form on $\Hi$; its restriction to the subspace $\Ran P(\mathbf{0}) \subset \Hi$, which is invariant under the action of $\theta$ in view of \ref{item:TRS}, is then a \emph{symplectic form}, and a symplectic vector space is necessarily even-dimensional.

The goal of our analysis will be to construct, whenever possible, a \emph{continuous symmetric Bloch frame} for the family $\set{P(\kk)}_{\kk \in \R^d}$, which we define now.

\begin{definition}[(Symmetric) Bloch frame] \label{dfn:Bloch}
Let $\set{P(\kk)}_{\kk \in \R^d}$ be a family of projectors satisfying Assumption \ref{assum:proj}. A \emph{Bloch frame} for $\set{P(\kk)}_{\kk \in \R^d}$ is a collection of maps $\R^d \ni \kk \mapsto \Xi_{a}(\kk) \in \Hi$, $a \in \set{1, \ldots, m}$, such that for all $\kk \in \R^d$ the set $\Xi(\kk) := \set{\Xi_1(\kk), \ldots, \Xi_m(\kk)}$ is an orthonormal basis spanning $\Ran P(\kk)$. A Bloch frame is called
\begin{enumerate}[label=$(\mathrm{F}_\arabic*)$,ref=$(\mathrm{F}_\arabic*)$]
 \item \label{item:F1} \emph{continuous} if all functions $\Xi_a \colon \R^d \to \Hi$, $a \in \set{1, \ldots, m}$, are continuous;
 \item \label{item:F2} \emph{periodic} if 
\[ \Xi_a(\kk + \mathbf{n}) = \Xi_a(\kk) \quad \text{for all } \kk \in \R^d, \: \mathbf{n} \in \Z^d, \: a \in \set{1, \ldots, m}; \]
\item \label{item:F3} \emph{time-reversal symmetric} (TRS) if
\begin{equation} \label{eqn:TRS}
\Xi_b(-\kk) = \sum_{a = 1}^{m} [\theta \Xi_a(\kk)] \eps_{ab} \quad \text{for all } \kk \in \R^d, \: b \in \set{1, \ldots, m}
\end{equation}
for some unitary and skew-symmetric matrix $\eps = (\eps_{ab})_{1 \le a,b \le m} \in U(m) \cap \bigwedge^2 \C^m$.
\end{enumerate}

A Bloch frame which is both periodic and time-reversal symmetric is called \emph{symmetric}.
\end{definition}

\begin{remark}[The reshuffling matrix $\eps$] \label{remark:eps}
The necessity of a reshuffling matrix $\eps = (\eps_{ab})_{1 \le a,b \le m}$ in the definition of TRS Bloch frames is evident if one looks at the $\theta$-invariant subspace at $\kk = \mathbf{0}$. Indeed, since the restriction of \eqref{eqn:symplectic} to $\Ran P(\mathbf{0})$ is a symplectic form, the vectors $\Xi \in \Ran P(\mathbf{0})$ and $\theta \Xi$ are orthogonal to each other. In particular, $\Xi = \theta \Xi$ implies $\Xi = 0$; this prevents to set the (naive) definition of TRS Bloch frame as $\Xi_a(-\kk) = \theta \Xi_a(\kk)$.

The skew-symmetry of $\eps$ follows also as a self-consistency condition for \eqref{eqn:TRS}. Indeed, by applying the antilinear operator $\theta$ to both sides of \eqref{eqn:TRS} we obtain
\[ \theta \Xi_b(-\kk)  = \sum_{a=1}^{m} [\theta^2 \Xi_a(\kk)] \overline{\eps_{ab}} = - \sum_{a=1}^{m} \Xi_a(\kk) \, \overline{\eps_{ab}} = - \sum_{a,c=1}^{m} [\theta \Xi_c(-\kk)] \eps_{ca} \, \overline{\eps_{ab}}. \]
By taking the scalar product of both sides of the above equation with $\theta \Xi_d(-\kk)$, and using the fact that
\[ \scal{\theta \Xi_d(-\kk)}{\theta \Xi_b(-\kk)} = \scal{\Xi_b(-\kk)}{\Xi_d(-\kk)} = \delta_{d,b} \]
because $\theta$ is antiunitary and $\Xi(-\kk)$ is an orthonormal basis of $\Ran P(-\kk)$, we deduce
\[ \delta_{d,b} = - \sum_{c=1}^{m} \delta_{d,c} \left(\eps \overline{\eps} \right)_{cb} = - \left(\eps \overline{\eps} \right)_{db} \]
or, in matrix form, $\eps \overline{\eps} = - \Id$. By unitarity of $\eps$, however, we also have $\eps \eps^* = \Id$: the two equalities then imply $\overline{\eps} = - \eps^*$, or $\eps = - \eps^t$, as wanted.

Notice that, according to \cite[Theorem 7]{Hua44}, the matrix $\eps$, being unitary and skew-symmetric, can be put in the form
\begin{equation} \label{eqn:eps=J}
\begin{pmatrix} 0 & 1 \\ -1 & 0 \end{pmatrix} \oplus \cdots \oplus \begin{pmatrix} 0 & 1 \\ -1 & 0 \end{pmatrix}
\end{equation}
in a suitable orthonormal basis. Thus, there is no loss of generality in assuming that $\eps$ is already of this form.
\end{remark}

We are now in position to state our goal: we tackle the following

\begin{problem} \label{pbl:Bloch}
Given a family of projectors $\set{P(\kk)}_{\kk \in \R^d}$ satisfying Assumption \ref{assum:proj} above, \emph{construct} (if possible) a \emph{continuous and periodic} (respectively \emph{symmetric}) Bloch frame for $\set{P(\kk)}_{\kk \in \R^d}$.
\end{problem}

We stress that here we seek for a \emph{constructive argument} that exhibits the required Bloch frames explicitly. Existence results, especially concerning \emph{periodic} frames in $d \le 3$, already exist in the literature \cite{Panati07, MonacoPanati15}. Notice that the existence of a continuous, periodic \emph{and TRS} Bloch frame is in general \emph{topologically obstructed} \cite{FiorenzaMonacoPanati16}, depending on the dimension $d$. In particular, in $d=2$ such obstruction can be encoded in a \emph{$\Z_2$-valued topological invariant} $\delta \in \Z_2$. We will come back to this point later on.

\subsection{Main results}

The intent of this paper is twofold. The first goal is to \emph{perform an explicit construction} of Bloch frames which are periodic and, possibly, also time-reversal symmetric, at least when topological obstructions vanish. The second goal is to shed some light and draw some connections between several proposals present in the literature on the $\Z_2$ invariants of TRS topological insulators.

The first set of results contained in this paper is summarized in the following statement.

\begin{theorem} \label{thm:MainResults_BB}
Let $\set{P(\kk)}_{\kk \in \R^d}$ satisfy Assumption \ref{assum:proj}.
\begin{enumerate}
 \item \label{item:Result_d=1} Let $d=1$. Then a continuous symmetric Bloch frame for $\set{P(k)}_{k \in \R}$ exists and can be constructed.
 \item \label{item:Result_d=2} Let $d=2$. 
 \begin{enumerate}[label=(\alph*),ref=(\alph*)]
  \item \label{item:Result_d=2_I=0} Assume that $\mathcal{I} = 0 \in \Z_2$, where $\mathcal{I}$ is defined in \eqref{eqn:rueda}. Then a continuous symmetric Bloch frame for $\set{P(\kk)}_{\kk \in \R^2}$ exists and can be constructed.
  \item \label{item:Result_d=2_beta} Conversely, if a continuous symmetric Bloch frame exists, then $\mathcal{I} = 0 \in \Z_2$.
 \end{enumerate}
 \item \label{item:Result_d=2_noTRS} Let $d=2$. Then a continuous \emph{periodic} Bloch frame for $\set{P(\kk)}_{\kk \in \R^2}$ \emph{always} exists and can be constructed.
\end{enumerate}
\end{theorem}

Part \ref{item:Result_d=1} of Theorem \ref{thm:MainResults_BB} is proved in Section \ref{sec:d=1}. Parts \ref{item:Result_d=2} and \ref{item:Result_d=2_noTRS} of Theorem \ref{thm:MainResults_BB} are proved in Section \ref{sec:d=2}. As was already noted and commented in the Introduction, the existence results from this statements, as far as symmetric frames are concerned, are not new \cite{Panati07,MonacoPanati15}, and recently also constructive proofs have been provided \cite{FiorenzaMonacoPanati16}. The methods employed in our proof, however, are different. Moreover, the paper contains the first constructive proof of continuous and \emph{periodic} Bloch frames in $d=2$ in presence of \emph{fermionic} TRS (for the bosonic setting, see \cite{FiorenzaMonacoPanati16_B, CorneanHerbstNenciu15, CancesLevittPanatiStoltz16}). As was argued in the Introduction, this translates in real space to the construction of (TRS) \emph{exponentially localized Wannier functions} for TRS topological insulators (see Theorem \ref{thm:Wannier} below).

The statement of Theorem \ref{thm:MainResults_BB}(\ref{item:Result_d=2}) involves the $\Z_2$ index $\mathcal{I}$, which was defined in \cite{GrafPorta13} in the context of tight-binding Hamiltonians modelling periodic quantum systems in presence of fermionic time-reversal symmetry. This index is associated to families of unitary matrices $\set{\alpha(k)}_{k \in \R}$ which are continuous, $\Z$-periodic and \emph{time-reversal symmetric}, in the sense that $\eps \, \alpha(k) = \alpha(-k)^t \, \eps$ (compare Assumption \ref{assum:alpha}). We sketched in the Introduction how such families of matrices appear in the present context: more details can be found in Section \ref{sec:d=d}. We prove in Section \ref{sec:GrafPorta} several alternative formulations of the original definition (compare \eqref{eqn:rueda}) of $\mathcal{I}$ from \cite{GrafPorta13}, which are more suited to our construction provided in Section \ref{sec:d=2}. With these we are also able to prove that $\mathcal{I} \in \Z_2$ gives a \emph{complete homotopy invariant} for continuous families of unitary matrices which are $\Z$-periodic and TRS, as in the following statement.

\begin{theorem} \label{thm:MainResults_GPhomotopy}
Two continuous, $\Z$-periodic, and TRS families of unitary matrices are homotopic (through a continuous homotopy of families with the same properties) if and only if their $\Z_2$ indices $\mathcal{I}$ coincide. In particular, one such family $\alpha$ is null-homotopic, in the sense specified above, if and only if $\mathcal{I}(\alpha) = 0 \in \Z_2$.
\end{theorem}
Theorem \ref{thm:MainResults_GPhomotopy} is proved in Section \ref{sec:Itopology}. Notice that \emph{all} continuous, periodic and TRS families of unitary matrices are null-homotopic, but the homotopy will break TRS if $\mathcal{I} = 1 \in \Z_2$ (compare the discussion at the beginning of Section \ref{sec:GrafPorta}). The content of the above statement was already observed in \cite[proof of Prop.~5]{AvilaSchulz-BaldesVillegas-Blas12}, where only a sketch of the proof was provided: our independent argument produces also an explicit homotopy between two families with the same index.

Coming back to Problem \ref{pbl:Bloch}, our proof then shows how the condition $\mathcal{I} = 0 \in \Z_2$ characterizes $2$-dimensional families of projectors which admit a continuous and symmetric Bloch frame. A number of other $\Z_2$ invariants have been proposed in the literature, as was already discussed in the Introduction. After the seminal work of Fu and Kane \cite{FuKane06}, a manifestly gauge-invariant formulation of the index was proposed by Prodan \cite{Prodan11_PRB}. More recently, using an obstruction-theoretical approach similar to the one employed in the present work, a true \emph{topological invariant} for $2$-dimensional families of projectors as in Assumption \ref{assum:proj} was defined in \cite{FiorenzaMonacoPanati16} by Fiorenza, Monaco and Panati. It is natural to ask how these different formulations are related to each other. This question is answered by the following

\begin{theorem} \label{thm:MainResults_Z2}
For a $2$-dimensional family of projectors satisfying Assumption \ref{assum:proj}, the Fu--Kane index $\Delta \in \Z_2$, the Fiorenza--Monaco--Panati invariant $\delta \in \Z_2$, the Graf--Porta index $\mathcal{I} \in \Z_2$ and the Prodan invariant $\xi \in \Z_2$ agree:
\[ \Delta = \delta = \mathcal{I} = \xi \quad \in \Z_2. \]
\end{theorem}

Each of the equalities in the above statement will be derived in Section \ref{sec:Z2} (compare Theorems~\ref{thm:GP_vs_FMP}, \ref{thm:deltaGeom} and \ref{thm:GP=Prodan}).

The first equality in the statement of Theorem \ref{thm:MainResults_Z2} was proved in \cite[Thm.~5]{FiorenzaMonacoPanati16}; besides, the relation $\Delta = \mathcal{I} \in \Z_2$ was already shown in \cite[Prop.~7.6]{GrafPorta13}. Our strategy is different. We first give a direct proof of the equality between the Graf--Porta and the Fiorenza--Monaco--Panati invariants (Theorem~\ref{thm:GP_vs_FMP}). In view of this equality, we are able to prove independently the relation between the invariant $\delta$ and the Fu--Kane index $\Delta$, in its formulation given by \cite[Eqn.~(A8)]{FuKane06}, making moreover our proof of Theorem~\ref{thm:MainResults_Z2} self-contained. As an interesting byproduct, which we also list among the main results of the paper, we obtain a \emph{geometric expression} for the topological invariant $\delta \in \Z_2$ which contains only the Berry connection and the Berry curvature associated to the family of projectors (compare Theorem~\ref{thm:deltaGeom}).

\begin{theorem} \label{thm:MainResults_delta}
Let $\set{P(\kk)}_{\kk \in \R^2}$ be a $2$-dimensional family of projectors as in Assumption \ref{assum:proj}. Then the associated $\Z_2$ invariant $\delta$ can be computed as
\[ \delta = \frac{1}{2\pi} \iint_{\B\sub{half}} \mathcal{F} - \frac{1}{2\pi} \left( \oint_{\Gamma_{1/2}} \A - \oint_{\Gamma_0} \A \right) \bmod 2, \]
where:
\begin{itemize}
\item $\B\sub{half} := \set{(k_1, k_2) \in \R^2: 0 \le k_1 \le 1, \: 0 \le k_2 \le 1/2}$, $\Gamma_0 = \set{(k_1,k_2) \in \B\sub{half} : k_2=0}$, $\Gamma_{1/2} = \set{(k_1,k_2) \in \B\sub{half} : k_2=1/2}$ (positively oriented with respect to $k_1$);
\item $\mathcal{F}$ is the \emph{Berry curvature $2$-form} \eqref{eqn:BerryCurvature} associated to $\set{P(\kk)}_{\kk \in \R^2}$;
\item $\A$ is the \emph{Berry connection $1$-form} \eqref{eqn:ABerry} computed with respect to a continuous and \emph{symmetric} Bloch frame.
\end{itemize}
\end{theorem}

In particular, our derivation of the above formula employs directly the family of projections, without any reference to the so-called sewing matrix and the associated Pfaffian formulation of the Fu--Kane index (compare Section~\ref{sec:TRIM}).

\begin{remark}[Analytic Bloch frames] \label{rmk:AnalyticBB}
Since the family of projectors $\set{P(\kk)}_{\kk \in \R^d}$ satisfies the analyticity assumption \ref{item:smooth}, one may ask whether \emph{real analytic} Bloch frames can be constructed for $\set{P(\kk)}_{\kk \in \R^d}$. Arguing as in \cite[Lemma 2.3]{CorneanHerbstNenciu15}, one can indeed show that whenever a \emph{continuous} Bloch frame exists, than it can be easily modified into a \emph{real analytic} one; the procedure preserves the symmetries (periodicity and TRS) whenever they are required. The result of Theorem \ref{thm:MainResults_BB} can therefore be strengthened and gives the existence of real analytic Bloch frames with the prescribed symmetry properties.
\end{remark}

\subsection{Relation with Wannier bases}

Following the discussion in the Introduction, we deduce the consequences that Theorem \ref{thm:MainResults_BB} (combined with the observation in Remark \ref{rmk:AnalyticBB}) has in the context of \emph{Wannier bases} for periodic Hamiltonians. For simplicity, we will consider the framework of continuous spin-$1/2$ Hamiltonians, but we stress that the result actually depends only on the \emph{symmetries} of the system (periodicity and fermionic TRS), and hence can be applied as well to discrete or tight-binding Hamiltonians, as for example the Fu--Kane Hamiltonian \cite{FuKane06} of quantum spin Hall insulators. We set first of all the following

\begin{definition}[Wannier basis] \label{dfn:Wannier}
Let $H$ be a $\Z^d$-periodic Hamiltonian on the spin-$1/2$ single-particle Hilbert space $L^2(\R^d) \otimes \C^2$, and assume that $H$ commutes with the time-reversal operator
\[ \Theta := (\Id \otimes \eu^{-\iu \pi \sigma_2/2}) C = - \iu \left( \Id \otimes \sigma_2 \right) C, \qquad \sigma_2 = \begin{pmatrix} 0 & -\iu \\ \iu & 0 \end{pmatrix}, \]
where $\sigma_2$ is the second Pauli matrix and $C$ denotes the complex-conjugation operator. Let $\Sigma \subset \sigma(H)$ be given by the union of a finite number $m$ of energy bands, which are well separated from the rest of the spectrum of $H$. Denote by $P_\Sigma$ the spectral projector of $H$ associated to $\Sigma$. Then a set of functions $\set{w_a}_{1 \le a \le m} \subset L^2(\R^d) \otimes \C^2$ is called a \emph{Wannier basis} if the following hold:
\begin{itemize}
 \item the set $\set{w_a(\cdot \, - \mathbf{n}) : 1 \le a \le m, \mathbf{n} \in \Z^d}$ gives an \emph{orthonormal basis} of $\Ran P_\Sigma$;
 \item the functions $w_a \in L^2(\R^d) \otimes \C^2$ are \emph{exponentially localized} for all $a \in \set{1, \ldots, m}$, namely there exists $b_0>0$ such that
 \[ \int_{\R^d} \eu^{b \, |\mathbf x|} \, \left| w_a(\mathbf x) \right|^2 \, \di \mathbf x < + \infty \quad \text{for all } b \in [0, b_0). \]
\end{itemize}

The Wannier basis is called \emph{time-reversal symmetric} (TRS) if moreover
\[ w_b(\mathbf x) = \sum_{a=1}^{m} \left[ \Theta w_a(\mathbf x) \right] \, \eps_{a b} \quad \text{for all } \mathbf x \in \R^d, \: b \in \set{1, \ldots, m}, \]
where $\eps = (\eps_{ab})_{1 \le a,b \le m}$ is a unitary and skew-symmetric marix.
\end{definition}

Notice that we include exponential localization as a defining property of a Wannier basis. Also observe that, choosing $\eps$ in the form \eqref{eqn:eps=J}, the TRS condition for a Wannier basis reads explicitly
\[ \begin{pmatrix} w_{2j-1}^\uparrow(\mathbf x) \\ w_{2j-1}^\downarrow(\mathbf x) \end{pmatrix} = \begin{pmatrix} \overline{w_{2j}^\downarrow(\mathbf x)} \\[5pt] \overline{w_{2j}^\uparrow(\mathbf x)} \end{pmatrix}, \quad w_a(\mathbf x) = \begin{pmatrix} w_a^\uparrow(\mathbf x) \\ w_a^\downarrow(\mathbf x) \end{pmatrix} \in \C^2, \quad j \in \set{1, \ldots, m/2}. \]

The strenghtened version of Theorem \ref{thm:MainResults_BB} argued in Remark \ref{rmk:AnalyticBB} now implies the following

\begin{theorem} \label{thm:Wannier}
Let $H$ be an operator as in Definition \ref{dfn:Wannier}, and let $\Sigma$ be an isolated set of $m$ energy bands.
\begin{enumerate}
 \item \label{item:Result_d=1} Let $d=1$. Then a TRS Wannier basis for $\Sigma$ exists and can be constructed.
 \item \label{item:Result_d=2} Let $d=2$. Then a Wannier basis for $\Sigma$ exists and can be constructed. If moreover $\mathcal{I} = 0 \in \Z_2$, where $\mathcal{I}$ is defined in \eqref{eqn:rueda}, then this Wannier basis can be constructed so that it is also TRS.
\end{enumerate}
\end{theorem}
\begin{proof}
The Bloch-Floquet-Zak transform
\[ \U \colon L^2(\R^d) \otimes \C^2 \to \int_{\T^d} \Hf \, \di \kk, \quad \Hf := L^2(\T^d) \otimes \C^2, \]
fibers the spectral projector $P_\Sigma$ of $H$ associated to the spectral subset $\Sigma \subset \sigma(H)$ along the crystal momentum $\kk \in \R^d$. Moreover, it is easily verified that $\U \, \Theta \, \U^{-1}$ acts on the space of fixed crystal momentum $\kk$ as
\[ (\U \, \Theta \, \U^{-1} \Xi)(\kk, \mathbf x) = (\theta \Xi)(-\kk, \mathbf x), \quad \Xi(\kk, \cdot) \in \Hf, \]
where $\theta \colon \Hf \to \Hf$ is the antiunitary operator given by
\[ \theta := (\Id_{\Hf} \otimes \eu^{-\iu \pi \sigma_2/2}) C = - \iu \left( \Id_{\Hf} \otimes \sigma_2 \right) C. \]
It follows then that the family of fiber projectors $\set{P_\Sigma(\kk)}_{\kk \in \R^d}$ is unitarily equivalent to a family of projectors $\set{P(k)}_{\kk \in \R^d}$ satisfying Assumption \ref{assum:proj} (compare \cite[Sec.~2.1]{CorneanHerbstNenciu15}).

Applying Theorem \ref{thm:MainResults_BB} and Remark \ref{rmk:AnalyticBB} to the latter, we obtain real analytic periodic Bloch frames, which may be also TRS depending on the dimension and on the vanishing of the topological obstruction $\mathcal{I}$. When mapped back to $\Ran P_\Sigma(\kk)$ by the above-mentioned unitary transformation, and then to $\Ran P_\Sigma \subset L^2(\R^d) \otimes \C^2$ via the Bloch-Floquet-Zak antitransform, these produce the desired Wannier bases. Notice indeed that the orthonormality of the Wannier basis follows from the orthonormality of the Bloch frame, and its exponential localization from the real analyticity in $\kk$ of the Bloch frame. As far as the TRS condition is concerned, if $\set{\Xi_a(\kk)}_{1 \le a \le m}$ is an orthonormal basis in $\Ran P_\Sigma(\kk)$ satisfying
\[ \Xi_b(-\kk) = \sum_{a=1}^{m} \left[\theta \Xi_a(\kk)\right] \, \eps_{ab}, \quad \kk \in \R^d, \: b \in \set{1, \ldots, m}, \]
then the associated Wannier functions
\[ w_b(\mathbf x) := \left(\U^{-1} \Xi_b\right)(\mathbf x) = \frac{1}{(2 \pi)^{d/2}} \int_{\R^d} \eu^{2 \pi \iu \kk \cdot \mathbf x} \, \Xi_b(\kk, \mathbf x) \, \di \kk \]
satisfies
\begin{align*}
w_b(\mathbf x) & = \frac{1}{(2 \pi)^{d/2}} \int_{\R^d} \eu^{2 \pi \iu \kk \cdot \mathbf x} \left( \sum_{a=1}^{m} \left[\theta \Xi_a(-\kk, \mathbf x)\right] \, \eps_{ab} \right) \di \kk \\
& = \sum_{a=1}^{m} \left[ \Theta \left( \frac{1}{(2 \pi)^{d/2}} \int_{\R^d} \eu^{-2 \pi \iu \kk \cdot \mathbf x} \, \Xi_a(-\kk, \mathbf x) \, \di \kk \right) \right] \eps_{ab} \\
& = \sum_{a=1}^{m} \left[ \Theta w_a(\mathbf x) \right] \eps_{ab}.
\end{align*}
The Wannier basis is then also TRS, and this concludes the proof.
\end{proof}

\subsection{Right action of unitary matrices on frames} \label{sec:unitary}

Before starting to attack Problem~\ref{pbl:Bloch}, we introduce some further notation. Let $\Fr(m, \Hi)$ denote the set of \emph{$m$-frames}, namely $m$-tuples of orthonormal vectors in $\Hi$. If $\Xi = \set{\Xi_1, \ldots, \Xi_m}$ is an $m$-frame, then we can obtain a new frame in $\Fr(m,\Hi)$ by means of a unitary matrix $M \in U(m)$, setting
\[ (\Xi \act M)_b := \sum_{a = 1}^{m} \Xi_a M_{ab}, \quad b \in \set{1, \ldots, m}. \]
This defines a free right action of $U(m)$ on $\Fr(m,\Hi)$.

Moreover, we can extend the action of the time-reversal operator $\theta \colon \Hi \to \Hi$ to $m$-frames, by setting
\[ (\theta \Xi)_a := \theta \Xi_a \qquad \text{for } \Xi = \set{\Xi_1, \ldots, \Xi_m} \in \Fr(m,\Hi). \]
Notice that, by the antilinearity of $\theta$, one has
\[ \theta (\Xi \act M) = (\theta \Xi) \act \overline{M}, \quad \text{for all } \Xi \in \Fr(m,\Hi), \: M \in U(m), \]
where $\overline{M}$ denotes the complex conjugate matrix.

We can recast properties \ref{item:F2} and \ref{item:F3} for a Bloch frame in this notation as
\begin{equation} \label{eqn:FramePeriodic}
\Xi(\kk + \mathbf{n}) = \Xi(\kk), \quad \text{for all } \kk \in \R^d, \:\mathbf{n} \in \Z^d,
\end{equation}
and
\begin{equation} \label{eqn:FrameTR}
\Xi(-\kk) = \theta \Xi(\kk) \act \eps, \quad \text{for all } \kk \in \R^d.
\end{equation}

Observe that also the action of a unitary operator $W \in \U(\Hi)$ can be extended component-wise to frames, setting
\[ (W \, \Xi)_a := W \, \Xi_a \qquad \text{for } \Xi = \set{\Xi_1, \ldots, \Xi_m} \in \Fr(m,\Hi). \]
This action commutes with the free right action of $U(m)$ defined above: $(W \Xi) \act M = W (\Xi \act M)$.

\subsection{Degree of a family of unitary matrices} \label{sec:DegreeUnitary}

We collect here some auxiliary results that will be used several times in the rest of the paper. These results concern the \emph{topological degree} of a family of unitary matrices defined on a circle $\T$, or equivalently of a periodic family of unitary matrices, when we identify a period with $\T$.

Recall \cite[Thm.~17.3.1]{DubrovinNovikovFomenko85} that the homotopy class of a continuous map $\varphi \colon \T \to \T$ of the circle onto itself identifies an element in the homotopy group $\pi_1(\T) \simeq \Z$. The integer associated to such homotopy class is called its \emph{degree} (or \emph{winding number}), denoted by $\deg([\varphi]) \in \Z$. Since any continuous map is homotopic to a smooth map, this integer can be computed via the integral Cauchy formula \cite[\S~13.4(b)]{DubrovinNovikovFomenko85}
\begin{equation} \label{eqn:degree}
\deg([\varphi]) = \frac{1}{2 \pi \iu} \, \oint_{\T} \varphi(z)^{-1} \, \partial_z \varphi(z) \, \di z,
\end{equation}
where $z$ is a variable running on the torus $\T$. Notice that, since $\deg$ is a group homomorphism, we have that $\deg([\overline{\varphi}]) = \deg([\varphi^{-1}]) = - \deg([\varphi])$.

Similarly, the homotopy class of a periodic map $\beta \colon \R \to U(m)$, $\beta(k) = \beta(k+1)$, selects an element in the homotopy group $\pi_1(U(m))$. It can be shown \cite[Ch.~8, Sec.~12]{Husemoller94} that the latter group is isomorphic to $\Z$, via the map $\pi_1(\det) \colon \pi_1(U(m)) \to \pi_1(U(1))$. This means that the homotopy class of a periodic map $\beta \colon \R \to U(m)$ is characterized by an integer, which is the degree of its determinant: $\deg([\det \beta]) \in \Z$.

We now give an alternative formula for this integer. We appeal to the following
\begin{lemma} \label{lemma:DegUnitary}
Let $\beta \colon \R \to U(m)$ be a smooth $\Z$-periodic map. Define
\[ B(k) := \det \beta(k), \quad D(k) := \tr \left( \beta(k)^* \beta'(k) \right), \]
where $\tr(\cdot)$ denotes the trace in $\C^m$ and the prime means derivative with respect to $k$. Then
\[ B(k)^{-1} \, B'(k) = D(k). \]
\end{lemma}
\begin{proof}
By unitarity of $\beta(k)$ we have that
\[ B(k+h) = B(k) \, \det \left( \beta(k)^* \, \beta(k+h) \right) = B(k) \, \det \left[ \Id + \beta(k)^* \left( \beta(k+h) - \beta(k) \right) \right]. \]
Since $\det (\Id+ H) = 1+ \tr(H) + \mathcal{O}(H^2)$ we get
\[ B(k+h) - B(k) = h \, B(k) \, D(k) + o(h). \]
This concludes the proof.
\end{proof}

It then follows that the degree of a continuous periodic map $\beta \colon \R \to U(m)$ can be also computed as
\begin{equation} \label{eqn:DegUn}
\deg([\det \beta]) = \frac{1}{2 \pi \iu} \oint_{\T} \tr \left( \beta(z)^* \partial_{z} \beta(z) \right) \, \di z.
\end{equation}

\goodbreak


\section{Symmetric Bloch frames: generalities} \label{sec:d=d}

We begin in this Section the analysis of Problem \ref{pbl:Bloch}, concerning the existence of continuous and periodic (or symmetric) Bloch frame for a $d$-dimensional family of projectors $\set{P(\kk)}_{\kk \in \R^d}$ satisfying Assumption \ref{assum:proj}. We take an approach which is suited to proceed inductively on the dimension $d$, namely to use an input frame in dimension $d-1$ to construct one in dimension~$d$. The main tool we will use to this aim is the \emph{parallel transport} associated to the family of projectors, which will be reviewed in a moment. After that, we will reduce Problem \ref{pbl:Bloch} on Bloch frames to an equivalent stament in terms of families of $m \times m$ unitary matrices. The latter will be then analyzed in the next Sections for $d \le 2$.

\subsection{Parallel transport} \label{sec:parallel}

As a starting point, we recall the definition of \emph{parallel transport} associated to a family of projectors $\set{P(\kk)}_{\kk \in \R^d}$ acting on an Hilbert space $\Hi$.

For $\kk = (k_1, \kk_\perp) \in \R^d$ (with $k_1 \in \R$ and $\kk_\perp = (k_2, \ldots, k_d) \in \R^{d-1}$), define
\begin{equation} \label{eqn:A_kperp(k1)}
A_{\kk_\perp}(k_1) := \iu \left[ \partial_{k_1} P(k_1, \kk_\perp), P(k_1, \kk_\perp) \right].
\end{equation}
Then $A_{\kk_\perp}(k_1)$ defines a self-adjoint operator on $\Hi$. For fixed $k_1' \in \R$, the solution to the problem
\begin{equation} \label{eqn:parallel_def}
\iu \, \partial_{k_1} T_{\kk_\perp}(k_1, k_1') = A_{\kk_\perp}(k_1) \, T_{\kk_\perp} (k_1, k_1'), \quad T_{\kk_\perp}(k_1', k_1') = \Id,
\end{equation}
defines a family of unitary operators, called the \emph{parallel transport unitaries}. The family satisfies the properties listed in the following result.

\begin{proposition} \label{prop:parallel}
Let $\set{P(\kk)}_{\kk \in \R^d}$ be a family of orthogonal projectors acting on an Hilbert space $\Hi$. Then the family of parallel transport unitaries $\set{T_{\kk_\perp}(k_1, k_1')}_{k_1, k_1' \in \R, \: \kk_\perp \in \R^{d-1}}$ defined in \eqref{eqn:parallel_def} satisfies the following properties:
\begin{enumerate}[label=$(\mathrm{T}_\arabic*)$,ref=$(\mathrm{T}_\arabic*)$]
 \item \label{item:T-smooth} if $\set{P(\kk)}_{\kk \in \R^d}$ satisfies \ref{item:smooth}, then for fixed $k_1' \in \R$ the map $\R^d \ni \kk = (k_1, \kk_\perp) \mapsto T_{\kk_\perp}(k_1, k_1') \in \U(\Hi)$ is real analytic;
 \item \label{item:T-periodic} if $\set{P(\kk)}_{\kk \in \R^d}$ satisfies \ref{item:periodic}, then for all $k_1, k_1' \in \R$ and $\kk_\perp \in \R^{d-1}$ 
 \[ T_{\kk_\perp}(k_1+1, k_1'+1) = T_{\kk_\perp}(k_1, k_1') \]
 and
 \[ T_{\kk_\perp + \mathbf{n}_\perp}(k_1, k_1') = T_{\kk_\perp}(k_1, k_1') \quad  \text{for } \mathbf{n}_\perp \in \Z^{d-1}; \]
 \item \label{item:T-TRS} if $\set{P(\kk)}_{\kk \in \R^d}$ satisfies \ref{item:TRS}, then for all $k_1, k_1' \in \R$ and $\kk_\perp \in \R^{d-1}$ 
 \[ T_{-\kk_\perp}(-k_1, -k_1') = \theta \, T_{\kk_\perp}(k_1, k_1') \, \theta^{-1}; \]
 \item \label{item:T-group} the group properties
 \[ T_{\kk_\perp}(k_1, k_1') \, T_{\kk_\perp}(k_1', k_1'') = T_{\kk_\perp}(k_1, k_1''), \quad T_{\kk_\perp}(k_1, k_1')^{-1} = T_{\kk_\perp}(k_1', k_1) \]
 hold for all $k_1, k_1', k_1'' \in \R$ and all $\kk_\perp \in \R^{d-1}$;
 \item \label{item:T-intertwine} the intertwining property
 \[ P(k_1, \kk_\perp) = T_{\kk_\perp}(k_1, k_1') \, P(k_1', \kk_\perp) \, T_{\kk_\perp}(k_1, k_1')^{-1} \]
 holds for all $k_1, k_1' \in \R$ and $\kk_\perp \in \R^{d-1}$ .
\end{enumerate}
\end{proposition}

Even though these results are relatively standard, a proof of all these properties can be found for example in \cite{FreundTeufel16} or in \cite[Sec.~2.6]{CorneanHerbstNenciu15}. Only a slight modification of the argument presented there is needed to prove \ref{item:T-TRS}, allowing for the fact that $\theta^2 = - \Id$ (contrary to the assumption $\theta^2 = \Id$ adopted in \cite{CorneanHerbstNenciu15}). Indeed, one can immediately check using \ref{item:TRS} that
\[ A_{-\kk_\perp}(-k_1) = \theta \, A_{\kk_\perp}(k_1) \, \theta^{-1}, \]
where $A_{\kk_\perp}(k_1)$ is defined in \eqref{eqn:A_kperp(k1)}. The above implies in particular that both sides of \ref{item:T-TRS} satisfy the same Cauchy problem.

\subsection{Matching matrices} \label{sec:matching}

We return now to the problem of finding continuous and symmetric Bloch frames. Let $\set{P(\kk)}_{\kk \in \R^d}$ be a family of projectors satisfying Assumption \ref{assum:proj}, and let $m$ denote its rank. The restriction $\set{P(0,\kk_\perp)}_{\kk_\perp \in \R^{d-1}}$ is then a $(d-1)$-dimensional family of projectors which also satisfies Assumption \ref{assum:proj}. Assume that we have constructed a continuous and periodic (respectively symmetric) Bloch frame $\Xi = \set{\Xi(0, \kk_\perp)}_{\kk_\perp \in \R^{d-1}}$ for such a family. Define now
\begin{equation} \label{eqn:Psi}
\Psi(k_1, \kk_\perp) := T_{\kk_\perp}(k_1,0) \, \Xi(0, \kk_\perp), \quad \kk = (k_1, \kk_\perp) \in \R^d.
\end{equation}
Upon this definition, $\Psi(\kk)$ gives an orthonormal basis in $\Ran P(\kk)$ by the intertwining property \ref{item:T-intertwine} and depends analytically on $\kk$ in view of \ref{item:T-smooth}; moreover, this frame is also periodic in $\kk_\perp$ by \ref{item:T-periodic}. Whenever $\Xi$ is TRS, then also $\Psi$ is, in view of \ref{item:T-TRS}. 

However, $\Psi(\kk)$ fails in general to satisfy periodicity in $k_1$. This can be seen as follows. Since $\Psi(k_1 + 1, \kk_\perp)$ and $\Psi(k_1, \kk_\perp)$ are frames in $\Ran P(k_1 + 1, \kk_\perp) = \Ran P(k_1, \kk_\perp)$ (by \ref{item:periodic}), they must differ by the action of a unitary matrix:
\begin{equation} \label{eqn:alpha}
\Psi(k_1 + 1, \kk_\perp) = \Psi(k_1, \kk_\perp) \act \alpha(\kk_\perp).
\end{equation}
The $m \times m$ matrix $\alpha(\kk_\perp)$ is indeed independent of $k_1$ as can be seen by computing its entries. First of all, notice that
\begin{align*}
\Psi(k_1 + 1, \kk_\perp) & = T_{\kk_\perp}(k_1 + 1,0) \, \Xi(0, \kk_\perp) = T_{\kk_\perp}(k_1+1, 1) \, T_{\kk_\perp}(1,0) \, \Xi(0, \kk_\perp) \\ 
& = T_{\kk_\perp}(k_1,0) \, T_{\kk_\perp}(1,0) \, \Xi(0, \kk_\perp)
\end{align*}
where we used \ref{item:T-group} in the second equality and \ref{item:T-periodic} in the last. By unitarity of $T_{\kk_\perp}(k_1,0)$, we obtain then
\begin{align*}
\alpha(\kk_\perp)_{ab} & = \scal{\Psi_a(k_1, \kk_\perp)}{\Psi_b(k_1+1,\kk_\perp)} \\
& = \scal{T_{\kk_\perp}(k_1,0) \, \Xi_a(0, \kk_\perp)}{T_{\kk_\perp}(k_1,0) \, T_{\kk_\perp}(1,0) \, \Xi_b(0, \kk_\perp)} \\
& = \scal{\Xi_a(0, \kk_\perp)}{T_{\kk_\perp}(1,0) \, \Xi_b(0, \kk_\perp)}.
\end{align*}
An equivalent formulation of the above relation, and a possible alternative definition of the matrix $\alpha$, is
\begin{equation} \label{eqn:alpha_vs_T}
T_{\kk_\perp}(1,0) \, \Xi(0, \kk_\perp) = \Xi(0, \kk_\perp) \act \alpha(\kk_\perp).
\end{equation}
The matrix $\alpha(\kk_\perp)$ then expresses how the input frame $\Xi(0, \kk_\perp)$ must be unitarily rotated to match its parallel-transported version $T_{\kk_\perp}(1,0) \, \Xi(0, \kk_\perp)$ along one full period in $k_1$. Hence, we call $\alpha(\kk_\perp)$ the \emph{matching matrix}. As was already noticed, it measures the failure of the Bloch frame $\Psi(k_1, \kk_\perp)$ to be periodic with respect to $k_1$, while enjoying all the other properties (being real analytic, periodic in $\kk_\perp$, and, possibly, TRS).

We list the main properties which are satisfied by the family of the matching matrices $\alpha$ in the next Proposition.

\begin{proposition} \label{prop:alpha}
The $m \times m$ matrix $\alpha(\kk_\perp)$, defined by the relation \eqref{eqn:alpha}, is unitary. The family of matrices $\alpha = \set{\alpha(\kk_\perp)}_{\kk_\perp \in \R^{d-1}}$ depends analytically on $\kk_\perp$ and is $\Z^{d-1}$-periodic. If $\Xi$ is also TRS, then $\alpha$ satisfies 
\begin{equation} \label{eqn:CS'}
\eps \, \alpha(\kk_\perp) = \alpha(-\kk_\perp)^t \, \eps.
\end{equation}
\end{proposition}

In the following, we will say that a family $\alpha$ of $m \times m$ unitary matrices is \emph{time-reversal symmetric} (or TRS for short) if it satisfies \eqref{eqn:CS'}. 

\begin{proof}[Proof of Proposition \ref{prop:alpha}]
Most of the properties follow at once from the properties of the parallel transport unitaries $T_{\kk_\perp}(1,0)$ via \eqref{eqn:alpha_vs_T}. We prove only \eqref{eqn:CS'}. First, observe that \eqref{eqn:alpha_vs_T} implies
\[ T_{-\kk_\perp}(1,0)^* \, \Xi(0,-\kk_\perp) = \Xi(0,-\kk_\perp) \act \alpha(-\kk_\perp)^* \]
as can be seen at once from the equivalent formulation in terms of the matrix entries of $\alpha(-\kk_\perp)$. With this, we can compute
\begin{align*}
\theta \Xi(0,-\kk_\perp) \act (\eps \, \alpha(\kk_\perp)) & = \Xi(0,\kk_\perp) \act \alpha(\kk_\perp) \\
 & = T_{\kk_\perp}(1,0) \, \Xi(0,\kk_\perp) = T_{\kk_\perp}(1,0) \, \theta \Xi(0,-\kk_\perp) \act \eps \\
 & = \theta T_{-\kk_\perp}(-1,0) \, \Xi(0,-\kk_\perp) \act \eps = \theta T_{-\kk_\perp}(0,1) \, \Xi(0,-\kk_\perp) \act \eps \\ 
 & = \theta T_{-\kk_\perp}(1,0)^* \, \Xi(0,-\kk_\perp) \act \eps = \theta \left( \Xi(0,-\kk_\perp) \act \alpha(-\kk_\perp)^* \right) \act \eps \\
 & = \theta \Xi(0,-\kk_\perp) \act \left( \alpha(-\kk_\perp)^t \, \eps \right).
\end{align*}
By the freeness of the $U(m)$-action on $m$-frames, we deduce that
\[ \eps \, \alpha(\kk_\perp) = \alpha(-\kk_\perp)^t \, \eps \]
which is exactly \eqref{eqn:CS'}.
\end{proof}

\subsection{Enforcing periodicity in $k_1$}

We now come back to Problem \ref{pbl:Bloch} for the projectors $\set{P(\kk)}_{\kk \in \R^d}$. The frame $\Psi(\kk)$ defined in \eqref{eqn:Psi} needs to be modified in order to enforce periodicity in $k_1$. Clearly, if $\alpha(\kk_\perp)$ were the identity matrix, then $\Psi(\kk)$ would be already $\Z$-periodic also in $k_1$. The strategy is then to look for a family of unitary matrices which ``rotates'' $\alpha$ to the identity, while ensuring that the properties of the frame $\Psi$ (namely continuity, periodicity in $\kk_\perp$ and, if required, TRS) are preserved. More formally, we have the following result.

\begin{proposition} \label{prop:rotate}
Let $\Psi$ be the Bloch frame defined in \eqref{eqn:Psi}, and let $\alpha$ be as in \eqref{eqn:alpha}. Assume that there exists a continuous family $\set{\beta(\kk) = \beta(k_1, \kk_\perp)}_{\kk = (k_1, \kk_\perp) \in \R^d}$ of $m \times m$ unitary matrices, which is $\Z^{d-1}$-periodic with respect to $\kk_\perp$ and satisfies
\begin{equation} \label{eqn:alpha_vs_beta}
\beta(k_1, \kk_\perp)^{-1} \, \alpha(\kk_\perp) \, \beta(k_1 + 1, \kk_\perp) = \Id.
\end{equation}
Then 
\begin{equation} \label{eqn:beta}
\Xi(\kk) := \Psi(\kk) \act \beta(\kk)
\end{equation}
defines a continuous and $\Z^d$-periodic Bloch frame for $\set{P(\kk)}_{k \in \R^d}$. If moreover $\Psi$ is TRS and $\set{\beta(\kk)}_{\kk \in \R^d}$ satisfies also
\begin{equation} \label{eqn:beta-TRS}
\beta(-\kk) = \eps^{-1} \, \overline{\beta(\kk)} \, \eps,
\end{equation}
then $\Xi$ is also TRS.

Conversely, if $\Xi$ is a continuous and $\Z^d$-periodic Bloch frame for $\set{P(\kk)}_{k \in \R^d}$, then the family of unitary matrices $\set{\beta(\kk)}_{\kk \in \R^d}$ defined by \eqref{eqn:beta}, \ie
\[ \beta(\kk)_{ab} := \scal{\Psi_a(\kk)}{\Xi_b(\kk)}, \quad a, b \in \set{1, \ldots, m}, \]
is continuous, $\Z^{d-1}$-periodic with respect to $\kk_\perp$, with $\kk = (k_1, \kk_\perp)$, and satisfies \eqref{eqn:alpha_vs_beta}. If $\Xi$ and $\Psi$ are both also TRS, then $\set{\beta(\kk)}_{\kk \in \R^d}$ satisfies also \eqref{eqn:beta-TRS}.
\end{proposition}
\begin{proof}
Unitarity of $\beta(\kk)$ guarantees that $\Xi(\kk)$, defined by \eqref{eqn:beta}, is still an orthonormal basis in $\Ran P(\kk)$. Moreover, continuity and periodicity of $\Xi(\kk) = \Xi(k_1, \kk_\perp)$ with respect to $\kk_\perp$ are clear, so we only need to prove the other two symmetry properties. We start from periodicity in $k_1$. In view of \eqref{eqn:alpha_vs_beta} and \eqref{eqn:alpha} we have
\begin{align*}
\Xi(k_1+1,\kk_\perp) & = \Psi(k_1+1,\kk_\perp) \act \beta(k_1+1,\kk_\perp) = \Psi(k_1,\kk_\perp) \act (\alpha(\kk_\perp) \, \beta(k_1+1,\kk_\perp) ) \\
& = \Psi(k_1,\kk_\perp) \act \beta(k_1, \kk_\perp) = \Xi(k_1,\kk_\perp)
\end{align*}
so that $\Xi(\kk)$ is periodic with respect to $k_1$. If \eqref{eqn:beta-TRS} holds, then it yields TRS for $\Xi$, since then 
\begin{align*}
\Xi(-\kk) & = \Psi(-\kk) \act \beta(-\kk) = \theta \Psi(\kk) \act (\eps \, \beta(-\kk) ) \\
& = \theta \Psi(\kk) \act (\overline{\beta(\kk)} \, \eps) = \theta ( \Psi(\kk) \act \beta(\kk) ) \act \eps \\
& = \theta \Xi(\kk) \act \eps.
\end{align*}

By going through the same chain of equalities as above, also the converse statement can be proved.
\end{proof}

\begin{remark} \label{rmk:beta0}
Notice that, if $\set{\beta(\kk)}_{\kk \in \R^d}$ as in the statement of Proposition \ref{prop:rotate} exists, then 
\[ \widetilde{\beta}(k_1, \kk_\perp) := \beta(0, \kk_\perp)^{-1} \, \beta(k_1, \kk_\perp) \]
satisfies its same properties, namely it is continuous, $\Z^{d-1}$-periodic in $\kk_\perp$, satisfies \eqref{eqn:alpha_vs_beta} and, possibly, \eqref{eqn:beta-TRS}. Thus, without loss of generality we may always assume that $\beta(0, \kk_\perp) \equiv \Id$ for all $\kk_\perp \in \R^{d-1}$. At the level of frames, this means that $\Xi(k_1, \kk_\perp)$ as in \eqref{eqn:beta} coincides with the ``input'' frame $\Xi(0, \kk_\perp)$ at $k_1 = 0$.
\end{remark}

\begin{remark}
If a family $\beta$ of unitary matrices as in the statement of Proposition \ref{prop:rotate} exists, then it can be used to continuously deform the family $\alpha$ of matching matrices to the identity. Indeed, setting
\[ \alpha_s(\kk_\perp) := \beta(-s/2, \kk_\perp) \, \beta(s/2, \kk_\perp)^{-1}, \quad \kk_\perp \in \R^{d-1}, \: s \in [0,1], \]
defines a continuous homotopy of $\alpha$ (at $s=1$, compare \ref{eqn:alpha_vs_beta}) to $\Id$ (at $s=0$), which is $\Z^{d-1}$-periodic in $\kk_\perp$ and TRS whenever $\beta$ satisfies \ref{eqn:beta-TRS}.
\end{remark}

Problem \ref{pbl:Bloch} on finding a continuous and $\Z^d$-periodic (respectively symmetric) Bloch frame is thus reduced by Proposition \ref{prop:rotate} to the following one.

\begin{problem} \label{pbl:alpha}
Given a family $\set{\alpha(\kk_\perp)}_{\kk_\perp \in \R^{d-1}}$ of unitary matrices which is real analytic and $\Z^{d-1}$-periodic in $\kk_\perp$, as well as possibly TRS, \emph{construct} (if possible) a \emph{continuous} family $\set{\beta(\kk) = \beta(k_1, \kk_\perp)}_{\kk = (k_1, \kk_\perp) \in \R^d}$ which is \emph{$\Z^{d-1}$-periodic} in $\kk_\perp$ and satisfies \eqref{eqn:alpha_vs_beta}, as well as possibly \eqref{eqn:beta-TRS}.
\end{problem}

Our Theorem \ref{thm:MainResults_BB} can then be reformulated purely in terms of TRS families of unitary matrices as in the following statement, which is then of independent interest.

\begin{theorem} \label{thm:MainResults_Beta}
Let $\set{\alpha(\kk_\perp)}_{\kk_\perp \in \R^{d-1}}$ be a continuous, $\Z^{d-1}$-periodic, and TRS family of $m \times m$ unitary matrices.
\begin{enumerate}
 \item \label{item:BetaResult_d=1} Let $d=1$. Then a continuous family $\set{\beta(k)}_{k \in \R}$ of unitary matrices satisfying \eqref{eqn:alpha_vs_beta} and \eqref{eqn:beta-TRS} exists and can be constructed.
 \item \label{item:BetaResult_d=2} Let $d=2$. 
 \begin{enumerate}[label=(\alph*),ref=(\alph*)]
  \item \label{item:BetaResult_d=2_I=0} Assume that $\mathcal{I}(\alpha) = 0 \in \Z_2$, where $\mathcal{I}(\alpha)$ is defined in \eqref{eqn:rueda}. Then a continuous family $\set{\beta(\kk) = \beta(k_1, k_2)}_{\kk = (k_1, k_2) \in \R^2}$ of unitary matrices which is $\Z$-periodic in $k_2$ and satisfies \eqref{eqn:alpha_vs_beta} and \eqref{eqn:beta-TRS} exists and can be constructed.
  \item \label{item:BetaResult_d=2_beta} Conversely, if a family $\beta$ as in \ref{item:BetaResult_d=2_I=0} exists, then $\mathcal{I}(\alpha) = 0 \in \Z_2$.
 \end{enumerate}
 \item \label{item:BetaResult_d=2_noTRS} Let $d=2$. Then a continuous family $\set{\beta(\kk) = \beta(k_1, k_2)}_{\kk = (k_1, k_2) \in \R^2}$ of unitary matrices which is $\Z$-periodic in $k_2$ and satisfies \eqref{eqn:alpha_vs_beta} \emph{always} exists and can be constructed.
\end{enumerate}
\end{theorem}

In the next two Sections, we will prove the above result and perform the construction of the family $\beta$ for $d \le 2$, and thus produce continuous and periodic (respectively symmetric) Bloch frames. The construction for $d=3$ will be investigated elsewhere \cite{CorneanMonaco}. Notice how, when $d=2$, the construction of a family $\beta$ satisfying also \eqref{eqn:beta-TRS} requires a further $\Z_2$ index $\mathcal{I}(\alpha) \in \Z_2$ to vanish, corresponding to the topological obstruction to the existence of a continuous symmetric Bloch frame \cite{GrafPorta13,FiorenzaMonacoPanati16}, as stated in Theorem \ref{thm:MainResults_BB}(\ref{item:Result_d=2}). The connection between these two (topological) obstructions will be investigated more thoroughly in Section \ref{sec:Z2}.

\subsection{Construction of the matrix $\beta$}

Coming back for a moment to the general case of arbitrary dimension $d$, we illustrate here a series of situations in which Problem \ref{pbl:alpha} admits a positive answer. The strategy of the proof of Theorem \ref{thm:MainResults_Beta} will then be to reduce ourselves to these ``favourable'' situations.

\subsubsection{When $\alpha$ has a logarithm}

A solution to Problem \ref{pbl:alpha} can be produced in any dimension, provided the family of matrices $\alpha$ admits a ``good logarithm'', as detailed in the next result.

\begin{proposition} \label{prop:GoodLog}
Let $\alpha = \set{\alpha(\kk_\perp)}_{\kk_\perp \in \R^{d-1}}$ be as in the statement of Problem \ref{pbl:alpha}. Assume that $\alpha(\kk_\perp)$ is in the form 
\begin{equation} \label{eqn:alphalog}
\alpha(\kk_\perp) = \alpha\sub{log}(\kk_\perp) := \eu^{\iu \, h(\kk_\perp)},
\end{equation}
where $h = \set{h(\kk_\perp)}_{\kk_\perp \in \R^{d-1}}$ is a continuous and $\Z^{d-1}$-periodic family of self-adjoint matrices:
\begin{equation} \label{eqn:h_properties}
h(\kk_\perp) = h(\kk_\perp)^* = h(\kk_\perp + \mathbf{n}_\perp), \quad \kk_\perp \in \R^{d-1}, \: \mathbf{n}_\perp \in \Z^{d-1}.
\end{equation}
Then there exists a continuous family $\beta\sub{log} = \set{\beta\sub{log}(\kk) = \beta\sub{log}(k_1, \kk_\perp)}_{\kk = (k_1, \kk_\perp) \in \R^d}$ of unitary matrices which is periodic with respect to $\kk_\perp$ and satisfies \eqref{eqn:alpha_vs_beta}. Moreover, $\beta\sub{log}$ can be explicitly constructed.

If moreover $h$ can be chosen such that 
\begin{equation} \label{eqn:hTRS}
\eps \, h(\kk_\perp) = h(-\kk_\perp)^t \, \eps, \quad \kk_\perp \in \R^{d-1},
\end{equation}
then $\beta\sub{log}$ can be constructed so as to satisfy also \eqref{eqn:beta-TRS}.
\end{proposition}
\begin{proof}
The equality in \eqref{eqn:alphalog} is clearly equivalent to
\begin{equation} \label{eqn:1st_log}
\eu^{-\iu \, h(\kk_\perp) /2} \, \alpha\sub{log}(\kk_\perp) \, \eu^{-\iu \, h(\kk_\perp) /2} = \Id.
\end{equation}
A family of matrices $\set{\beta\sub{log}(\kk) = \beta\sub{log}(k_1, \kk_\perp)}_{\kk = (k_1, \kk_\perp) \in \R^d}$ which is periodic in $\kk_\perp$ and satisfies \eqref{eqn:alpha_vs_beta} is then obtained by setting
\[ \beta\sub{log}(k_1, \kk_\perp) := \eu^{- \iu \, k_1 \, h(\kk_\perp)}, \quad k_1 \in \left[-1/2, 1/2\right], \: \kk_\perp \in \R^{d-1}, \]
and extending this definition to all $k_1 \in \R$ by using repeatedly
\begin{equation} \label{eqn:ext_beta}
\beta\sub{log}(k_1 + 1, \kk_\perp) := \alpha(\kk_\perp)^{-1} \, \beta\sub{log}(k_1,\kk_\perp).
\end{equation}
Continuity in $k_1$ of this extension follows directly from \eqref{eqn:1st_log}; moreover, $\beta\sub{log}$ satisfies \eqref{eqn:alpha_vs_beta} and periodicity in $\kk_\perp$ by construction. 

If $h$ satisfies also \eqref{eqn:hTRS}, then $\beta\sub{log}$ enjoys the property \eqref{eqn:beta-TRS}, since in view of the self-adjointness of $h(\kk_\perp)$
\begin{equation} \label{eqn:betalogTRS} 
\eps^{-1} \, \overline{\beta(\kk)} \, \eps = \eps^{-1} \, \overline{\eu^{- \iu \, k_1 \, h(\kk_\perp)}} \, \eps = \eu^{\iu \, k_1 \, [ \eps^{-1} \, \overline{h(\kk_\perp)} \, \eps]} = \eu^{- \iu \, (-k_1) \, h(-\kk_\perp)} = \beta(-\kk).
\end{equation}
\end{proof}

We will see in Section \ref{sec:d=1} that the above situation is exactly the one satisfied by the family of matching matrices associated of a $1$-dimensional family of projectors (compare Section \ref{sec:matching1d}). Thus Proposition \ref{prop:GoodLog} will prove Theorem \ref{thm:MainResults_Beta}(\ref{item:BetaResult_d=1}).

A useful way to actually produce a ``good logarithm'' as in the statement of the above Proposition is via the \emph{Cayley transform}, as we will now detail. This relates the existence of a logarithm with some spectral properties of the family of unitary matrices.

\begin{proposition}[Cayley transform] \label{prop:Cayley}
Let $\set{\widetilde{\alpha}(\kk_\perp)}_{\kk_\perp \in \R^{d-1}}$ be a family of unitary matrices which is continuous and $\Z^{d-1}$-periodic. Assume that $-1$ lies in the resolvent set of $\widetilde{\alpha}(\kk_\perp)$ for all $\kk_\perp \in \R^{d-1}$. Then $\widetilde{\alpha}(\kk_\perp)$ is the form \eqref{eqn:alphalog}, namely there exists a family $\set{\widetilde{h}(\kk_\perp)}_{\kk_\perp \in \R^{d-1}}$ of self-adjoint matrices which is continuous, $\Z^{d-1}$-periodic and such that for all $\kk_\perp \in \R^{d-1}$ and $\mathbf{n}_\perp \in \Z^{d-1}$
\[ \widetilde{\alpha}(\kk_\perp) = \eu^{\iu \, \widetilde{h}(\kk_\perp)}, \quad \widetilde{h}(\kk_\perp) = \widetilde{h}(\kk_\perp)^* = \widetilde{h}(\kk_\perp+ \mathbf{n}_\perp). \]

If moreover $\widetilde{\alpha}(\kk_\perp)$ is TRS, \ie it satisfies \eqref{eqn:CS'}, then the above family of self-adjoint matrices can be chosen so that
\[ \eps \, \widetilde{h}(\kk_\perp) = \widetilde{h}(-\kk_\perp)^t \, \eps. \]
\end{proposition}
\begin{proof}
We will employ the Cayley transform method (compare \cite[Sec.~2.7.3]{CorneanHerbstNenciu15}) to construct the logarithm $\widetilde{h}$. The matrix
\[ s(\kk_\perp) := \iu \, \left( \Id - \widetilde{\alpha}(\kk_\perp) \right) \, \left( \Id + \widetilde{\alpha}(\kk_\perp) \right)^{-1} \]
is self-adjoint, depends continuously on $\kk_\perp$ and is $\Z$-periodic; moreover, if $\widetilde{\alpha}$ satisfies \eqref{eqn:CS'} then
\[ \eps \, s(\kk_\perp) = s(-\kk_\perp)^t \, \eps. \]
One also immediately verifies that
\[ \widetilde{\alpha}(\kk_\perp) = \left( \Id + \iu \, s(\kk_\perp) \right) \, \left( \Id - \iu \, s(\kk_\perp) \right)^{-1}. \]

Let $\mathcal{C}$ be a closed, positively-oriented contour in the complex plane which encircles the real spectrum of $s(\kk_\perp)$ for all $\kk_\perp \in \R$. Let $\log(\cdot)$ denote the choice of the complex logarithm corresponding to the branch cut on the negative real semi-axis. Then
\[ \widetilde{h}(\kk_\perp) := \frac{1}{2 \pi} \oint_\mathcal{C} \log \left( \frac{1 + \iu \, z}{1 - \iu \, z} \right) \left( s(\kk_\perp) - z \right)^{-1} \, \di z \]
obeys all the required properties.
\end{proof}

\subsubsection{The two-step logarithm}

Clearly, not all families of unitary matrices $\alpha$ admit a continuous and periodic logarithm, that is, they are not all in the form $\alpha\sub{log}$ as in \eqref{eqn:alphalog}. Roughly speaking, this is related to the fact that their eigenvalues can cross and wind around the unit circle, thus preventing a continuous choice of a branch cut for the logarithm which lies always in the resolvent set of $\alpha(\kk_\perp)$.

Using the Cayley transform (Proposition \ref{prop:Cayley}), however, we can construct a family $\beta$ as required in Problem \ref{pbl:alpha} for a family $\alpha$ which is close to one in the form \eqref{eqn:alphalog}.

\begin{proposition} \label{prop:AlmostLog}
Let $\alpha = \set{\alpha(\kk_\perp)}_{\kk_\perp \in \R^{d-1}}$ be as in the statement of Problem \ref{pbl:alpha}. Assume that $\alpha$ satisfies 
\begin{equation} \label{eqn:almost_alphalog}
\sup_{\kk_\perp \in \R^{d-1}} \norm{\alpha(\kk_\perp) - \alpha\sub{log}(\kk_\perp)} < 2, \quad \alpha\sub{log}(\kk_\perp) := \eu^{\iu \, h(\kk_\perp)},
\end{equation}
where $h = \set{h(\kk_\perp)}_{\kk_\perp \in \R^{d-1}}$ is a continuous $\Z^{d-1}$-periodic family of self-adjoint matrices as in \eqref{eqn:h_properties}. Then there exists a continuous family $\beta = \set{\beta(\kk) = \beta(k_1, \kk_\perp)}_{\kk = (k_1, \kk_\perp) \in \R^d}$ of unitary matrices which is periodic with respect to $\kk_\perp$ and satisfies \eqref{eqn:alpha_vs_beta}. Moreover, $\beta\sub{log}$ can be explicitly constructed.

If moreover $\alpha$ is TRS in the sense of \eqref{eqn:CS'} and $h$ can be chosen such that \eqref{eqn:hTRS} holds, then $\beta\sub{log}$ can be constructed so as to satisfy also \eqref{eqn:beta-TRS}.
\end{proposition}
\begin{proof}
The closeness condition \eqref{eqn:almost_alphalog} implies that
\[ \sup_{\kk_\perp \in \R^{d-1}} \norm{\widetilde{\alpha}(\kk_\perp) - \Id} < 2, \quad \text{where} \quad \widetilde{\alpha}(\kk_\perp) := \eu^{-\iu \, h(\kk_\perp)/2} \, \alpha(\kk_\perp) \, \eu^{-\iu \, h(\kk_\perp)/2}. \]
In particular, $-1$ lies the resolvent of $\widetilde{\alpha}(\kk_\perp)$ for all $\kk_\perp \in \R^{d-1}$, and Proposition \ref{prop:Cayley} applies. It follows that there exists a continuous and $\Z^{d-1}$-periodic family $\widetilde{h} = \set{\widetilde{h}(\kk_\perp)}_{\kk_\perp \in \R^{d-1}}$ such that
\[ \eu^{-\iu \, \widetilde{h}(\kk_\perp)/2} \, \eu^{-\iu \, h(\kk_\perp)/2} \, \alpha(\kk_\perp) \, \eu^{-\iu \, h(\kk_\perp)/2} \, \eu^{-\iu \, \widetilde{h}(\kk_\perp)/2} = \Id. \]
Set now
\[ \beta(k_1, \kk_\perp) := \eu^{- \iu \, k_1 \, h(\kk_\perp)} \, \eu^{- \iu \, k_1 \, \widetilde{h}(\kk_\perp)}, \quad k_1 \in \left[-1/2, 1/2\right], \: \kk_\perp \in \R^{d-1}. \]
Arguing similarly to what we did in the proof of Proposition \ref{prop:GoodLog}, we can use again \eqref{eqn:ext_beta} to extend this definition to all $k_1 \in \R$ and obtain a family of matrices $\set{\beta(\kk)}_{\kk \in \R^2}$ with all the properties required in the statement. 

If the family $h$ satisfies also \eqref{eqn:hTRS} and $\alpha$ is TRS, then $\widetilde{\alpha}$ is TRS as well, as one can immediately check. The Cayley transform preserves this symmetry, and $\widetilde{h}$ then also satisfies \eqref{eqn:hTRS}. An argument analogous to \eqref{eqn:betalogTRS} then yields that $\beta$ enjoys the property \eqref{eqn:beta-TRS}.
\end{proof}

The above Proposition covers the case of the family of matching matrices associated to a $2$-dimensional family of projectors, as we will discuss in Section \ref{sec:d=2}. Indeed, we prove that in that case the family of matrices $\alpha$ is \emph{arbitrarily close} to a family of matrices $\alpha\sub{gap}$ which indeed admits a continuous and periodic logarithm. This will prove Theorem \ref{thm:MainResults_Beta}(\ref{item:BetaResult_d=2_noTRS}). Moreover, if the $\Z_2$ index $\mathcal{I}(\alpha)$ from \eqref{eqn:rueda} vanishes, then the latter logarithm can be chosen to satisfy also \eqref{eqn:hTRS}, yielding a family $\beta$ enjoying \eqref{eqn:alpha_vs_beta} and \eqref{eqn:beta-TRS}; the vanishing of the index is also a necessary condition for the existence of the latter family $\beta$. This will prove also Theorem \ref{thm:MainResults_Beta}(\ref{item:BetaResult_d=2}) and conclude the proof of Theorem \ref{thm:MainResults_BB}.

\goodbreak


\section{Symmetric Bloch frames in $d=1$} \label{sec:d=1}

This Section contains the proof of Theorem \ref{thm:MainResults_BB}(\ref{item:Result_d=1}), which is equivalent to Theorem \ref{thm:MainResults_Beta}(\ref{item:BetaResult_d=1}) via Proposition \ref{prop:rotate} as shown in the previous Section. We will see that the construction of $1$-dimensional continuous and symmetric Bloch frames can be expressed purely in terms of parallel transport, and leads automatically a \emph{real analytic} frame.

\subsection{Symmetric Bloch frames in $d=0$} \label{sec:d=0}

Let $\set{P(k)}_{k \in \R}$ be a family of projectors satisfying Assumption \ref{assum:proj}. The inductive construction of a continuous symmetric Bloch frame that we presented in Section \ref{sec:d=d} requires to focus first on the subspace $\Ran P(0) \subset \Hi$. We construct a TRS (\ie symplectic, compare \eqref{eqn:symplectic}) frame in this space as follows.

Let $\Xi_1 \in \Ran P(0)$ be any unit vector, and define $\Xi_2 := \theta \Xi_1$. As was already observed (compare Remark \ref{remark:eps}), the antiunitarity of $\theta$ implies that $\Xi_2$ is again of unit length and is orthogonal to $\Xi_1$.

If $m=2$, then $\Xi(0) := \set{\Xi_1, \Xi_2}$ constitutes an orthonormal basis in $\Ran P(0)$ satisfying $\Xi(0) = \theta \Xi(0) \act \eps$, with $\eps$ in the form \eqref{eqn:eps=J}. If $m > 2$, instead, choose a unit vector $\Xi_3$ in $\Span \{\Xi_1,$ $\Xi_2\}^\perp \subset \Ran P(0)$. Define $\Xi_4 := \theta \Xi_3$; as before, $\Xi_4$ and $\Xi_3$ are orthogonal to each other, and in addition $\Xi_4$ is also orthogonal to $\Xi_1$ and $\Xi_2$ since
\begin{gather*}
\scal{\Xi_4}{\Xi_1} = \scal{\theta \Xi_3}{\Xi_1} = \scal{\theta\Xi_1}{\theta^2 \Xi_3} = - \scal{\Xi_2}{\Xi_3} = 0, \\
\scal{\Xi_4}{\Xi_2} = \scal{\theta \Xi_3}{\Xi_2} = \scal{\theta\Xi_2}{\theta^2 \Xi_3} = - \scal{\Xi_1}{\Xi_3} = 0.
\end{gather*}
By iterating this procedure, we arrive at an orthonormal basis $\Xi(0) := \{ \Xi_1, \ldots,$ $\Xi_m\}$ satisfying 
\[ \Xi(0) = \theta \Xi(0) \act \eps. \]

\subsection{Extension by parallel transport to $d=1$}

We now want to parallel transport this frame along the real line $k \in \R$. To do so, we first notice that by the spectral theorem there exists $M = M^* \in \BH$ with spectrum in $(-\pi, \pi]$ such that $T(1, 0) = \eu^{\iu M}$, where $T(\cdot, \cdot)$ denotes the parallel transport introduced in Section \ref{sec:parallel}. Notice that the relation \ref{item:T-TRS} for $T(1,0)$ reads $\theta \, T(1,0) \, \theta^{-1} = T(1,0)^*$, which in particular implies that the projector-valued spectral measure of $T(1,0)$ commutes with $\theta$. It then follows by spectral calculus that $M$ can be chosen to also commute with $\theta$ (for details, see \cite[Sec.~2.6]{CorneanHerbstNenciu15}). We introduce then the unitary-valued map defined by
\begin{equation} \label{eqn:parallel}
W(k) := T(k, 0) \eu^{- \iu k M} \quad \in \U(\Hi).
\end{equation}
This map enjoys a number of properties, which follow from the corresponding ones satisfied by the parallel transport unitaries, namely:
\begin{enumerate}[label=$(\mathrm{W}_\arabic*)$,ref=$(\mathrm{W}_\arabic*)$]
 \item \label{item:W-smooth} the map $\R \ni k \mapsto W(k) \in \U(\Hi)$ is real analytic;
 \item \label{item:W-periodic} the map $k \mapsto W(k)$ is $\Z$-periodic, \ie $W(k+1) = W(k)$ for $k \in \R$;
 \item \label{item:W-TRS} the map $k \mapsto W(k)$ satisfies $W(-k) = \theta \, W(k) \, \theta^{-1}$.
\end{enumerate}

Next, define
\begin{equation} \label{eqn:Xi1d}
\Xi(k) := W(k) \, \Xi(0), \quad k \in \R,
\end{equation}
where $W(k) \in \U(\Hi)$ is defined in \eqref{eqn:parallel}. We observe first of all that each element $\Xi_a(k)$ of the frame lies in $\Ran P(k)$, in view of the intertwining property \ref{item:T-intertwine} and the fact that $T(1,0) = \eu^{\iu M}$ commutes with $P(1) = P(0)$ (and hence so does $M$). Moreover, properties \ref{item:W-smooth}, \ref{item:W-periodic} and \ref{item:W-TRS} imply that $\Xi(k)$ forms an orthonormal basis, depending analytically on $k$, which is periodic and TRS. Indeed, for example
\[ \theta \Xi(k) \act \eps = \theta \left( W(k) \, \Xi(0) \right) \act \eps = W(-k) (\theta \Xi(0) \act \eps) = W(-k) \Xi(0) = \Xi(-k). \]

This solves the Problem stated in Section \ref{sec:problem} for $d=1$, and proves Theorem \ref{thm:MainResults_BB}(\ref{item:Result_d=1}).

\subsection{Matching matrix in $d=1$, and a geometric index} \label{sec:matching1d}

The relation of the above construction to the language of matching matrices, introduced in the previous Section, can be understood from \eqref{eqn:alpha_vs_T}, which reads now
\begin{equation} \label{eqn:alpha_vs_T1d}
T(1,0) \, \Xi(0) = \Xi(0) \act \alpha.
\end{equation}
We see that in this case there is only a single matching matrix $\alpha$, which represents the restriction of the parallel transport unitary $T(1,0)$ to $\Ran P(0) = \Ran P(1)$ in the basis $\Xi(0)$. The matrix $\alpha$ then admits a logarithm, chosen with respect to any branch cut which does not intersect the unit circle in one of its eigenvalues: $\alpha = \eu^{\iu \, h}$, $h = h^*$. It is immediately recognized that the self-adjoint matrix $h$ is then the matrix associated to the restriction to $\Ran P(0)$ of the operator $M=M^*$ such that $T(1,0) = \eu^{\iu \, M}$, again computed with respect to the basis $\Xi(0)$. The family of matrices defined by $\beta(k) := \eu^{- \iu \, k \, h}$, as in the proof of Proposition \ref{prop:GoodLog}, then ``corrects'' the parallel transport $T(k,0)$ to enforce periodicity in $k$, in the sense of Proposition \ref{prop:rotate}. This also proves Theorem \ref{thm:MainResults_Beta}(\ref{item:BetaResult_d=1}).

Notice that, in this context, the TRS condition \eqref{eqn:CS'} for the matching matrix $\alpha$ reads
\begin{equation} \label{eqn:alpha_TRS1d}
\eps \, \alpha = \alpha^t \, \eps.
\end{equation}
This property implies, as is well known, that its eigenvalues have even multiplicity. We recall the proof of this fact in the following Lemma.

\begin{lemma}[Kramers degeneracy] \label{lemma:Kramers}
Let $\alpha \in U(m)$ satisfy \eqref{eqn:alpha_TRS1d}. Then $\alpha$ has \emph{Kramers degeneracy}, namely all its eigenvalues are even degenerate. 
\end{lemma}
\begin{proof}
Let $\Theta$ be the antiunitary operator $\Theta := \eps C$, where $C$ is the complex conjugation in $\C^m$ with respect to the frame in which $\eps$ has the form \eqref{eqn:eps=J}. Notice that $\Theta^2 = - \Id$, and that, in view of the unitarity of $\alpha$, \eqref{eqn:alpha_TRS1d} reads
\[ \Theta \, \alpha^{-1} \, \Theta^{-1} = \alpha. \]
Finally, let $\mathbf{v} \in \C^m$ be an eigenvector for $\alpha$ of eigenvalue $\lambda \in U(1)$. Then
\[ \alpha \, \Theta \, \mathbf{v} = \Theta \alpha^{-1} \, \mathbf{v} = \Theta \, \left(\lambda^{-1} \, \mathbf{v}\right) = \lambda \, \Theta \, \mathbf{v} \]
so that also $\Theta \, \mathbf{v}$ is an eigenvector of eigenvalue $\lambda$. Moreover using the properties of $\Theta$
\[ \scal{\Theta \, \mathbf{v}}{\mathbf{v}} = \scal{\Theta \, \mathbf{v}}{\Theta (\Theta \, \mathbf{v})} = - \scal{\Theta \, \mathbf{v}}{\mathbf{v}} \quad \Longrightarrow \quad \scal{\Theta \, \mathbf{v}}{\mathbf{v}} = 0 \]
so that $\Theta \, \mathbf{v}$ is also orthogonal (and in particular linearly independent) to $\mathbf{v}$. This shows that eigenvectors of $\alpha$ come in $\Theta$-related pairs, as wanted.
\end{proof}

An alternative way to formulate Kramers degeneracy is in the form of a factorization property, which hold for TRS unitary matrices.

\begin{lemma} \label{lemma:factor1d}
A unitary matrix $\alpha \in U(m)$ satisfies \eqref{eqn:alpha_TRS1d} if and only if there exists $\gamma \in U(m)$ such that
\[ \alpha = \eps^{-1} \, \gamma^t \, \eps \, \gamma. \]
\end{lemma}
\begin{proof}
The proof can be found in \cite[Lemma~1]{FiorenzaMonacoPanati16} (see also \cite{Hua44,Schulz-Baldes15}). Assuming \eqref{eqn:alpha_TRS1d}, the matrix $\gamma$ can be constructed as follows: Writing $\alpha = \eu^{\iu \, h}$ with the spectrum of $h$ in $(-\pi, \pi]$, then $\gamma := \eu^{\iu \, h/2}$ satisfies the above condition.
\end{proof}

We want to show now that the determinant of $\alpha$ and of the matrix $\gamma$ appearing in Lemma \ref{lemma:factor1d}, or rather their arguments, can be considered as \emph{geometric $\bmod \: 2$ indices} of the $1$-dimensional family of projectors $\set{P(k)}_{k \in \R}$ as in Assumption \ref{assum:proj} to which they are associated. This will exploit the connection of $\alpha$ with the parallel transport operator $T(1,0)$ given by \eqref{eqn:alpha_vs_T}. Even though these quantities are \emph{not} topological invariants for the latter (in the sense that their values can change if the family undergoes a continuous deformation), they can be still expressed in geometric terms by means of the \emph{Berry connection} associated to the projectors. Moreover, the quantity $\det \gamma$ will play a crucial r\^{o}le in the evaluation of the $\Z_2$ invariant for $2$-dimensional families of projectors (see \eg \eqref{eqn:IntegralRueda_b}).

\begin{proposition} \label{prop:detalphaBerry}
Let $\set{P(k)}_{k \in \R}$ satisfy \ref{item:smooth} and \ref{item:periodic} from Assumption \ref{assum:proj}. Let $\Xi^{\mathrm{per}}(0)$ be any orthonormal basis in $\Ran P(0)$, and define 
\[ \Xi^{\mathrm{per}}(k) := W(k) \, \Xi^{\mathrm{per}}(0), \quad k \in \R, \]
where $W(k)$ is as in \eqref{eqn:parallel}. Then $\Xi^{\mathrm{per}}$ gives a real analytic and periodic Bloch frame for $\set{P(k)}_{k \in \R}$. Moreover, if $\alpha \in U(m)$ is defined via \eqref{eqn:alpha_vs_T1d}, then
\begin{equation} \label{eqn:det=Berry}
\det \left( \alpha \right) = \exp \left(- \iu \oint_\T \A\right),
\end{equation}
where $\T := \R / \Z$ is the \emph{Brillouin torus} and $\A$ is the \emph{trace of the Berry connection $1$-form}
\begin{equation} \label{eqn:ABerry}
\A := - \iu \sum_{a=1}^{m} \scal{\Xi^{\mathrm{per}}_a(k)}{\partial_k \Xi^{\mathrm{per}}_a(k)} \di k.
\end{equation} 
\end{proposition}

The equality in \eqref{eqn:det=Berry} is known in the gauge theory community as a formula evaluating a \emph{Wilson loop} \cite{KarpMansouriRno00, FiorenzaSatiSchreiber15}, and is present also in the condensed matter literature on Berry's phase \cite{Berry84}. We give below an independent proof, more suited to our language of Bloch frames.

\begin{proof}[Proof of Proposition \ref{prop:detalphaBerry}]
When $\set{P(k)}_{k \in \R}$ satisfies \ref{item:smooth} and \ref{item:periodic}, then the family of parallel transport operators $\set{T(k,0)}_{k \in \R}$ satisfies all the properties in Proposition \ref{prop:parallel} except for \ref{item:T-TRS}, and consequently the family of unitary operators $\set{W(k)}_{k \in \R}$ satisfies \ref{item:W-smooth} and \ref{item:W-periodic}. This yields that indeed $\Xi^{\mathrm{per}}$ is a real analytic and periodic Bloch frame for $\set{P(k)}_{k \in \R}$.

As for \eqref{eqn:det=Berry}, with the above notation $T(1,0) = \eu^{\iu M}$ and $\alpha = \eu^{\iu \, h}$, we deduce at once that
\[ \det \left( \alpha \right) = \eu^{\iu \, \tr(h)} = \eu^{\iu \Tr \left( P(0) \, M \, P(0) \right)} \]
where $\tr(\cdot)$ denotes the trace in $\C^m \simeq \Ran P(0)$ and $\Tr(\cdot)$ the trace in the Hilbert space $\Hi$. Thus, in order to prove \eqref{eqn:det=Berry}, we just need to show that
\begin{equation} \label{eqn:TrM}
\Tr \left( P(0) \, M \, P(0) \right) = - \oint_\T \A.
\end{equation}

In order to do so, define first of all $G(k) := [\partial_k P(k), P(k)]$, so that $\partial_k T(k,0) = G(k) \, T(k,0)$ (compare \eqref{eqn:parallel_def}). We compute the derivative of the unitary operator $W(k)$ defined in \eqref{eqn:parallel} as follows:
\[ \partial_k W(k) = \left( \partial_k T(k,0) \right) \eu^{- \iu k M} - \iu T(k,0) \, M \eu^{- \iu k M} = G(k) W(k) - \iu W(k) M. \]
Plugging this in \eqref{eqn:Xi1d} yields
\[ \partial_k \Xi^{\mathrm{per}}_a(k) = \left( \partial_k W(k) \right) \Xi^{\mathrm{per}}_a(0) = G(k) \Xi^{\mathrm{per}}_a(k) - \iu W(k) M \Xi^{\mathrm{per}}_a(0), \quad a \in \set{1, \ldots, m}. \]

From the above relation, we compute
\begin{align*}
\scal{\Xi^{\mathrm{per}}_a(k)}{\partial_k \Xi^{\mathrm{per}}_a(k)} & = \scal{\Xi^{\mathrm{per}}_a(k)}{G(k) \Xi^{\mathrm{per}}_a(k)} - \iu \scal{W(k) \Xi^{\mathrm{per}}_a(0)}{W(k) M \Xi^{\mathrm{per}}_a(0)} \\
& = \scal{\Xi^{\mathrm{per}}_a(k)}{G(k) \Xi^{\mathrm{per}}_a(k)} - \iu \scal{\Xi^{\mathrm{per}}_a(0)}{M \Xi^{\mathrm{per}}_a(0)},
\end{align*}
since $W(k)$ is unitary. The first term on the right-hand side of the above equality actually vanishes, since
\begin{align*}
\scal{\Xi^{\mathrm{per}}_a(k)}{[\partial_k P(k), P(k)] \, \Xi^{\mathrm{per}}_a(k)} & = \scal{\Xi^{\mathrm{per}}_a(k)}{\left(\partial_k P(k)\right) P(k) \, \Xi^{\mathrm{per}}_a(k)} \\
& \quad - \scal{\Xi^{\mathrm{per}}_a(k)}{ P(k) \left( \partial_k P(k)\right) \, \Xi^{\mathrm{per}}_a(k)} = 0
\end{align*}
as $P(k) \, \Xi^{\mathrm{per}}_a(k) = \Xi^{\mathrm{per}}_a(k)$ and $P(k) = P(k)^*$. Since $\Xi^{\mathrm{per}}(0)$ forms an orthonormal basis in $\Ran P(0)$, we obtain then
\[ \sum_{a=1}^{m} \scal{\Xi^{\mathrm{per}}_a(k)}{\partial_k \Xi^{\mathrm{per}}_a(k)} = - \iu \, \Tr \left( P(0) \, M \, P(0) \right). \]
By integrating both sides of the above equality with respect to $k$ over one period, we get exactly \eqref{eqn:TrM}. This concludes the proof of \eqref{eqn:det=Berry}.
\end{proof}

\begin{proposition} \label{prop:TRSBerry}
Let $\set{P(k)}_{k \in \R}$ be as in Assumption \ref{assum:proj}. Let $\Xi^{\mathrm{TRS}}$ be any real analytic and symmetric Bloch frame for $\set{P(k)}_{k \in \R}$. Compute the trace of the Berry connection $\A$ as in \eqref{eqn:ABerry}, using the frame $\Xi^{\mathrm{TRS}}$. Then
\begin{equation} \label{eqn:BerryEven}
\oint_\T \A = 2 \int_{[0,1/2]} \A.
\end{equation}
\end{proposition}
\begin{proof}
Let
\[ \omega(k) := \sum_{a=1}^{m} \scal{\Xi^{\mathrm{TRS}}_a(k)}{\partial_k \Xi^{\mathrm{TRS}}_a(k)}, \]
so that $\A = - \iu \, \omega(k) \, \di k$. The equality stated in \eqref{eqn:BerryEven} is implied at once by the fact that $\omega$ is an even and periodic function of $k$, since $\T \simeq [-1/2,1/2]$. The symmetry of $\omega$ descends from the fact that the Bloch frame $\Xi^{\mathrm{TRS}}$, which is used to compute $\A$, is time-reversal symmetric. Indeed, we begin by observing that
\begin{align*}
\theta \, \partial_k \Xi^{\mathrm{TRS}}(k) & = \theta \, \lim_{h \to 0} \frac{\Xi^{\mathrm{TRS}}(k+h) - \Xi^{\mathrm{TRS}}(k)}{h} = \lim_{h \to 0} \frac{\theta \, \Xi^{\mathrm{TRS}}(k+h) - \theta \, \Xi^{\mathrm{TRS}}(k)}{h} \\
& = \lim_{h \to 0} \frac{\Xi^{\mathrm{TRS}}(-k-h) \act \eps^{-1} - \Xi^{\mathrm{TRS}}(k) \act \eps^{-1}}{h} = - (\partial_k \Xi^{\mathrm{TRS}})(-k) \act \eps^{-1}
\end{align*}
or equivalently
\[ (\partial_k \Xi^{\mathrm{TRS}}_b)(-k) = - \sum_{c=1}^{m} [\theta \, \partial_k \Xi^{\mathrm{TRS}}_c(k)] \, \eps_{cb}. \]
By taking the scalar product of both sides of the above equality with $\Xi^{\mathrm{TRS}}_b(-k)$, we obtain
\begin{align*}
\scal{\Xi^{\mathrm{TRS}}_b(-k)}{(\partial_k \Xi^{\mathrm{TRS}}_b)(-k)} & = - \sum_{a,c=1}^{m} \scal{[\theta \, \Xi^{\mathrm{TRS}}_a(k)] \, \eps_{ab}}{[\theta \, \partial_k \Xi^{\mathrm{TRS}}_c(k)] \, \eps_{cb}} \\
& = - \sum_{a,c=1}^{m} \overline{\eps_{ab}} \, \scal{\theta \, \Xi^{\mathrm{TRS}}_a(k)}{\theta \, \partial_k \Xi^{\mathrm{TRS}}_c(k)} \, \eps_{cb} \\
& = - \sum_{a,c=1}^{m} \scal{\partial_k \Xi^{\mathrm{TRS}}_c(k)}{\Xi^{\mathrm{TRS}}_a(k)} \, \eps_{cb} \, (\eps^*)_{ba}.
\end{align*}

Summing over $b \in \set{1,\ldots,m}$ and employing the unitarity of $\eps$ yields
\[ \omega(-k) = - \sum_{a,c=1}^{m} \scal{\partial_k \Xi^{\mathrm{TRS}}_c(k)}{\Xi^{\mathrm{TRS}}_a(k)} \, \sum_{b=1}^{m} \eps_{cb} \, (\eps^*)_{ba} = - \sum_{a=1}^{m} \scal{\partial_k \Xi^{\mathrm{TRS}}_a(k)}{\Xi^{\mathrm{TRS}}_a(k)} = \omega(k) \]
where is the last equality we used the fact that
\[ - \scal{\partial_k \Xi^{\mathrm{TRS}}_a(k)}{\Xi^{\mathrm{TRS}}_a(k)} = \scal{\Xi^{\mathrm{TRS}}_a(k)}{\partial_k \Xi^{\mathrm{TRS}}_a(k)}, \quad a \in \set{1, \ldots, m}, \]
as can be easily derived by differentiating the equality $\scal{\Xi^{\mathrm{TRS}}_a(k)}{\Xi^{\mathrm{TRS}}_a(k)}=1$. 
\end{proof}

Combining the above two results, we obtain

\begin{proposition} \label{prop:BerryConnection}
Let $\set{P(k)}_{k \in \R}$ be as in Assumption \ref{assum:proj}, and let $\Xi$ be the real analytic and symmetric Bloch frame defined in \eqref{eqn:Xi1d}. Compute the trace of the Berry connection $\A$ as in \eqref{eqn:ABerry}, using the frame $\Xi$. Also, let $\alpha \in U(m)$ be defined as in \eqref{eqn:alpha_vs_T1d}, and compute $\gamma \in U(m)$ as in Lemma \ref{lemma:factor1d}. Then we have
\begin{equation} \label{eqn:gammaBerry1d}
2 \left( \frac{1}{2 \pi \iu} \log \det \gamma \right) \equiv \frac{1}{2 \pi} \oint_{\T} \A \bmod 2.
\end{equation}
\end{proposition}
\begin{proof}
The frame $\Xi$ in \eqref{eqn:Xi1d} is real analytic, periodic and TRS, so both the above Propositions apply. From \eqref{eqn:TrM} we can compute $\det \gamma = \det(\eu^{\iu \, h/2}) = \eu^{\iu \, \tr(h)/2} = \eu^{\iu \, \Tr(P(0) \, M \, P(0))/2}$, yielding
\[ \det \gamma = \exp \left( - \frac{\iu}{2} \, \oint_{\T} \A \right). \]
Since \eqref{eqn:BerryEven} holds, we have
\[  \det \gamma = \exp \left( -\iu \, \int_{[0, 1/2]} \A \right). \]
It follows that
\[ 2 \left( \frac{1}{2 \pi \iu} \log \det \gamma \right) \equiv 2 \left( \frac{1}{2 \pi} \int_{[0, 1/2]} \A \right) \equiv \frac{1}{2 \pi} \oint_{\T} \A \bmod 2 \]
using again \eqref{eqn:BerryEven}. This completes the proof.
\end{proof}

\begin{remark}[Gauge dependence of the Berry connection] \label{rmk:Gauge}
Notice that, although $\A$ is a \emph{gauge-dependent} quantity, the expression $\exp( - \iu \oint_\T \A)$ is not (as would {\it a posteriori} follow from \eqref{eqn:det=Berry}, where the left-hand side is independent of the choice of gauge). Indeed, if a (periodic) unitary gauge $\beta: \T \to U(m)$ is applied to the frame $\Xi$, \ie $\Xi \to \Xi \act \beta$, then $\A$ changes according to
\begin{equation} \label{eqn:BerryGauge}
\A \to \A + \iu \tr( \beta(k)^* \di \beta(k) ).
\end{equation}
Since $(2 \pi \iu)^{-1} \oint_\T \tr( \beta(k)^* \di \beta(k) )$ computes the degree of the map $\det \beta \colon \T \to U(1)$ (see \eqref{eqn:DegUn}), it is an integer, and we have that the quantity $\exp( - \iu \oint_\T \A)$ remains unchanged.

We remark that, when both periodicity and the TRS condition are enforced on the frame which is used to compute the (trace of the) Berry connection, then $\oint_\T \A$ is actually well-defined modulo $4 \pi$. Indeed, choosing a unitary gauge $\beta$ which preserves periodicity and TRS of the frame $\Xi$ is equivalent to choosing a map $\beta \colon \R \to U(m)$ such that
\begin{equation} \label{eqn:gammaperTR}
\beta(k) = \beta(k+1) \quad \text{and} \quad \beta(-k) = \eps^{-1} \, \overline{\beta(k)} \, \eps, \qquad k \in \R,
\end{equation}
as one can immediately verify by asking that both $\Xi(k)$ and $\Xi(k) \act \beta(k)$ satisfy \eqref{eqn:FramePeriodic} and \eqref{eqn:FrameTR} (compare the proof of Proposition \ref{prop:rotate}). The degree of the determinant of a map $\beta \colon \T \to U(m)$ satisfying \eqref{eqn:gammaperTR} must then be necessarily even. Indeed, the second equality in \eqref{eqn:gammaperTR} yields $\det \beta(-k) = (\det \beta(k))^{-1}$, and by differentiating 
\[ (\partial_k \det \beta)(-k) = (\det \beta(k))^{-2} \, (\partial_k \det \beta)(k) \]
or equivalently
\[ (\det \beta(-k))^{-1} \, (\partial_k \det \beta)(-k) = (\det \beta(k))^{-1} \, (\partial_k \det \beta)(k). \]
By integrating on $\T$ we then obtain
\begin{equation} \label{eqn:deggamma}
\begin{aligned}
\deg([\det \beta]) & = \frac{1}{2 \pi \iu} \oint_\T (\det \beta(k))^{-1} \, (\partial_k \det \beta)(k) \, \di k \\
& = 2 \left( \frac{1}{2 \pi \iu} \int_{0}^{1/2} (\det \beta(k))^{-1} \, (\partial_k \det \beta)(k) \, \di k \right).
\end{aligned}
\end{equation}

Notice now that, evaluating the second equality in \eqref{eqn:gammaperTR} at $k = 0$ and $k=1/2$ and using periodicity, we obtain
\[ \beta(k_*) = \eps^{-1} \, \overline{\beta(k_*)} \, \eps, \quad k_* \in \set{0,1/2}. \]
By putting $\eps$ in the symplectic form \eqref{eqn:eps=J}, we realize that the above condition states that $\beta(0)$ and $\beta(1/2)$ are actually \emph{symplectic matrices}, and thus unimodular: $\det \beta(0) = \det \beta(1/2) = 1$. The term inside the brackets on the right-hand side of \eqref{eqn:deggamma} computes then the degree of the \emph{periodic} function $\det \beta : [0,1/2] \to U(1)$, and is thus an integer. We conclude as wanted that $\deg([\det \beta]) \in 2 \Z$, which by \eqref{eqn:BerryGauge} implies that $\oint_\T \A$ is well-defined up to $(2 \cdot 2 \pi) \Z$. This is what makes the right-hand side of \eqref{eqn:gammaBerry1d} well-defined $\bmod \: 2$.
\end{remark}

\begin{remark}
The condition \eqref{eqn:alpha_TRS1d}, which holds for the matrix $\alpha$, can be restated as the skew-symmetry of the matrix $\eps \, \alpha$, since $\eps^t = - \eps$. The determinant of the matrix $\gamma$ as in Lemma \ref{lemma:factor1d}, which appears on the left-hand side of \eqref{eqn:gammaBerry1d}, can then be also expressed in a more intrinsic form (that is, directly in terms of $\alpha$) as
\[ \det \gamma = \Pf (\eps \, \alpha), \]
where $\Pf(\cdot)$ denotes the \emph{Pfaffian} of a skew-symmetric matrix. This follows from the properties of the Pfaffian, and the way it changes under similarity tranformations: indeed
\[ \Pf (\eps \, \alpha) = \Pf (\gamma^t \, \eps \, \gamma) = (\det \gamma) \, (Pf \, \eps) = \det \gamma \]
as $\Pf \eps = 1$.
\end{remark}

\goodbreak


\section{Symmetric Bloch frames in $d=2$} \label{sec:d=2}

We switch now to the $2$-dimensional case. Let $\set{P(\kk)}_{\kk \in \R^2}$ be a family of projectors satisfying Assumption \ref{assum:proj}. The restriction $\set{P(0,k_2)}_{k_2 \in \R}$ is then a $1$-dimensional family of projectors which also satisfies Assumption \ref{assum:proj}. We have showed in Section \ref{sec:d=1} how to construct a real analytic symmetric Bloch frame $\Xi(0, k_2)$ for such a family. As we discussed in Section \ref{sec:matching}, this leads to the definition of a family $\alpha = \set{\alpha(k_2)}_{k_2 \in \R}$ of matching matrices, given by
\begin{equation} \label{eqn:alpha_vs_Tk2}
T_{k_2}(1,0) \, \Xi(0, k_2) = \Xi(0, k_2) \act \alpha(k_2).
\end{equation}
The properties of the family of matching matrices, stated in Proposition \ref{prop:alpha}, are summarized in the following

\begin{assumption} \label{assum:alpha}
The family $\alpha = \set{\alpha(k_2)}_{k_2 \in \R}$ of $m \times m$ unitary matrices enjoys the following properties:
\begin{enumerate}[label=$(\mathrm{A}_\arabic*)$,ref=$(\mathrm{A}_\arabic*)$]
\item \label{item:alpha_analytic} the map $\R \ni k_2 \mapsto \alpha(k_2) \in U(m)$ is real analytic;
\item \label{item:alpha_periodic} the map $k_2 \mapsto \alpha(k_2)$ is $\Z$-periodic;
\item \label{item:alpha_TRS} for all $k_2 \in \R$
\[ \eps \, \alpha(k_2) = \alpha(-k_2)^t \, \eps \]
holds. \qedhere
\end{enumerate}
\end{assumption}

Recall how, in Section \ref{sec:d=d}, our original Problem \ref{pbl:Bloch} concerning the construction of a Bloch frame for $\set{P(\kk)}_{\kk \in \R^2}$ which is continuous and symmetric was reduced to Problem \ref{pbl:alpha} for a family of matrices $\alpha$ as in Assumption \ref{assum:alpha}.

\subsection{Kramers degeneracy and factorization of matching matrices} \label{sec:matching2d}

We deduce here a few properties that follow from the above conditions.

Notice first of all that \ref{item:alpha_TRS} can be restated at $k=0$ and $k=1/2$ as
\begin{equation} \label{eqn:alpha_TRS0}
\eps \, \alpha(0) = \alpha(0)^t \, \eps \quad \text{and} \quad \eps \, \alpha(1/2) = \alpha(1/2)^t \, \eps
\end{equation}
due to the periodicity \ref{item:alpha_periodic} of $\alpha$. In particular, the matrices $\alpha(0)$ and $\alpha(1/2)$ have Kramers degeneracy in view of Lemma \ref{lemma:Kramers}, and moreover Lemma \ref{lemma:factor1d} applies to both, yielding the existence of two matrices $\gamma(0), \gamma(1/2) \in U(m)$ such that
\begin{equation} \label{eqn:gammas}
\alpha(0) = \eps^{-1} \, \gamma(0)^t \, \eps \, \gamma(0) \quad \text{and} \quad \alpha(1/2) = \eps^{-1} \, \gamma(1/2)^t \, \eps \, \gamma(1/2).
\end{equation}

A stronger form of Kramers degeneracy is implied by the following factorization result for the family of matching matrices. This results generalizes Lemma \ref{lemma:factor1d} to periodic and TRS \emph{families} of unitary matrices, and will be useful also in the discussion of the $\Z_2$ index associated to such families (see Section \ref{sec:GrafPorta} below).

\begin{lemma} \label{lemma:factor}
Let $\set{\alpha(k_2)}_{k_2 \in \R}$ be as in Assumption \ref{assum:alpha}. There exists a family $\set{\gamma(k_2)}_{k_2 \in \R}$ of unitary matrices which is continuous, $\Z$-periodic, and such that
\begin{equation} \label{eqn:factor}
\alpha(k_2) = \eps^{-1} \, \gamma(-k_2)^t \, \eps \, \gamma(k_2).
\end{equation}

Conversely, if there exists a family $\set{\gamma(k_2)}_{k_2 \in \R}$ with all the above properties, then the family $\alpha$ defined by \eqref{eqn:factor} is continuous and satisfies \ref{item:alpha_periodic} and \ref{item:alpha_TRS}.
\end{lemma}
\begin{proof}
Let $\gamma(0), \gamma(1/2) \in U(m)$ be as in \eqref{eqn:gammas}. Write $\gamma(0)^{-1} \, \gamma(1/2) = \eu^{\iu \, N}$, with $N=N^*$, and define
\[ \widetilde{\gamma}(k_2) := \gamma(0) \, \eu^{2 \, \iu \,k_2 \, N}, \quad k_2 \in \left[0, 1/2 \right]. \]
The above family of matrices then interpolates between $\gamma(0)$ and $\gamma(1/2)$ as $k_2$ runs from $0$ to $1/2$. We can extend the definition of $\gamma(k_2)$ for $k_2 \in [-1/2,0]$ by setting
\[ \gamma(k_2) := \begin{cases} 
\widetilde{\gamma}(k_2) & \text{if } k_2 \in \left[ 0, 1/2 \right], \\
\left(\eps \, \alpha(-k_2) \, \widetilde{\gamma}(-k_2)^{-1} \, \eps^{-1} \right)^t & \text{if } k_2 \in \left[ -1/2, 0 \right].
\end{cases} \]
This extension is continuous at $k_2 = 0$, since in view of the first equality in \eqref{eqn:gammas} and the fact that $\widetilde{\gamma}(0) = \gamma(0)$ we have
\[ \left(\eps \, \alpha(0) \, \widetilde{\gamma}(0)^{-1} \, \eps^{-1} \right)^t = \left(\gamma(0)^t \, \eps \, \gamma(0) \, \widetilde{\gamma}(0)^{-1} \, \eps^{-1} \right)^t = \widetilde{\gamma}(0). \]
Notice moreover that, arguing similarly, one obtains also that $\gamma(-1/2) = \gamma(1/2)$ in view of the second equality in \eqref{eqn:gammas}. Thus, we can extend continuously $\gamma(k_2)$ in a periodic fashion to all $k_2 \in \R$. The required property \eqref{eqn:factor} is then satisfied by construction, owing to the periodicity \ref{item:alpha_periodic} of $\alpha$.

The converse statement is of immediate verification.
\end{proof}

\subsection{Interlude: The Graf--Porta index $\mathcal{I}$} \label{sec:GrafPorta}

From a topological viewpoint, it seems that families of matrices as in Assumption \ref{assum:alpha} are all trivial. Indeed, if $\alpha$ satisfies the latter Assumption, then \ref{item:alpha_TRS} implies that the function $k_2 \mapsto \det \alpha(k_2)$ is even, and hence its degree is zero, as can be realized at once from the integral formula \eqref{eqn:degree}. In particular, any such family of matrices is null-homotopic, \ie it can be continuously deformed to the constant family of matrices $\alpha_0(k_2) \equiv \Id$.

In \cite{GrafPorta13}, however, Graf and Porta introduced a $\Z_2$-valued index associated to families of unitary matrices as in Assumption \ref{assum:alpha}. We will recognize in the next Subsections that the vanishing of this index is equivalent to the existence of a family of matrices $\beta$ which solves Problem \ref{pbl:alpha}, and that it identifies the homotopy class of the family whenever a TRS constraint is required. It is then useful to recall its definition and main properties, for the reader's convenience.

The construction in \cite{GrafPorta13} goes as follows. Let $\set{P(\kk)}_{\kk \in \R^2}$ be a family of projectors satisfying Assumption \ref{assum:proj}. Given any continuous Bloch frame $\Psi(\kk)$ which satisfies TRS and periodicity in $k_2$, a matching matrix $\alpha(k_2)$ is defined by comparing the values of $\Psi$ at $k_1 = -1/2$ and $k_1 = 1/2$, like in \eqref{eqn:alpha}. Such family of unitary matrices then satisfies Assumption \ref{assum:alpha}, implying in particular Kramers degeneracy of $\alpha(0)$ and $\alpha(1/2)$ (Lemma \ref{lemma:Kramers}).

Choose a continuous labelling
\[ \sigma(\alpha(k_2)) = \set{\lambda_1(k_2), \ldots, \lambda_m(k_2)}, \quad k_2 \in \left[0,1/2\right] \] 
of the eigenvalues of $\alpha(k_2)$ (this is possible in view of \cite[Thm.~II.5.2]{Kato66}). We interpret $\lambda_i$ as the $i$-th coordinate of a continuous map $\Lambda \colon [0,1/2] \to U(1)^m$. We choose also a continuous extension $\widetilde{\Lambda} \colon [-1/2,0] \to U(1)^m$ of $\Lambda$, in such a way that $\widetilde{\Lambda}(0) = \Lambda(0)$ and $\widetilde{\Lambda}(-1/2) = \Lambda(1/2)$, and each $\widetilde{\lambda}_i(k_2)$ is ``even degenerate'' (\ie it coincides with some $\widetilde{\lambda}_j(k_2)$ with $j \ne i$) for all $k_2 \in [-1/2,0]$. The map $\widetilde{\Lambda} \, \sharp \, \Lambda \colon [-1/2,1/2] \to U(1)^m$, defined by 
\[ \widetilde{\Lambda} \, \sharp \, \Lambda(k_2) := \begin{cases}
\widetilde{\Lambda}(k_2) & \text{if } k_2 \in [-1/2,0], \\
\Lambda(k_2) & \text{if } k_2 \in [0,1/2],
\end{cases} \]
is then periodic, and so is the product $\det \widetilde{\Lambda} \, \sharp \, \Lambda(k_2) := \prod_{i=1}^{m} (\widetilde{\Lambda} \, \sharp \, \Lambda)_i(k_2)$ of all its coordinates. As any periodic map with values in $U(1)$, it admits a \emph{topological degree} (as was discussed in Section \ref{sec:DegreeUnitary}), an integer which can be evaluated through \eqref{eqn:degree}. It is proved in \cite{GrafPorta13} that 
\begin{equation} \label{eqn:rueda}
\mathcal{I}(\alpha) := \deg([\det \widetilde{\Lambda} \, \sharp \, \Lambda]) \bmod 2
\end{equation}
is a well-defined $\Z_2$ index, \ie the parity of the integer $\deg([\det \widetilde{\Lambda} \, \sharp \, \Lambda])$ does not depend on the chosen ``even degenerate'' extension $\widetilde{\Lambda}$.

We next prove a different formula for the Graf--Porta $\Z_2$ index, expressed purely in terms of the family of matrices $\set{\alpha(k_2)}_{k_2 \in \R}$ from which it is computed.

\begin{proposition} \label{prop:IntegralRueda}
Let $\set{\alpha(k_2)}_{k_2 \in \R}$ be as in Assumption \ref{assum:alpha}. Then its $\Z_2$ index $\mathcal{I}(\alpha)$, defined as in \eqref{eqn:rueda}, equals
\begin{subequations}
\begin{align}
\mathcal{I}(\alpha) & = \frac{1}{2 \pi \iu} \left( \int_{0}^{1/2} A(k_2)^{-1} \, A'(k_2) \, \di k_2 + 2 \log \frac{\prod_{j=1}^{n} \eu^{\iu \, \mu_j}}{\prod_{j=1}^{n} \eu^{\iu \, \nu_j}} \right) \mod 2 \label{eqn:IntegralRueda_a} \\
& = \frac{1}{2 \pi \iu} \left( \int_{0}^{1/2} A(k_2)^{-1} \, A'(k_2) \, \di k_2 + 2 \log \frac{C(0)}{C(1/2)} \right) \mod 2 \label{eqn:IntegralRueda_b} \\
& = \frac{1}{2 \pi \iu} \int_{-1/2}^{1/2} C(k_2)^{-1} \, C'(k_2) \, \di k_2 \mod 2 \label{eqn:IntegralRueda_c}
\end{align}
\end{subequations}
where:
\begin{itemize}
\item $A(k_2) := \det \alpha(k_2)$, and the prime denotes derivative with respect to $k_2$;
\item $\set{\eu^{\iu \, \mu_j}}_{1\le j\le n}$, $n := m/2$, (respectively $\set{\eu^{\iu \, \nu_j}}_{1\le j\le n}$) is the set of the doubly degenerate eigenvalues (repeated according to multiplicity) of the matrix $\alpha(0)$ (respectively $\alpha(1/2)$);
\item $C(k_2) := \det \gamma(k_2)$, with $\gamma(k_2)$ as in \eqref{eqn:factor}.
\end{itemize}
\end{proposition}

\begin{remark} \label{rmk:logIntRueda}
Notice that the branch of the logarithm appearing on the right-hand side of \eqref{eqn:IntegralRueda_a} and \eqref{eqn:IntegralRueda_b} should be chosen so that the branch cut does not intersect the unit circle in any of the eigevalues of $\alpha(0)$ or $\alpha(1/2)$. If this restriction is satisfied, then the choice of the branch does not influence the value of $\mathcal{I}(\alpha) \in \Z_2$, as a different choice would shift the logarithm by a multiple of $2 \pi \iu$, and hence the expressions in \eqref{eqn:IntegralRueda_a} and \eqref{eqn:IntegralRueda_b} are well defined $\bmod\:2$.
\end{remark}

\begin{proof}[Proof of Proposition \ref{prop:IntegralRueda}]
The computation of the Graf--Porta index requires an ``even degenerate'' extension $\widetilde{\Lambda}$ of the family of eigenvalues of $\set{\alpha(k_2)}_{k_2 \in [0,1/2]}$ to the interval $[-1/2,0]$, but the value of $\mathcal{I}(\alpha)$ is independent of such an extension. We can therefore perform an appropriate choice of $\widetilde{\Lambda}$. We define $\widetilde{\Lambda}$ as follows: By Kramers degeneracy (Lemma \ref{lemma:Kramers}), the eigenvalues of $\alpha(0)$ and $\alpha(1/2)$ come in pairs. Let $\set{\eu^{\iu \mu_1}, \ldots, \eu^{\iu \mu_n}}$, $n = m/2$, be the even-degenerate eigenvalues of $\alpha(0)$, and respectively $\set{\eu^{\iu \nu_1}, \ldots, \eu^{\iu \nu_n}}$ those of $\alpha(1/2)$, as in the statement; the arguments $\mu_j$ and $\nu_j$, $j \in \set{1, \ldots, n}$, are chosen in $[0,2 \pi)$. Also, let $W_0$ be the unitary map which diagonalizes $\alpha(0)$, namely $\alpha(0) = W_0 \left( \bigoplus_{j=1}^{n} \eu^{\iu \mu_j} \, \Id_2 \right) W_0^*$, and analogously let $W_{1/2}$ be the unitary map diagonalizing $\alpha(1/2)$. Define then, for $k_2 \in [-1/2,0]$,
\[ \widetilde{\alpha}(k_2) := W(k_2) \, \eta(k_2) \, W(k_2)^*, \quad \text{where} \quad \eta(k_2) := \bigoplus_{j=1}^{n} \eu^{\iu \, [(1+2k_2) \, \mu_j - 2 \, k_2 \, \nu_j]} \, \Id_2 \]
and where $W(k_2)$ is any continuous path of unitary matrices interpolating $W_{1/2}$ and $W_0$ for $k_2 \in [-1/2,0]$ (such path exists because the unitary group is path-connected%
\footnote{A possible choice of this path is as follows: Writing $W_{1/2}^{-1} \, W_0 = \eu^{\iu \, X}$ with $X = X^*$, set
\[ W(k_2) := W_{1/2} \, \eu^{\iu \, (1+2k_2) \, X}, \quad k_2 \in [-1/2,0]. \]}%
). The spectrum of $\widetilde{\alpha}(k_2)$, which is completely determined by the one of $\eta(k_2)$, gives then a possible extension $\widetilde{\Lambda}$.

Set $\widetilde{A}(k_2) := \det \widetilde{\alpha}(k_2)$. We compute
\begin{align*}
\int_{-1/2}^{0} \widetilde{A}(k_2)^{-1} \, \widetilde{A}'(k_2) \, \di k_2 & = \prod_{j=1}^{n} \int_{-1/2}^{0} \eu^{-2\iu [(1+2k_2) \mu_j - 2 k_2 \nu_j]} \, \partial_{k_2} \left( \eu^{2\iu [(1+2k_2) \mu_j - 2 k_2 \nu_j]} \right) \, \di k_2 \\
& = 2 \iu \sum_{j=1}^{m} \mu_j - \nu_j \equiv 2 \log \dfrac{\prod_{j=1}^{n} \eu^{\iu \mu_j}}{\prod_{j=1}^{n} \eu^{\iu \nu_j}} \mod (2 \cdot 2 \pi \iu).
\end{align*}
The products appearing on the right-hand side of the above equality count each of the doubly degenerate eigenvalues of $\alpha(0)$ and $\alpha(1/2)$ only once. Since the matrix $\gamma(0)$ (respectively $\gamma(1/2)$) appearing in \eqref{eqn:gammas} has by construction doubly degenerate eigenvalues $\eu^{\iu \mu_j/2}$ (respectively $\eu^{\iu \nu_j/2}$), $j \in \set{1, \ldots, n}$ (compare the proof of Lemma \ref{lemma:factor1d}), the product mentioned before indeed computes its determinant $C(0)$ (respectively $C(1/2)$):
\[ 2 \log \dfrac{\prod_{j=1}^{n} \eu^{\iu \mu_j}}{\prod_{j=1}^{n} \eu^{\iu \nu_j}} = 2 \log \frac{C(0)}{C(1/2)}. \]
This concludes the proof of \eqref{eqn:IntegralRueda_a} and \eqref{eqn:IntegralRueda_b}, since by the definition \eqref{eqn:rueda}
\[ \mathcal{I}(\alpha) = \frac{1}{2 \pi \iu} \left( \int_{0}^{1/2} A(k_2)^{-1} \, A'(k_2) \, \di k_2 + \int_{-1/2}^{0} \widetilde{A}(k_2)^{-1} \, \widetilde{A}'(k_2) \, \di k_2 \right) \mod 2. \]

We come to \eqref{eqn:IntegralRueda_c}. From \eqref{eqn:factor} we deduce at once that
\[ A(k_2) := \det \alpha(k_2) = C(k_2) \, C(-k_2), \quad C(k_2) := \det \gamma(k_2). \]
One can then immediately compute that
\[ A(k_2)^{-1} \, A'(k_2) = C(k_2)^{-1} \, C'(k_2) - C(-k_2)^{-1} \, C'(-k_2). \]
Integrating the above equation for $k_2 \in [0,1/2]$ and using the periodicity of $\alpha(k_2)$ yields
\begin{align*}
\int_{0}^{1/2} A(k_2)^{-1} \, A'(k_2) \, \di k_2 & = \int_{0}^{1/2} C(k_2)^{-1} \, C'(k_2) \, \di k_2 - \int_{0}^{1/2} C(-k_2)^{-1} \, C'(-k_2) \, \di k_2 \\
& = \int_{0}^{1/2} C(k_2)^{-1} \, C'(k_2) \, \di k_2 - \int_{-1/2}^{0} C(k_2)^{-1} \, C'(k_2) \, \di k_2.
\end{align*}
By using now that
\[ \frac{2}{2 \pi \iu} \log \frac{C(0)}{C(1/2)} \equiv \frac{2}{2 \pi \iu} \int_{-1/2}^{0} C(k_2)^{-1} \, C'(k_2) \, \di k_2 \mod 2, \]
we obtain from the formula \eqref{eqn:IntegralRueda_b} we proved above that
\[ \mathcal{I}(\alpha) = \frac{1}{2 \pi \iu} \int_{-1/2}^{1/2} C(k_2)^{-1} \, C'(k_2) \, \di k_2 \mod 2. \]
The index of the matrix $\alpha$ is then equal to the parity of the degree of the determinant of the matrix $\gamma$ appearing in \eqref{eqn:factor}, introduced in Section \ref{sec:DegreeUnitary}. This concludes the proof.
\end{proof}

The above formula \eqref{eqn:IntegralRueda_a} shows in particular in a direct way that the index $\mathcal{I}(\alpha)$ is a homotopy invariant of the family of matrices $\alpha$. Later we will be able to prove that the Graf--Porta index actually gives a \emph{complete homotopy invariant} for families of matrices as in Assumption \ref{assum:alpha} (see Theorem \ref{thm:homotopyI}).

\begin{lemma} \label{lemma:Ihomotopy}
Let $\set{\alpha_s(k_2)}_{k_2 \in \R}$, $s \in [0,1]$, be a continuous homotopy of families of matrices as Assumption \ref{assum:alpha}. Then $\mathcal{I}(\alpha_s)$ is constant in $s$, and in particular
\[ \mathcal{I}(\alpha_0) = \mathcal{I}(\alpha_1) \in \Z_2. \]
\end{lemma}
\begin{proof}
The right-hand side of \eqref{eqn:IntegralRueda_a} depends continuously only on the spectral properties (determinant and eigenvalues) of the family itself, and is integer-valued: as such, it must be constant. 

Notice that the branch cut for the logarithm appearing in \eqref{eqn:IntegralRueda_a} can be chosen to be locally constant in $s$; on the other hand, we have already discussed in Remark \ref{rmk:logIntRueda} that a different choice of the branch cut does not affect the parity of the integer which computes $\mathcal{I}(\alpha)$. This concludes the proof.
\end{proof}

The above homotopy invariance can be used to show that families of matrices as in Assumption \ref{assum:alpha} which are sufficiently close to each other have the same $\Z_2$ index.

\begin{proposition} \label{prop:closeness}
Let $\set{\alpha_0(k_2)}_{k_2 \in \R}$ and $\set{\alpha_1(k_2)}_{k_2 \in \R}$ be as in Assumption \ref{assum:alpha}. Assume that
\[ \sup_{k_2 \in \R} \norm{\alpha_1(k_2) - \alpha_0(k_2)} < 2. \]
Then
\[ \mathcal{I}(\alpha_0) = \mathcal{I}(\alpha_1) \in \Z_2. \]
\end{proposition}
\begin{proof}
Write $\alpha_0(k_2) = \eps^{-1} \, \gamma(-k_2)^t \, \eps \, \gamma(k_2)$ as in Lemma \ref{lemma:factor}. Consider the matrix 
\[ \widetilde{\alpha}(k_2) := \eps^{-1} \, \overline{\gamma(-k_2)} \, \eps \, \alpha_1(k_2) \, \gamma(k_2)^*. \]
Then
\[ \eps \, \widetilde{\alpha}(k_2) = \overline{\gamma(-k_2)} \, \eps \, \alpha_1(k_2) \, \gamma(k_2)^* = \overline{\gamma(-k_2)} \, \alpha_1(-k_2)^t \, \eps \, \gamma(k_2)^* = \widetilde{\alpha}(-k_2)^t \, \eps \]
so that the family $\widetilde{\alpha}$ is continuous and satisfies \ref{item:alpha_periodic} and \ref{item:alpha_TRS}. Moreover
\[ \sup_{k_2 \in \R} \norm{\widetilde{\alpha}(k_2) - \Id} = \sup_{k_2 \in \R} \norm{\alpha_0(k_2) - \alpha_1(k_2)} < 2 \]
and hence the Cayley transform (Proposition \ref{prop:Cayley}) defines a continuous periodic logarithm $\widetilde{h}$ for $\widetilde{\alpha}$ which satisfies also \eqref{eqn:hTRS}. 

Setting then
\[ \alpha_s(k_2) := \eps^{-1} \, \gamma(-k_2)^t \, \eps \, \eu^{\iu \, s \, \widetilde{h}(k_2)} \, \gamma(k_2) \]
yields a continuous homotopy between the families $\alpha_0$ and $\alpha_1$, satisfying Assumption \ref{assum:alpha} for all $s \in [0,1]$. The conclusion then follows from Lemma \ref{lemma:Ihomotopy}.
\end{proof}

To conclude this Subsection, we provide also a useful reinterpretation of the $\Z_2$ index $\mathcal{I}(\alpha)$ of Graf and Porta as an intersection number. 

\begin{proposition} \label{prop:intersect}
Let $f \colon [0,1/2] \to U(1)$ be a continuous map, and denote by $\mathcal{G}$ the corresponding graph. Also, let $\mathcal{L}$ be the union of the graphs of the functions $\lambda_i \colon [0,1/2] \to U(1)$, defined by the eigenvalues of the matching matrices $\set{\alpha(k_2)}_{k_2 \in \R}$. Assume that $\mathcal{G}$ and $\mathcal{L}$ intersect transversely finitely many times. Then $\mathcal{I}(\alpha)$ equals the parity of the number of crossings of $\mathcal{G}$ with $\mathcal{L}$.
\end{proposition}

The analogous statement for a \emph{constant} function $f$ was already used in the proof of the bulk-edge correspondence in \cite{GrafPorta13}. A sketch of the argument can be also found in \cite[proof of Prop.~5]{AvilaSchulz-BaldesVillegas-Blas12}.

\begin{proof}
The cardinality $\# \set{\mathcal{G} \cap \mathcal{L}}$ is well-defined in view of the transversality hypothesis. Clearly the number of crossings of the graph of $f$ with all the graphs of the different $\lambda_i$'s equals the sum of the number of crossings of the graph of $f$ with each individual $\lambda_i$, hence we may as well assume that there is only one such curve $\lambda$. Denote this number of crossings by $N$. Also, choose an ``even degenerate'' extension $\widetilde{\lambda}$ of $\lambda$ to $[-1/2,0]$, as in the construction leading to the definition of $\mathcal{I}(\alpha)$, and extend $f$ to $\widetilde{f}$ on $[-1/2,0]$ in order to make the joining function $\widetilde{f} \, \sharp \, f$ even. Then, since each intersection of the graph of $\widetilde{f}$ with the one of $\widetilde{\lambda}$ is counted twice (due to the ``even degeneracy'' property of $\widetilde{\lambda}$), the parity of $N$ equals the parity of the number of intersections of the graph of $\widetilde{f} \, \sharp \, f$ with the one of $\widetilde{\lambda} \, \sharp \, \lambda$, as periodic functions defined on $[-1/2,1/2]$ with values in $U(1)$.

Denote this number of intersections $\bmod\, 2$ as $I_2(\widetilde{f} \, \sharp \, f, \widetilde{\lambda} \, \sharp \, \lambda)$. Then $I_2(\widetilde{f} \, \sharp \, f, \widetilde{\lambda} \, \sharp \, \lambda)$ equals the $\bmod\, 2$ intersection number
\[ I_2 \left( (\widetilde{f} \, \sharp \, f) \times (\widetilde{\lambda} \, \sharp \, \lambda), \triangle_{U(1) \times U(1)} \right) \]
of the graph of the function $(\widetilde{f} \, \sharp \, f) \times (\widetilde{\lambda} \, \sharp \, \lambda) \colon [-1/2, 1/2] \to U(1) \times U(1)$ with the diagonal $\triangle_{U(1) \times U(1)} \subset U(1) \times U(1)$. In view of the results in \cite[Chap.~2, \S4]{GuilleminPollack74}, the above quantity is invariant with respect to homotopic changes of the functions it involves. In particular, since we have chosen an even extension of $f$, we have that $\widetilde{f} \, \sharp \, f$ has vanishing winding number, which in turn implies that it is homotopic to a constant map, say equal to $z \in U(1)$. Hence
\[ I_2(\widetilde{f} \, \sharp \, f, \widetilde{\lambda} \, \sharp \, \lambda) = I_2(\text{const}, \widetilde{\lambda} \, \sharp \, \lambda) = \# (\widetilde{\lambda} \, \sharp \, \lambda)^{-1} (\set{z}) \mod 2. \]
The right-hand side coincides with the $\bmod \: 2$ reduction of the degree of the map $\widetilde{\lambda} \, \sharp \, \lambda$ \cite[Chap.~3, \S3]{GuilleminPollack74}, which by definition is exactly the index $\mathcal{I}(\alpha)$. Putting this chain of equalities together allows us to conclude that $N \equiv \mathcal{I}(\alpha) \bmod 2$, as claimed.
\end{proof}

\subsection{Construction of the approximants}

We come back to the construction of continuous and periodic (respectively symmetric) Bloch frames for a $2$-dimensional family of orthogonal projectors $\set{P(\kk)}_{\kk \in \R^2}$. We saw in Section \ref{sec:d=d} how this problem can be restated in terms of $\alpha$ (see Problem \ref{pbl:alpha}), and is then equivalent to the construction of a family $\beta = \set{\beta(\kk)}_{\kk \in \R^2}$ which satisfies the following properties:
\begin{enumerate}[label=$(\mathrm{B}_\arabic*)$,ref=$(\mathrm{B}_\arabic*)$]
\item \label{item:beta_C0} the map $\R^2 \ni \kk \mapsto \beta(\kk) \in U(m)$ is continuous;
\item \label{item:beta_alpha} the map $(k_1, k_2) \mapsto \beta(k_1, k_2)$ is $\Z$-periodic in $k_2$ and satisfies
\[ \beta(k_1, k_2)^{-1} \, \alpha(k_2) \, \beta(k_1 + 1, k_2) = \Id \]
for all $(k_1, k_2) \in \R^2$;
\item \label{item:beta_TRS} for all $\kk \in \R^2$ 
\[ \beta(-\kk) = \eps^{-1} \, \overline{\beta(\kk)} \, \eps \]
holds.
\end{enumerate}

Moreover, owing to Proposition \ref{prop:AlmostLog}, the construction of $\beta$ with the above properties can be performed whenever $\alpha$ is sufficiently close to a family of unitary matrices which admits a ``good logarithm''. We then employ the strategy of constructing such approximants in order to prove the following result, which is just a reformulation of parts (\ref{item:BetaResult_d=2}) and (\ref{item:BetaResult_d=2_noTRS}) of Theorem \ref{thm:MainResults_Beta}.

\begin{theorem} \label{thm:beta}
Let $\set{\alpha(k_2)}_{k_2 \in \R}$ be as in Assumption \ref{assum:alpha}.
\begin{enumerate}
\setcounter{enumi}{-1}
 \item \label{item:betagap} There exists a family $\set{\beta(\kk)}_{\kk \in \R^2}$ of $m \times m$ unitary matrices satisfying \ref{item:beta_C0} and \ref{item:beta_alpha}. Moreover, this family can be explicitly constructed.
 \item \label{item:I=0->beta} Assume that $\mathcal{I}(\alpha) = 0 \in \Z_2$, where $\mathcal{I}(\alpha)$ is given by \eqref{eqn:rueda}. Then the above family $\beta$ can be constructed so that also \ref{item:beta_TRS} holds.
 \item \label{item:beta->I=0} Conversely, if there exists a family of unitary matrices $\beta$ satisfying \ref{item:beta_C0}, \ref{item:beta_alpha} and \ref{item:beta_TRS}, then $\mathcal{I}(\alpha) = 0 \in \Z_2$.
\end{enumerate}
\end{theorem} 

In the rest of the Subsection, we prove parts (\ref{item:betagap}) and (\ref{item:I=0->beta}) of Theorem \ref{thm:beta}; part (\ref{item:beta->I=0}) will be proved in the next Subsection. The strategy of the proof can be summarized as follows:
\begin{itemize}
\item we construct approximants for the family $\alpha$ with ``generic'' spectral properties (see Proposition \ref{prop:GenericForm});
\item we show that these spectral conditions allow the possibility to produce a branch cut in the resolvent of the approximating matrices (see Proposition \ref{prop:BranchCut}); this may require a topological obstruction to vanish;
\item we show that the branch cut chosen above defines in turn a ``good logarithm'' for the approximants (see Proposition \ref{prop:ApproxLog});
\item we appeal to Proposition \ref{prop:AlmostLog} to construct the required family $\beta$.
\end{itemize}
We detail the above steps in what follows.

\subsubsection{Removal of extra degeneracies} \label{sec:Removal}

The first step is to put the matrix $\alpha$ in a ``generic form'', in order to remove crossings of eigenvalues of $\alpha$ other than the ones prescribed by Kramers degeneracy (Lemma \ref{lemma:Kramers}). This will define the required approximants.

\begin{proposition} \label{prop:GenericForm}
Let $\set{\alpha(k_2)}_{k_2 \in \R}$ be as in Assumption \ref{assum:alpha}.
\begin{enumerate}
\item \label{item:alphagap} There exists a continuous and $\Z$-periodic family of unitary matrices $\set{\alpha\sub{gap}(k_2)}_{k_2 \in \R}$ with \emph{non-degenerate} even spectrum, $\sigma(\alpha\sub{gap}(k_2)) = \sigma(\alpha\sub{gap}(-k_2))$, and such that
\[ \sup_{k_2 \in \R} \norm{\alpha(k_2) - \alpha\sub{gap}(k_2)} < 2. \]
Moreover, $\alpha\sub{gap}$ can be explicitly constructed.
\item \label{item:alphagen} There exists a continuous and $\Z$-periodic family of unitary matrices $\set{\alpha\sub{gen}(k_2)}_{k_2 \in \R}$, which is \emph{real analytic} around half-integer values of $k_2$, such that
\[ \eps \, \alpha\sub{gen}(k_2) = \alpha\sub{gen}(-k_2)^t \, \eps, \]
with \emph{non-degenerate} spectrum away from half-integer values of $k_2$, \emph{doubly degenerate} spectrum at half-integer values of $k_2$, and such that
\[ \sup_{k_2 \in \R} \norm{\alpha(k_2) - \alpha\sub{gen}(k_2)} < 2. \]
Moreover, $\alpha\sub{gen}$ can be explicitly constructed.
\end{enumerate}
\end{proposition}

The proof of Proposition \ref{prop:GenericForm} is rather technical, and is deferred to Appendix \ref{sec:ExtraDegen} (see Propositions \ref{propo-dec-1} and \ref{propo-dec-2}). The existence results of approximants with the required spectral properties could be argued by means of genericity arguments \cite{vonNeumannWigner29}, but we seek an explicit construction, in order to make the argument as ``algorithmic'' as possible, in view of possible numerical implementations.

Figure \ref{fig:alpha1} depicts the spectrum of a family of matrices $\alpha$ as in Assumption \ref{assum:alpha}, while Figure~\ref{fig:alpha2} illustrates the ``generic'' spectrum of $\alpha\sub{gen}$, in the half-period $[0,1/2]$. The family $\alpha\sub{gap}$ is obtained by further lifting the double (Kramers) degeneracies of $\alpha\sub{gen}$ at $k_2=0$ and $k_2=1/2$, in a way that preserves the symmetry of the spectrum under the exchange $k_2 \to -k_2$ (Proposition \ref{propo-dec-2}).

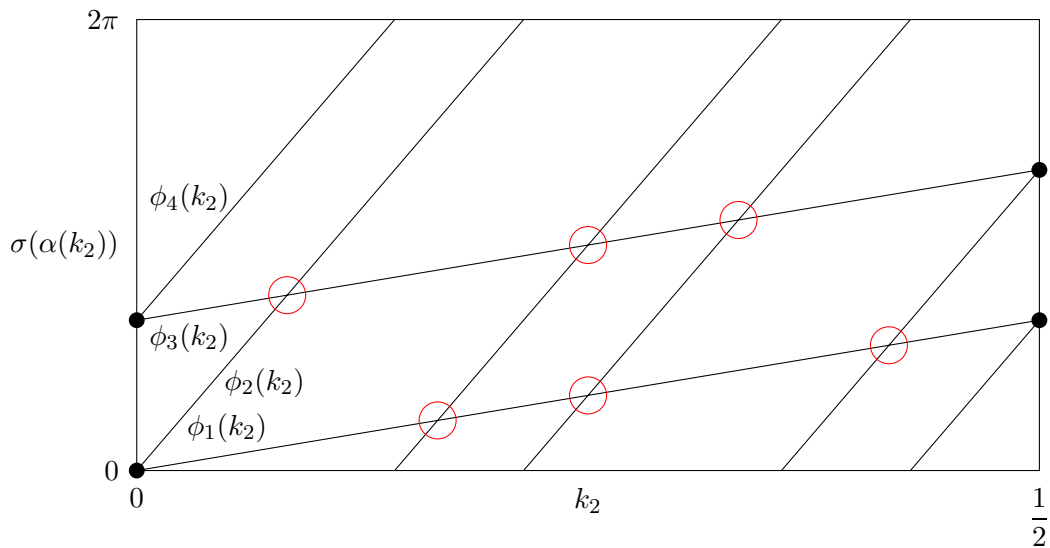
\begin{figure}[ht]
\centering
\begin{tikzpicture}
\draw (0,0) -- (12,0) -- (12,6) -- (0,6) -- cycle;
\draw (0,0) -- (12,2)
      (0,0) -- (36/7,6)
      (36/7,0) -- (72/7,6)
      (72/7,0) -- (12,2)
      (0,2) -- (12,4)
      (0,2) -- (24/7,6)
      (24/7,0) -- (60/7,6)
      (60/7,0) -- (12,4);
\draw [red] (4,2/3) circle (7pt)
            (10,5/3) circle (7pt)
            (6,1) circle (7pt)
            (2,7/3) circle (7pt)
            (8,10/3) circle (7pt)
            (6,3) circle (7pt);
\fill (0,0) circle (3pt)
      (0,2) circle (3pt)
      (12,2) circle (3pt)
      (12,4) circle (3pt);
\draw (1.2,.9) node [anchor = north] {$\phi_1(k_2)$}
      (1.7,1.5) node [anchor = north] {$\phi_2(k_2)$}
      (.7,2.1) node [anchor = north] {$\phi_3(k_2)$}
      (.7,3.3) node [anchor = south] {$\phi_4(k_2)$}
      (0,-.1) node [anchor = north] {$0$}
      (6,-.1) node [anchor = north] {$k_2$}
      (12,-.1) node [anchor = north] {$\dfrac{1}{2}$}
      (-.1,0) node [anchor = east] {$0$}
      (-.1,3) node [anchor = east] {$\sigma(\alpha(k_2))$}
      (-.1,6) node [anchor = east] {$2 \pi$};
\end{tikzpicture}
\caption{The spectrum of $\alpha(k_2)$. In each red circle, an eigenvalue crossing occurs.}
\label{fig:alpha1}
\end{figure}

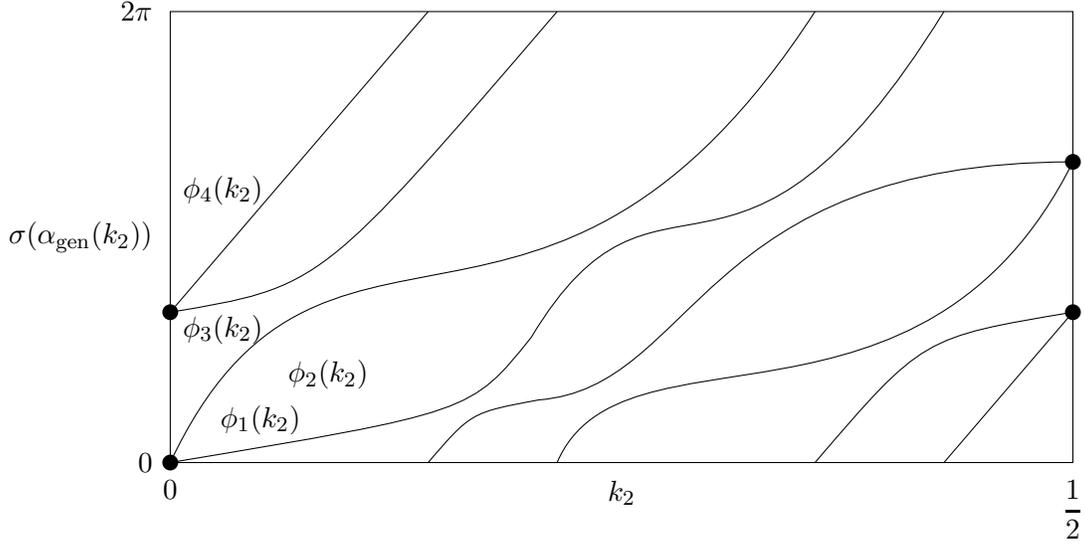
\begin{figure}[ht]
\centering
\begin{tikzpicture}
\draw (0,0) -- (12,0) -- (12,6) -- (0,6) -- cycle;
\draw (0,2) -- (24/7,6)
      (0,2) .. controls (2,7/3) .. (36/7,6)
      (0,0) .. controls (2,4.25) and (5,.5) .. (60/7,6)
      (0,0) .. controls (4,2/3) .. (4.8,5/3) .. controls (6.5,4.5) and (7.5,1.5) .. (72/7,6)
      (24/7,0) .. controls (4,2/3) .. (4.9,5/6) .. controls (7,1) and (7,4) .. (12,4)
      (36/7,0) .. controls (6,2) and (10,0) .. (12,4)
      (60/7,0) .. controls (10,5/3) .. (12,2)
      (72/7,0) -- (12,2);
\fill (0,0) circle (3pt)
      (0,2) circle (3pt)
      (12,2) circle (3pt)
      (12,4) circle (3pt);
\draw (1.2,.9) node [anchor = north] {$\phi_1(k_2)$}
      (2.1,1.5) node [anchor = north] {$\phi_2(k_2)$}
      (.7,2.1) node [anchor = north] {$\phi_3(k_2)$}
      (.7,3.3) node [anchor = south] {$\phi_4(k_2)$}
      (0,-.1) node [anchor = north] {$0$}
      (6,-.1) node [anchor = north] {$k_2$}
      (12,-.1) node [anchor = north] {$\dfrac{1}{2}$}
      (-.1,0) node [anchor = east] {$0$}
      (-.1,3) node [anchor = east] {$\sigma(\alpha\sub{gen}(k_2))$}
      (-.1,6) node [anchor = east] {$2 \pi$};
\end{tikzpicture}
\caption{The spectrum of $\alpha\sub{gen}(k_2)$ is non-degenerate for $k_2 \in (0,1/2)$, and only doubly degenerate at $0$ and $1/2$.}
\label{fig:alpha2}
\end{figure}

\subsubsection{The branch cut}

We now define a ($k_2$-dependent) branch cut to compute the logarithms of $\alpha\sub{gap}$ and $\alpha\sub{gen}$.

\begin{proposition} \label{prop:BranchCut}
Let $\set{\alpha(k_2)}_{k_2 \in \R}$ be a family of unitary matrices as in Assumption \ref{assum:alpha}.
\begin{enumerate}
\item \label{item:alphagap_cut} Let $\set{\alpha\sub{gap}(k_2)}_{k_2 \in \R}$ be as in Proposition \ref{prop:GenericForm}(\ref{item:alphagap}). Then one can construct a continuous, $\Z$-periodic and even function $\phi\sub{gap} \colon \R \to \R$ such that $\eu^{\iu \, \phi\sub{gap}(k_2)}$ lies in the resolvent set of $\alpha\sub{gap}(k_2)$ for all $k_2 \in \R$.
\item \label{item:alphagen_cut} Assume that $\mathcal{I}(\alpha) = 0 \in \Z_2$, where $\mathcal{I}(\alpha)$ is defined in \eqref{eqn:rueda}. Let $\set{\alpha\sub{gen}(k_2)}_{k_2 \in \R}$ be as in Proposition \ref{prop:GenericForm}(\ref{item:alphagen}). Then one can construct a continuous, $\Z$-periodic and even function $\phi\sub{gen} \colon \R \to \R$ such that $\eu^{\iu \, \phi\sub{gen}(k_2)}$ lies in the resolvent set of $\alpha\sub{gen}(k_2)$ for all $k_2 \in \R$.
\end{enumerate}
\end{proposition}
\begin{proof}
We first prove \eqref{item:alphagap_cut}. Since $\alpha\sub{gap}(k_2)$ has non-degenerate spectrum for all $k_2 \in \R$, we can choose continuous arguments for its eigenvalues:
\[ \sigma(\alpha\sub{gap}(k_2)) = \set{\eu^{\iu \, \phi_j(k_2)}}_{j =1, \ldots, m}, \quad \phi_j \colon \R \to \R \text{ continuous}. \]
Such arguments can be also chosen to be even functions of $k_2$, due to the symmetry of the spectrum of $\alpha\sub{gap}$. Let now
\[ \phi\sub{gap}(k_2) := \frac{\phi_1(k_2) + \phi_2(k_2)}{2}, \quad k_2 \in \R. \]
Then $\eu^{\iu \, \phi\sub{gap}(k_2)}$ indeed lies in the resolvent set of $\alpha\sub{gap}(k_2)$ and satisfies all the other required properties.

The proof of \eqref{item:alphagen_cut} is more involved. Let $a \in \R$ be such that all the arguments of the $n$ distinct eigenvalues of $\alpha\sub{gen}(0)$ lie in $(a, a+2\pi)$. To streamline the notation, we assume without loss of generality that $a=0$. By continuity, from each of the distinct eigenvalues of $\alpha\sub{gen}(0)$ stem two eigenvalues for small positive $0< k_2 \le \epsilon_0$, since $\alpha\sub{gen}(k_2)$ has non-degenerate spectrum in $(0,1/2)$. Moreover, if $\epsilon_0$ is small enough, one can choose the arguments of these eigenvalues to lie in $[0,2\pi)$ for all $k_2 \in [0, \epsilon_0]$ (\ie the eigenvalues do not wind around the unit circle in this interval); they will also be distinct in view of the non-degeneracy of the spectrum of $\alpha\sub{gen}(k_2)$. We then label these eigenvalues $\set{\lambda_i(k_2) = \eu^{\iu \phi_i(k_2)}}_{1 \le i \le m} = \sigma(\alpha\sub{gen}(k_2))$, for $k_2 \in [0, \epsilon_0]$, with respect to the increasing order of their arguments: $0 \le \phi_1(k_2) < \cdots < \phi_m(k_2) < 2 \pi$. This labelling can be employed in the whole $[0, 1/2]$, and leads to a continuous choice of the definition of the arguments $\phi_i(k_2)$; the order condition may of course be lost. 

We fix a small $\phi_{1/2} > 0$ and choose $\epsilon_1 > 0$ so small that the following hold. With respect to a \emph{possibly different} increasing labelling, the arguments $\set{\phi_j(k_2)}_{1 \le j \le m}$ of the eigenvalues of $\alpha\sub{gen}(k_2)$ can be chosen in $[0, 2\pi)$ for $k_2 \in [1/2-\epsilon_1, 1/2]$, and depend analytically on $k_2$ on the same interval: this can be achieved in view of the (local) real analyticity of $\alpha\sub{gen}(k_2)$ around half-integer values of $k_2$, which implies by the Analytic Rellich Theorem \cite[Sec.~2.7.4]{CorneanHerbstNenciu15} that also its eigenvalues can be chosen to be real analytic. We also require that 
\begin{equation} \label{eqn:at1/2}
\phi_{2j+1}(k_2) - \phi_{2j}(k_2) > \frac{\phi_{1/2}}{2}, \quad j \in \set{1, \ldots, n}, \quad k_2 \in \left[\frac{1}{2}-\epsilon_1, \frac{1}{2}\right], \quad \phi_{2n+1}(k_2) := 2 \pi.
\end{equation}
Moreover, we require that the gaps $\Delta_j(k_2) := \phi_{2j}(k_2) - \phi_{2j-1}(k_2)$, for $j \in \set{1, \ldots, n}$, be strictly monotonic decreasing in the same interval (with $\Delta_j(1/2) = 0$ due to Kramers degeneracy of $\alpha\sub{gen}(1/2)$): this is possible again in view of the real analyticity condition, at the price of taking possibly a smaller $\epsilon_1$. The region
\[ \Gamma_j := \bigcup_{k_2 \in [1/2-\epsilon_1, 1/2]} \left[\phi_{2j-1}(k_2),  \phi_{2j}(k_2)\right] \]
between two subsequent eigenvalues of $\alpha\sub{gen}(k_2)$ joining in the same eigenvalue of $\alpha\sub{gen}(1/2)$ will be called the \emph{$j$-th cusp}. Notice that there are exactly $n$ cusps and that they are all disjoint, if $\epsilon_1$ is small enough, because $\alpha\sub{gen}(1/2)$ has exactly $n$ distinct eigenvalues and $\alpha\sub{gen}(k_2)$ has non-degenerate spectrum for $k_2 \in (0,1/2)$. The width of the $j$-th cusp is given exactly by the value of the function $\Delta_j$, and is then shrinking to $0$.

Let $\phi_0>0$ be such that $\phi_m(k_2) < \phi_m(k_2) + \phi_0 < 2 \pi$ for all $k_2 \in [0, \epsilon_0]$. Also let $g > 0$ denote the minimal gap between the eigenvalues of $\alpha\sub{gen}(k_2)$ for $k_2 \in [\epsilon_0, 1/2-\epsilon_1]$. Define
\[ \phi\sub{gen}(k_2) := \phi_m(k_2) + b_0, \quad \text{where} \quad b_0 := \min \set{\phi_0, \frac{\phi_{1/2}}{3}, \frac{g}{3}}. \]
The choice of the curve lying in the resolvent of $\alpha\sub{gen}$ will then be $f(k_2) := \eu^{\iu \, \phi\sub{gen}(k_2)}$, for $k_2 \in [0, 1/2]$ (Figure \ref{fig:cut}). This curve will indeed not intersect the spectrum of $\alpha\sub{gen}(k_2)$ for $k_2 \in [0, \epsilon_0]$ because its argument is greater than the arguments of any of the eigenvalues (the difference with the largest of them being at least $\phi_0$), and neither will it intersect the spectrum in $[\epsilon_0, 1/2-\epsilon_1]$ because of the gap condition.

\begin{figure}[ht]
\centering
\begin{tikzpicture}
\draw (0,0) -- (12,0) -- (12,6) -- (0,6) -- cycle;
\draw (0,2) -- (24/7,6)
      (0,2) .. controls (2,7/3) .. (36/7,6)
      (0,0) .. controls (2,4.25) and (5,.5) .. (60/7,6)
      (0,0) .. controls (4,2/3) .. (4.8,5/3) .. controls (6.5,4.5) and (7.5,1.5) .. (72/7,6)
      (24/7,0) .. controls (4,2/3) .. (4.9,5/6) .. controls (7,1) and (7,4) .. (12,4)
      (36/7,0) .. controls (6,2) and (10,0) .. (12,4)
      (60/7,0) .. controls (10,5/3) .. (12,2)
      (72/7,0) -- (12,2);
\draw [dashed] (0,2+1/4) -- (24/7-3/14,6)
               (24/7-3/14,0) .. controls (4-3/14,2/3) .. (4.9-3/14,5/6+1/12) .. controls (7-3/14,1) and (7-3/14,4) .. (12,4+1/8);
\fill (0,0) circle (3pt)
      (0,2) circle (3pt)
      (12,2) circle (3pt)
      (12,4) circle (3pt);
\draw (1.2,.9) node [anchor = north] {$\phi_1(k_2)$}
      (2.1,1.5) node [anchor = north] {$\phi_2(k_2)$}
      (.7,2.1) node [anchor = north] {$\phi_3(k_2)$}
      (.7,3.5) node [anchor = south] {$\phi_4(k_2)$}
      (0,-.1) node [anchor = north] {$0$}
      (6,-.1) node [anchor = north] {$k_2$}
      (12,-.1) node [anchor = north] {$\dfrac{1}{2}$}
      (-.1,0) node [anchor = east] {$0$}
      (-.1,3) node [anchor = east] {$\sigma(\alpha\sub{gen}(k_2))$}
      (-.1,6) node [anchor = east] {$2 \pi$};
\end{tikzpicture}
\caption{The logarithm branch cut (dashed line).}
\label{fig:cut}
\end{figure}
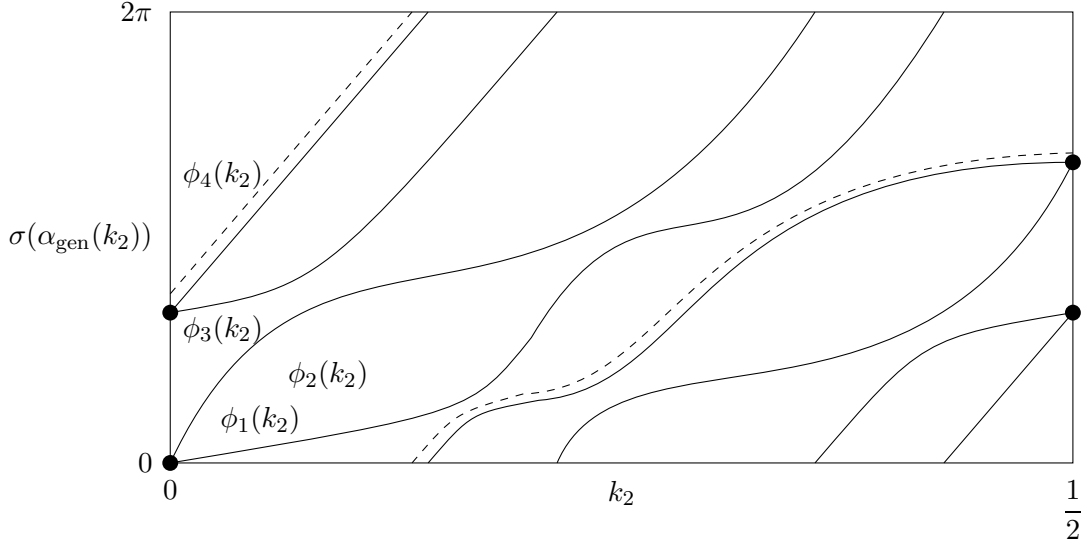

We claim that the curve traced by $f(k_2)$ for $k_2 \in [1/2-\epsilon_1, 1/2]$ still continues not to intersect the spectrum. Indeed, if $f(1/2-\epsilon_1)$ is outside all the cusps, then its graph is parallel to one of the ``even'' eigenvalues $\phi_{2j}$, and will end up on top of $\phi_{2j}(1/2)$ without ever touching the spectrum. If instead $f(1/2-\epsilon_1)$ were inside one of the cusps, say in the $j$-th one, then the graph of $f$ would be parallel to the one of $\phi_{2j-1}$, namely
\[ f(k_2) = \phi_{2j-1}(k_2) + b_0. \]
Define $\widetilde{\Delta}(k_2) := \phi_{2j}(k_2) - f(k_2) = \Delta_j(k_2) - b_0$. Since $\Delta_j$ is by assumption monotone decreasing, then so is $\widetilde{\Delta}$; moreover, $\widetilde{\Delta}$ changes sign in the interval $[1/2-\epsilon_1,1/2]$, passing from being positive to negative. Notice that the choice of $b_0$ also guarantees that the curve traced by the graph of $f$ does not enter $\Gamma_{j+1}$ after exiting $\Gamma_j$, in view of \eqref{eqn:at1/2} and the fact that $b_0 \le \phi_{1/2}/3$. It follows that $f(k_2)$ crosses the spectrum of $\alpha(k_2)$ (more precisely the eigenvalue $\phi_{2j}(k_2)$) exactly once (Figure \ref{fig:cut1}). This would produce an example of a continuous curve intersecting the graphs of the eigenvalues of $\alpha\sub{gen}(k_2)$ an odd number of times: this contradicts the hypothesis that $\mathcal{I}(\alpha) = \mathcal{I}(\alpha\sub{gen}) = 0 \in \Z_2$ in view of Proposition \ref{prop:intersect}. Notice that indeed $\mathcal{I}(\alpha\sub{gen}) = \mathcal{I}(\alpha)$ in view of Proposition \ref{prop:closeness}, because the two families are uniformly norm-close to each other.

\begin{figure}[ht]
\centering
\begin{tikzpicture}
\path [name path = spectrum1] (0,1) .. controls (3,1) and (4,4) .. (12,4);
\path [name path = spectrum2] (2,0) .. controls (3,1) .. (12,4);
\path [name path = spectrum3] (8,0) .. controls (9,1) .. (12,1);
\path [name path = spectrum4] (10,0) -- (12,1);
\path [name path = epsilon] (10.5,0) -- (10.5,6);
\draw (0,0) -- (12,0) -- (12,6) -- (0,6) -- cycle;
\fill [color = gray!30, name intersections = {of = spectrum1 and epsilon, by = e1}, name intersections = {of = spectrum2 and epsilon, by = e2}, name intersections = {of = spectrum3 and epsilon, by = e3}, name intersections = {of = spectrum4 and epsilon, by = e4}] (12,4) -- (e1) -- (e2) -- cycle
(12,1) -- (e3) -- (e4) -- cycle;
\draw (0,4) -- (2,6)
      (0,4) .. controls (6,4) .. (8,6)
      (0,1) -- (10,6);
\draw (0,1) .. controls (3,1) and (4,4) .. (12,4);
\draw (2,0) .. controls (3,1) .. (12,4);
\draw (8,0) .. controls (9,1) .. (12,1);
\draw (10,0) -- (12,1);
\draw [dotted] (10.5,0) -- (10.5,6);
\draw [dashed] (0,4+1/3) -- (2-1/3,6)
      [name path = cut] (2-1/3,0) .. controls (3-1/3,1+1/3) .. (12,4+1/3);
\draw [red, name intersections = {of = spectrum1 and cut, by = c1}] (c1) circle (7pt);
\fill (0,4) circle (3pt)
      (0,1) circle (3pt)
      (12,4) circle (3pt)
      (12,1) circle (3pt);
\draw (0,-.1) node [anchor = north] {$0$}
      (6,-.1) node [anchor = north] {$k_2$}
      (12,-.1) node [anchor = north] {$\dfrac{1}{2}$}
      (10.5,-.1) node [anchor = north] {$\dfrac{1}{2}-\epsilon_1$}
      (-.1,0) node [anchor = east] {$0$}
      (-.1,3) node [anchor = east] {$\sigma(\alpha\sub{gen}(k_2))$}
      (-.1,6) node [anchor = east] {$2 \pi$};
\end{tikzpicture}
\caption{The logarithm branch cut (dashed line) must cross the eigenvalues when $\mathcal{I}(\alpha\sub{gen}) = 1$. The shaded area are the cusps around the doubly-degenerate eigenvalues of $\alpha\sub{gen}(1/2)$.}
\label{fig:cut1}
\end{figure}
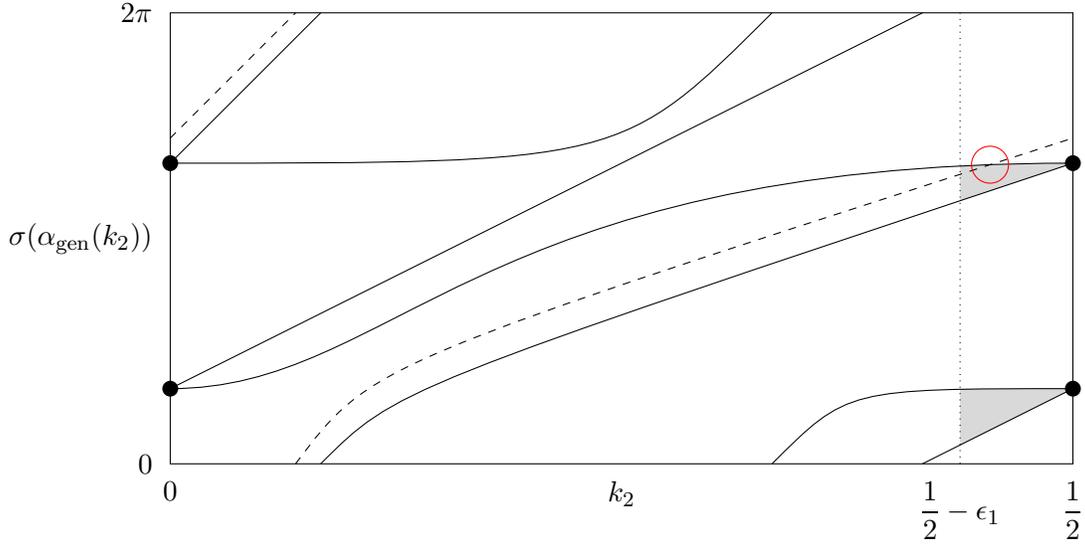

Since the spectrum of $\alpha\sub{gen}(k_2)$ is even-symmetric in $k_2$ (as follows from \ref{item:alpha_TRS}), we can extend $f(k_2)$ to negative $k_2 \in [-1/2,0]$ by taking the even extension, and its graph will still not intersect the spectrum of $\alpha\sub{gen}$. The periodicity of the spectrum of $\alpha\sub{gen}$ implies that the graph of the $\Z$-periodic extension of $f$ (that is, of $\phi\sub{gen}$) to $\R$ will still always lie in the resolvent set. This concludes the proof.
\end{proof}

\subsubsection{Logarithms for the approximants}

We conclude the argument leading to the proof of parts (\ref{item:betagap}) and (\ref{item:I=0->beta}) from Theorem \ref{thm:beta} (or equivalently of parts (\ref{item:Result_d=2}\ref{item:Result_d=2_I=0}) and (\ref{item:Result_d=2_noTRS}) from Theorem \ref{thm:MainResults_BB}, in view of Proposition \ref{prop:rotate}) by using the Cayley transform method. 

\begin{proposition} \label{prop:ApproxLog}
Let $\set{\alpha\sub{gap}(k_2)}_{k_2 \in \R}$ and $\set{\alpha\sub{gen}(k_2)}_{k_2 \in \R}$ be as in Proposition \ref{prop:GenericForm}. Then it is possible to construct continuous families $\set{h\sub{gap}(k_2)}_{k_2 \in \R}$ and $\set{h\sub{gen}(k_2)}_{k_2 \in \R}$ of matrices, with
\begin{gather*}
h\sub{gap}(k_2) = h\sub{gap}(k_2)^* = h\sub{gap}(k_2+1), \\
h\sub{gen}(k_2) = h\sub{gen}(k_2)^* = h\sub{gen}(k_2+1), \quad \eps \, h\sub{gen}(k_2) = h\sub{gen}(-k_2)^t \, \eps,
\end{gather*}
such that
\[ \alpha\sub{gap}(k_2) = \eu^{\iu \, h\sub{gap}(k_2)} \quad \text{and} \quad \alpha\sub{gen}(k_2) = \eu^{\iu \, h\sub{gen}(k_2)}. \]
\end{proposition}
\begin{proof}
Define for $k_2 \in \R$
\[
\widetilde{\alpha}\sub{gap}(k_2) := \eu^{\iu \, (\pi - \phi\sub{gap}(k_2))} \, \alpha\sub{gap}(k_2), \quad
\widetilde{\alpha}\sub{gen}(k_2) := \eu^{\iu \, (\pi - \phi\sub{gen}(k_2))} \, \alpha\sub{gen}(k_2),
\]
with $\phi\sub{gap}$ and $\phi\sub{gen}$ as in Proposition \ref{prop:BranchCut}. Then $-1$ lies in the resolvent set of $\widetilde{\alpha}\sub{gap}(k_2)$ and $\widetilde{\alpha}\sub{gen}(k_2)$ for all $k_2 \in \R$, and moreover $\eps \, \widetilde{\alpha}\sub{gen}(k_2) = \widetilde{\alpha}\sub{gen}(-k_2)^t \, \eps$ because $\phi\sub{gen}(k_2) = \phi\sub{gen}(-k_2)$. Thus, Proposition \ref{prop:Cayley} can be applied to these two families of matrices, and yields
\begin{align*}
\widetilde{\alpha}\sub{gap}(k_2) & = \eu^{\iu \, \widetilde{h}\sub{gap}(k_2)}, \quad \widetilde{h}\sub{gap}(k_2) = \widetilde{h}\sub{gap}(k_2)^* = \widetilde{h}\sub{gap}(k_2+1), \\
\widetilde{\alpha}\sub{gen}(k_2) & = \eu^{\iu \, \widetilde{h}\sub{gen}(k_2)}, \quad \widetilde{h}\sub{gen}(k_2) = \widetilde{h}\sub{gen}(k_2)^* = \widetilde{h}\sub{gen}(k_2+1), \quad \eps \, \widetilde{h}\sub{gen}(k_2) = \widetilde{h}\sub{gen}(-k_2)^t \, \eps.
\end{align*}

Setting
\[ h\sub{gap}(k_2)  := (-\pi + \phi\sub{gap}(k_2)) \Id + \widetilde{h}\sub{gap}(k_2), \quad
h\sub{gen}(k_2)  := (-\pi + \phi\sub{gen}(k_2)) \Id + \widetilde{h}\sub{gen}(k_2),\]
we obtain
\[ \alpha\sub{gap}(k_2) = \eu^{\iu \, h\sub{gap}(k_2)}, \quad \alpha\sub{gen}(k_2) = \eu^{\iu \, h\sub{gen}(k_2)}, \]
with logarithms which satisfy all the required properties.
\end{proof}

\subsection{$\mathcal{I}$ as a topological obstruction}

In the previous Subsection we showed how the condition $\mathcal{I}(\alpha) = 0 \in \Z_2$, with $\alpha$ the family of matching matrices and $\mathcal{I}(\alpha)$ as in \eqref{eqn:rueda}, is sufficient for the existence of a family of unitary matrices $\beta$ satisfying \ref{item:beta_C0}, \ref{item:beta_alpha} and \ref{item:beta_TRS}. We show now that the condition is also necessary. This concludes the proof of Theorem \ref{thm:beta} (or equivalently of Theorem \ref{thm:MainResults_BB}), characterizing completely the condition $\mathcal{I}=0 \in \Z_2$ as the (topological) obstruction to a continuous and symmetric Bloch frame in $d=2$.

\begin{proof}[Proof of Theorem \ref{thm:beta}(\ref{item:beta->I=0})]
In view of Remark \ref{rmk:beta0}, we can without loss of generality assume that $\beta(0,k_2) \equiv \Id$ for all $k_2 \in \R$. Evaluating \ref{item:beta_alpha} at $k_1=-1/2$ and using \ref{item:beta_TRS} yields
\[ \alpha(k_2) = \beta(-1/2, k_2) \, \beta(1/2, k_2)^{-1} = \eps^{-1} \, \overline{\beta(1/2,-k_2)} \, \eps \, \beta(1/2,k_2)^{-1}. \]
By setting 
\[ \gamma(k_2) := \beta(1/2, k_2)^{-1} \]
the above can then be rewritten as
\begin{equation} \label{eqn:gamma(k2)}
\alpha(k_2) = \eps^{-1} \, \gamma(-k_2)^t \, \eps \, \gamma(k_2).
\end{equation}
Notice that this definition agrees at $k_2=0$ and $k_2=1/2$ with the matrices $\gamma(0)$ and $\gamma(1/2)$ appearing in \eqref{eqn:gammas}, and gives an alternative proof of Lemma \ref{lemma:factor} under the additional hypothesis of the existence of $\beta$.

Define now $B(k_1, k_2) := \det \beta(k_1, k_2)^{-1}$ and 
\[ d(k_1) := \frac{1}{2 \pi \iu} \int_{-1/2}^{1/2} B(k_1, k_2)^{-1} \, B'(k_1, k_2) \, \di k_2, \quad k_1 \in \R. \]
The function $d$ is integer-valued, because for fixed $k_1$ it computes the degree of the determinant of $\beta(k_1, \cdot)^{-1}$, and continuous in $k_1$, as $\beta(\kk)$ depends continuously on $\kk$: as such, it is constant. Since $\beta(0,k_2) \equiv \Id$, we have that $d(0)=0$. On the other hand, in view of our definition of $\gamma$ we have that $d(1/2)$ computes the degree of $\det \gamma$, as in \eqref{eqn:degree}. Since the latter appears on the right-hand side of formula \eqref{eqn:IntegralRueda_c} for $\mathcal{I}(\alpha)$, we conclude that $\mathcal{I}(\alpha) = 0 \in \Z_2$, as wanted.
\end{proof}

\subsection{$\mathcal{I}$ as a complete homotopy invariant}  \label{sec:Itopology}

As a corollary of Theorem \ref{thm:beta}(\ref{item:I=0->beta}), we can prove that the Graf--Porta $\Z_2$ index completely characterizes families of matrices as in Assumption \ref{assum:alpha} up to homotopy. The next statement is just a reformulation of Theorem \ref{thm:MainResults_GPhomotopy}.

\begin{theorem} \label{thm:homotopyI}
Let $\alpha_0=\set{\alpha_0(k_2)}_{k_2 \in \R}$ and $\alpha_1=\set{\alpha_1(k_2)}_{k_2 \in \R}$ be as in Assumption \ref{assum:alpha}. Then there exists a continuous homotopy  $\alpha_s=\set{\alpha_s(k_2)}_{k_2 \in \R}$, $s \in [0,1]$, which connects $\alpha_0$ and $\alpha_1$ via families as in Assumption \ref{assum:alpha} if and only if
\begin{equation} \label{eqn:SameIndex}
\mathcal{I}(\alpha_0) = \mathcal{I}(\alpha_1) \in \Z_2.
\end{equation}
\end{theorem}
\begin{proof}
The ``only if'' statement was already proved in Lemma \ref{lemma:Ihomotopy} (and used in the proof of Proposition \ref{prop:BranchCut}, leading to Theorem \ref{thm:beta}(\ref{item:I=0->beta})). Thus we only need to prove the converse statement.

Write
\[ \alpha_0(k_2) = \eps^{-1} \, \gamma_0(-k_2)^t \, \eps \, \gamma_0(k_2), \quad \alpha_1(k_2) = \eps^{-1} \, \gamma_1(-k_2)^t \, \eps \, \gamma_1(k_2), \]
as in Lemma \ref{lemma:factor}. Define
\[ \alpha(k_2) := \eps^{-1} \, \overline{\gamma_0(-k_2)} \, \eps \, \alpha_1(k_2) \, \gamma_0(k_2)^*. \]
Then an easy computation leads to
\[ \alpha(k_2) = \eps^{-1} \, \gamma(-k_2)^t \, \eps \, \gamma(k_2), \quad \text{with} \quad \gamma(k_2) := \gamma_1(k_2) \, \gamma_0(k_2)^*. \]
It follows again from Lemma \ref{lemma:factor} that $\alpha$ defines a continuous, $\Z$-periodic and TRS family of unitary matrices. We compute its $\Z_2$ index by means of \eqref{eqn:IntegralRueda_c}: we obtain
\[ \mathcal{I}(\alpha) \equiv \deg([\det \gamma]) = \deg([\det \gamma_1]) - \deg([\det \gamma_0]) \equiv \mathcal{I}(\alpha_1) - \mathcal{I}(\alpha_0) \equiv 0 \bmod 2 \]
in view of the hypothesis \eqref{eqn:SameIndex}.

It now follows from Theorem \ref{thm:beta}(\ref{item:I=0->beta}) that there exists a family $\beta$ of unitary matrices which satisfies \ref{item:beta_C0}, \ref{item:beta_alpha} and \ref{item:beta_TRS} with respect to $\alpha$. In particular, evaluating \ref{item:beta_alpha} at $k_1=-1/2$ we obtain
\[ \alpha(k_2) = \beta(-1/2,k_2) \, \beta(1/2, k_2)^{-1}. \]
It is then immediately realized that
\[ \alpha_s(k_2) := \eps^{-1} \, \gamma_0(-k_2)^t \, \eps \, \beta(-s/2,k_2) \, \beta(s/2, k_2)^{-1} \, \gamma_0(k_2) \]
gives a continuous homotopy between $\alpha_0$ and $\alpha_1$ with all the required properties.
\end{proof}

\goodbreak


\section{More on the $\Z_2$ invariant} \label{sec:Z2}

This Section is devoted to a more thorough investigation of the r\^ole of the $\Z_2$ invariant which appeared in the construction of continuous and symmetric Bloch frames for a $2$-dimensional family of projectors $\set{P(\kk)}_{\kk \in \R^2}$ as in Assumption \ref{assum:proj}, and in its reformulation in terms of the family $\alpha$ of matching matrices as in Assumption \ref{assum:alpha}. We saw indeed that the existence of such frames is equivalent to the condition $\mathcal{I} = 0 \in \Z_2$, with $\mathcal{I} = \mathcal{I}(\alpha)$ as in \eqref{eqn:rueda}. This condition seems however to be dependent on the choice of a continuous symmetric Bloch frame for the $1$-dimensional restriction $\set{P(0,k_2)}_{k_2 \in \R}$, which then defines the associated family of matching matrices by \eqref{eqn:alpha_vs_Tk2}.

To show that, instead, the Graf--Porta index $\mathcal{I} \in \Z_2$ can be recognized as a true topological invariant of the family of projectors, we prove in Theorem \ref{thm:GP_vs_FMP} that it agrees numerically with the $\Z_2$ invariant $\delta$ defined in \cite{FiorenzaMonacoPanati16} exactly as the invariant quantifying the topological obstruction to the existence of a continuous and symmetric Bloch frame in $d=2$. The index $\mathcal{I}$ hence ``inherits'' from $\delta$ all its properties, characterizing in particular the isomorphism class of its associated Bloch bundle. In turn, by means of the relation between $\delta$ and $\mathcal{I}$ we will be able to prove a formula for the $\Z_2$ invariant (compare \eqref{eqn:deltaGeom}) which depends purely on geometric terms, namely the \emph{Berry connection} and \emph{Berry curvature} associated to the family of projectors.

\subsection{The Fiorenza--Monaco--Panati invariant $\delta$} \label{sec:FMP}

For the reader's convenience, we briefly recall how the $\Z_2$ invariant $\delta$ introduced in \cite{FiorenzaMonacoPanati16} is defined.

Given a family of projectors $\set{P(\kk)}_{\kk \in \R^2}$ satisfying Assumption \ref{assum:proj}, the strategy employed in \cite{FiorenzaMonacoPanati16} to construct a continuous symmetric Bloch frame consists in considering any continuous frame $\Psi(\kk)$, defined for $\kk$ in the \emph{effective unit cell}
\[ \Beff := \set{\kk = (k_1, k_2) \in \R^2 : 0 \le k_1 \le \frac{1}{2}, \: - \frac{1}{2} \le k_2 \le \frac{1}{2}}, \]
and try to modify it into a frame $\Phi(\kk)$ which satisfies certain conditions on the boundary $\partial \Beff$ (compare (V) and (E) in \cite[Prop.~1]{FiorenzaMonacoPanati16}), that allow to impose TRS and periodicity to obtain a continuous, symmetric frame on the whole $\R^2$. This modification is performed on high-symmetry points (\emph{vertices}) of the effective unit cell, and then extended to the \emph{edges} on the boundary of $\Beff$ which connect them. Such a frame can then be constructed on $\partial \Beff$, but its extension to the interior of $\Beff$ is in general topologically obstructed. To compute such obstruction, one defines the unitary-matrix-valued map $U \colon \partial \Beff \to U(m)$ by setting
\begin{equation} \label{eqn:healer}
\Phi(\kk) = \Psi(\kk) \act U(\kk), \quad \kk \in \partial \Beff.
\end{equation}
Thus by definition $U(\kk)$ maps the input frame $\Psi(\kk)$ to the modified, symmetric frame $\Phi(\kk)$ for $\kk \in \partial \Beff$. The extension of the frame $\Phi$ on the inside of $\Beff$ can be then performed in a continuous way and preserving all the symmetries if and only if
\begin{equation} \label{eqn:delta}
\delta := \deg([\det U]) \mod 2
\end{equation}
vanishes (see \cite[Thm.~2]{FiorenzaMonacoPanati16}). The quantity $\delta \in \Z_2$ defined by the above relation can be shown to be a topological invariant of the family of projectors $\set{P(\kk)}_{\kk \in \R^2}$, in the sense that its value does not depend on the choices performed in the construction sketched above, and that it remains unchanged under homotopic deformations which preserve the conditions stated in Assumption \ref{assum:proj}.

\subsubsection{The unitary $U(\kk)$}

We recall also the definition of $U$, as this will be convenient in what follows. Given the frame $\Psi(\kk)$, one computes the \emph{obstruction unitaries} $U\sub{obs}(\kk_*)$ by
\begin{equation} \label{eqn:Uobs}
\theta \Psi(\kk_*) \act \eps = \Psi(\kk_*) \act U\sub{obs}(\kk_*), \quad \text{for } \kk_* \in V, 
\end{equation}
where
\begin{equation} \label{eqn:TRIM}
V := \set{\left( 0, 0 \right), \, \left( 0, -\frac{1}{2} \right), \, \left( \frac{1}{2}, 0 \right), \, \left( \frac{1}{2}, - \frac{1}{2} \right)}.
\end{equation}
One can show that $U\sub{obs}(\kk_*)$ satisfies
\[ \eps \, U\sub{obs}(\kk_*) = U\sub{obs}(\kk_*)^t \, \eps \]
and that this implies the existence of a matrix $U(\kk_*)$ such that
\[ U\sub{obs}(\kk_*) = U(\kk_*) \, \eps^{-1} \, U(\kk_*)^t \, \eps \]
(see \cite[Lemma 1]{FiorenzaMonacoPanati16}; compare also \cite[Thm.~1]{Schulz-Baldes15} and \cite{Hua44}). To obtain $U(\kk_*)$ explicitly, one writes $U\sub{obs}(\kk_*) = \eu^{\iu M(\kk_*)}$, with $M(\kk_*) = M(\kk_*)^*$ a self-adjoint matrix whose spectrum is contained in $[0, 2 \pi)$, and sets
\begin{equation} \label{eqn:U(k*)}
U(\kk_*) := \eu^{\iu M(\kk_*)/2}.
\end{equation}

Once $U(\kk_*)$ is defined at the four points $\kk_* \in V$, one chooses an arbitrary interpolation $\widetilde{U}$ between these matrices defined on $S := \partial \Beff \cap \set{k_2 \le 0}$, \ie on the four edges joining these four points in $\partial \Beff$. This can always be done because the unitary group $U(m)$ is path-connected. Once this is done, one defines 
\[ \widetilde{\Phi}(\kk) := \Psi(\kk) \act \widetilde{U}(\kk), \quad \kk \in S, \]
and extends the definition of $\widetilde{\Phi}$ to the whole $\partial \Beff$ by setting
\begin{equation} \label{eqn:extend}
\Phi(\kk) := \begin{cases}
\widetilde{\Phi}(\kk) & \text{if } \kk \in S, \\
\widetilde{\Phi}(k_1,-1/2) & \text{if } \kk = (k_1, 1/2), \: k_1 \in [0, 1/2], \\
\theta \widetilde{\Phi}(k_*,-k_2) \act \eps & \text{if } \kk = (k_*, k_2), \: k_* \in \set{0, 1/2}, \: k_2 \in [0, 1/2].
\end{cases}
\end{equation}
This definition of $\Phi(\kk)$ guarantees that it is periodic and TRS, where it is defined. One can then extend the definition of the unitary $U$ from $S$ to the whole $\partial \Beff$, by comparing the frame $\Phi(\kk)$ and the frame $\Psi(\kk)$ as in \eqref{eqn:healer}. With this $U \colon \partial \Beff \to U(m)$, one can then compute the $\Z_2$ invariant $\delta$ as in \eqref{eqn:delta}.

\subsubsection{Relation between $U$ and $\alpha$}

We now relate the matching matrix $\alpha$ to the unitary $U$ that can be computed, via the procedure illustrated in Section \ref{sec:d=d}, from the frame $\Psi(\kk)$ constructed in \eqref{eqn:Psi} by means of parallel transport. Notice first that by \eqref{eqn:healer} and \eqref{eqn:extend} it follows that $U(1/2,k_2) = \widetilde{U}(1/2,k_2)$ if $k_2 \in [-1/2,0]$. We study now how $U(1/2,k_2)$ is defined for $k_2 \in [0,1/2]$. We have, in view of the TRS of $\Psi(\kk)$,
\begin{align*}
\Phi(1/2,k_2) & = \theta \widetilde{\Phi}(1/2,-k_2) \act \eps = \theta \left( \Psi(1/2,-k_2) \act \widetilde{U}(1/2,-k_2) \right) \act \eps \\
 & = \theta \Psi(1/2,-k_2) \act \left( \overline{\widetilde{U}(1/2,-k_2)} \, \eps \right) \\
 & = \Psi(-1/2,k_2) \act \left( \eps^{-1} \, \overline{\widetilde{U}(1/2,-k_2)} \, \eps \right) \\
 & = \Psi(1/2,k_2) \act \left( \alpha(k_2)^{-1} \, \eps^{-1} \, \overline{\widetilde{U}(1/2,-k_2)} \, \eps \right).
\end{align*}
By comparing the above relation with the defining identity
\[ \Phi(1/2,k_2) = \Psi(1/2,k_2) \act U(1/2,k_2) \]
and the already proved relation $U(1/2,-k_2) = \widetilde{U}(1/2,-k_2)$ for $k_2 \in [0,1/2]$, we conclude by the freeness of the $U(m)$-action on $m$-frames that
\[ U(1/2,k_2) = \alpha(k_2)^{-1} \, \eps^{-1} \, \overline{U(1/2,-k_2)} \, \eps, \] 
or equivalently
\begin{equation} \label{eqn:alpha_vs_U}
\alpha(k_2) = \eps^{-1} \, \overline{U(1/2,-k_2)} \, \eps \, U(1/2,k_2)^{-1}.
\end{equation}
Notice also that by setting
\begin{equation} \label{eqn:gamma_vs_U}
\gamma(k_2) := U(1/2,k_2)^{-1}
\end{equation}
then \eqref{eqn:alpha_vs_U} reduces exactly to \eqref{eqn:factor}. This gives an alternative proof of Lemma \ref{lemma:factor}.

\subsection{Equivalence between $\mathcal{I}$ and $\delta$} \label{sec:GPappendix}

Having reviewed the definition of both the Graf--Porta index $\mathcal{I}$ and the Fiorenza--Monaco--Panati invariant $\delta$, both associated to $2$-dimensional continuous periodic families of projectors enjoying a fermionic time-reversal symmetry, we are finally able to prove that two $\Z_2$ indices agree numerically. 

\begin{theorem} \label{thm:GP_vs_FMP}
Let $\set{P(\kk)}_{\kk \in \R^2}$ be a family of projectors satisfying Assumption \ref{assum:proj}. Let $\Psi$ be the Bloch frame for $\set{P(\kk)}_{\kk \in \R^2}$ defined in \eqref{eqn:Psi}. Define the corresponding matching matrix $\alpha$ as in \eqref{eqn:alpha} and the unitary matrix $U$ as in \eqref{eqn:healer}. Then
\[ \mathcal{I}(\alpha) \equiv \deg([\det U]) \bmod 2, \quad \text{or equivalently} \quad \mathcal{I}(\alpha) = \delta \in \Z_2. \]
\end{theorem}
\begin{proof}

Observe that, if the frame $\Psi(\kk)$ that is used to compute the family of matching matrices $\alpha$ is the one appearing in \eqref{eqn:Psi}, then $U(0,k_2) \equiv \Id$, as $\Psi(0,k_2) = \Xi(0,k_2)$ is already continuous and symmetric. Moreover, the interpolation of $\widetilde{U}(k_1,-1/2)$ between $U(0,-1/2)$ and $U(1/2,-1/2)$, the latter being computed in terms of the obstruction matrices $U\sub{obs}(\kk_*)$ via \eqref{eqn:U(k*)}, is not actually relevant for the evaluation of the degree of $\det U(\kk)$ along $\partial \Beff$, since periodicity in $k_2$ of the frames $\Phi(\kk)$ and $\Psi(\kk)$, which are mapped to each other through $U(\kk)$, dictate that also $U(k_1,1/2) = \widetilde{U}(k_1,-1/2)$ for $k_1 \in [0,1/2]$ (compare \eqref{eqn:extend}). This implies in particular that, using the integral formula in \eqref{eqn:degree} to evaluate $\deg([\det U])$, the contributions coming from integrating along the two short edges $\set{k_1 \in [0,1/2], \, k_2 = \pm 1/2}$ of $\partial \Beff$ cancel each other out. We deduce that, with the choice of $\Psi$ that we did above,
\begin{equation} \label{eqn:degUk2}
\deg([\det U]) = \frac{1}{2 \pi \iu} \int_{-1/2}^{1/2} (\det U(1/2,k_2))^{-1} \, \partial_{k_2} \det U(1/2,k_2) \, \di k_2.
\end{equation}

We can conclude then, in view of the formula \eqref{eqn:IntegralRueda_c} for the Graf--Porta index and owing to \eqref{eqn:gamma_vs_U}, that
\[ \mathcal{I}(\alpha) \equiv - \deg([\det U]) \equiv \delta \mod 2 \]
as claimed.
\end{proof}

\subsection{Geometric formula for the $\Z_2$ invariant} \label{sec:GeoFormula}

In view of the equality between the Fiorenza--Monaco--Panati and the Graf--Porta $\Z_2$ invariants (Theorem \ref{thm:GP_vs_FMP}), we can express the $\Z_2$ invariant $\delta$ defined in \eqref{eqn:delta} in terms of the family of matching matrices $\alpha$ appearing in \eqref{eqn:alpha}. The latter is in turn related to the parallel transport operator $T_{k_2}(1,0)$ via \eqref{eqn:alpha_vs_Tk2}: we want to exploit this connection to express $\delta$ in more geometric terms.

In order to do so, we first need to express \eqref{eqn:alpha_vs_Tk2} in a somewhat more intrinsic form. To this end, we begin by noticing that the parallel transport unitary $\mathcal{T}(k_2) := T_{k_2}(1,0) \in \U(\Hi)$ maps the range of the projector $P(0,k_2)$ into that of $P(1,k_2) = P(0,k_2)$, and thus commutes with $P(0,k_2)$. According to \eqref{eqn:alpha_vs_Tk2}, $\alpha(k_2)$ is nothing but the matrix which represents the restriction of this unitary map to $P(0,k_2)$, in the basis given by $\Xi(0,k_2)$. 

If we denote by $A(k_2) = \det \alpha(k_2)$, it follows from Lemma \ref{lemma:DegUnitary} that
\[ A(k_2)^{-1} \, A'(k_2) = \tr \left( \alpha(k_2)^* \, \partial_{k_2} \alpha(k_2) \right). \]
We want to write the right-hand side of the above equality directly in terms of the parallel transport operator $\mathcal{T}(k_2)$. In order to do so, we exploit the following Lemma.

\begin{lemma}
Let $\set{\alpha(k_2)}_{k_2 \in \R}$ be as in \eqref{eqn:alpha_vs_Tk2}. Then
\begin{equation} \label{eqn:A->S}
\tr \left( \alpha(k_2)^* \, \partial_{k_2} \alpha(k_2) \right) = \Tr \left( P(0,k_2) \, \mathcal{T}(k_2)^* \, \partial_{k_2} \mathcal{T}(k_2) \, P(0,k_2) \right),
\end{equation}
where $\tr(\cdot)$ denotes the trace in $\C^m$ and $\Tr(\cdot)$ denotes the trace in the Hilbert space $\Hi$.
\end{lemma}
\begin{proof}
The defining relation \eqref{eqn:alpha_vs_Tk2} can be recast in the form
\[ \alpha(k_2)_{ab} = \scal{\Xi_a(0,k_2)}{T_{k_2}(1,0) \, \Xi_b(0,k_2)}, \quad a,b \in \set{1, \ldots, m}. \]
The above implies in particular that
\begin{align*}
\alpha(k_2)^*_{ab} & = \scal{\mathcal{T}(k_2) \, \Xi_a(0,k_2)}{\Xi_b(0,k_2)}, \\
\partial_{k_2} \alpha(k_2)_{ba} & = \scal{\partial_{k_2} \Xi_b(0,k_2)}{\mathcal{T}(k_2) \, \Xi_a(0,k_2)} + \scal{\Xi_b(0,k_2)}{\mathcal{T}(k_2) \, (\partial_{k_2} \Xi_a(0,k_2))} \\
& \quad + \scal{\Xi_b(0,k_2)}{(\partial_{k_2} \mathcal{T}(k_2)) \, \Xi_a(0,k_2)}.
\end{align*}
Using the fact that 
\[ P(0,k_2) = \sum_{a=1}^{m} \ket{\Xi_a(0,k_2)} \bra{\Xi_a(0,k_2)} = \sum_{a=1}^{m} \ket{\mathcal{T}(k_2) \, \Xi_a(0,k_2)} \bra{\mathcal{T}(k_2) \, \Xi_a(0,k_2)} \]
(as $\mathcal{T}(k_2)$ is unitary on $\Ran P(0,k_2)$), we put the two equalities above together and obtain
\begin{align*}
\tr \left( \alpha(k_2)^* \, \partial_{k_2} \alpha(k_2) \right) & = \sum_{a,b=1}^{m} \alpha(k_2)^*_{ab} \, \partial_{k_2} \alpha(k_2)_{ba} \\
& = \sum_{a,b=1}^{m} \big\{ \scal{\partial_{k_2} \Xi_b(0,k_2)}{\mathcal{T}(k_2) \, \Xi_a(0,k_2)} \scal{\mathcal{T}(k_2) \, \Xi_a(0,k_2)}{\Xi_b(0,k_2)}  \\
& \quad + \scal{\mathcal{T}(k_2) \, \Xi_a(0,k_2)}{\Xi_b(0,k_2)} \scal{\Xi_b(0,k_2)}{\mathcal{T}(k_2) \, (\partial_{k_2} \Xi_a(0,k_2))} \\
& \quad + \scal{\mathcal{T}(k_2) \, \Xi_a(0,k_2)}{\Xi_b(0,k_2)} \scal{\Xi_b(0,k_2)}{(\partial_{k_2} \mathcal{T}(k_2)) \, \Xi_a(0,k_2)} \big\} \\
& = \sum_{a=1}^{m} \big\{ \scal{\partial_{k_2} \Xi_a(0,k_2)}{\Xi_a(0,k_2)} + \scal{\Xi_a(0,k_2)}{\partial_{k_2} \Xi_a(0,k_2)} \\
& \quad + \scal{\mathcal{T}(k_2) \, \Xi_a(0,k_2)}{(\partial_{k_2} \mathcal{T}(k_2)) \, \Xi_a(0,k_2)} \big\} \\
& = \sum_{a=1}^{m} \partial_{k_2} \left( \scal{\Xi_a(0,k_2)}{\Xi_a(0,k_2)} \right) \\
& \quad + \sum_{a=1}^{m} \scal{\Xi_a(0,k_2)}{(\mathcal{T}(k_2)^* \, \partial_{k_2} \mathcal{T}(k_2)) \, \Xi_a(0,k_2)} \\
& = \Tr \left( P(0,k_2) \, \mathcal{T}(k_2)^* \, \partial_{k_2} \mathcal{T}(k_2) \, P(0,k_2) \right)
\end{align*}
where in the last equality we employed the fact that $\set{\Xi_a(0,k_2)}_{1 \le a \le m}$ gives an orthonormal basis in $\Ran P(0,k_2)$. This concludes the proof.
\end{proof}

To compute the right-hand side of \eqref{eqn:A->S}, we appeal to the following result.

\begin{lemma} \label{lemma:variation_of_constants}
Let $T_{k_2}(k_1,0)$ be the parallel transport unitary defined in Section \ref{sec:parallel}, and $G(k_1,k_2) := [\partial_{k_1} P(k_1,k_2), P(k_1,k_2)]$. Then
\begin{equation} \label{eqn:VoC}
\partial_{k_2} T_{k_2}(k_1,0) = T_{k_2}(k_1,0) \, \int_{0}^{k_1} T_{k_2}(s,0)^* \, \partial_{k_2} G(s, k_2) \, T_{k_2}(s,0) \, \di s.
\end{equation}
\end{lemma}
\begin{proof}
Recall that the parallel transport $T_{k_2}(k_1,0)$ satisfies the Cauchy problem \eqref{eqn:parallel_def}, which in this case reads
\[ \partial_{k_1} T_{k_2}(k_1,0) = G(k_1, k_2) \, T_{k_2}(k_1,0), \quad T_{k_2}(0,0) = \Id. \]
Call $R(k_1,k_2)$ the expression on the right-hand side of \eqref{eqn:VoC}; clearly $R(0,k_2) = 0$. Taking a derivative with respect to $k_1$, one then obtains
\begin{align*} 
\partial_{k_1} R(k_1, k_2) & = \partial_{k_1} T_{k_2}(k_1,0) \, \left( \int_{0}^{k_1} T_{k_2}(s,0)^* \, \partial_{k_2} G(s, k_2) \, T_{k_2}(s,0) \, \di s \right) + \\
& \quad + T_{k_2}(k_1,0) \, \left( T_{k_2}(k_1,0)^* \, \partial_{k_2} G(k_1, k_2) \, T_{k_2}(k_1,0) \right) \\
& = G(k_1, k_2) \, T_{k_2}(k_1,0) \, \left( \int_{0}^{k_1} T_{k_2}(s,0)^* \, \partial_{k_2} G(s, k_2) \, T_{k_2}(s,0) \, \di s \right) + \\
& \quad + \partial_{k_2} G(k_1, k_2) \, T_{k_2}(k_1,0) \\
& = G(k_1, k_2) \, R(k_1, k_2) + \partial_{k_2} G(k_1, k_2) \, T_{k_2}(k_1,0).
\end{align*}

On the other hand, the operator $\partial_{k_2} T_{k_2}(k_1,0)$ also vanishes when $k_1 = 0$ and satisfies the differential equation
\[ \partial_{k_1} \partial_{k_2} T_{k_2}(k_1,0) = G(k_1, k_2) \partial_{k_2} T_{k_2}(k_1,0) + \partial_{k_2} G(k_1,k_2) \, T_{k_2}(k_1,0) . \]
It follows that both sides of \eqref{eqn:VoC} solve the same Cauchy problem, and hence they must coincide.
\end{proof}

Evaluating \eqref{eqn:VoC} at $k_1 = 1$, we obtain
\begin{multline} \label{eqn:T*dT}
P(0,k_2) \, \mathcal{T}(k_2)^* \, \partial_{k_2} \mathcal{T}(k_2) \, P(0,k_2) \\ = \int_{0}^{1} P(0,k_2) \, T_{k_2}(k_1,0)^* \, \partial_{k_2} G(k_1, k_2) \, T_{k_2}(k_1,0) \, P(0,k_2) \, \di k_1.
\end{multline}
We observe that the operator $T_{k_2}(k_1,0)$ maps $\Ran P(0,k_2)$ to $\Ran P(k_1,k_2)$, and hence
\begin{gather*}
T_{k_2}(k_1,0) \, P(0,k_2) = P(k_1,k_2) \, T_{k_2}(k_1,0) \, P(0,k_2), \\
P(0,k_2) \, T_{k_2}(k_1,0)^* = P(0,k_2) \, T_{k_2}(k_1,0)^* \, P(k_1,k_2).
\end{gather*}
Thus we can freely insert the projector $P(k_1,k_2)$ on the right-hand side of \eqref{eqn:T*dT} before and after $\partial_{k_2} G(k_1, k_2)$. Using now the definition of $G(k_1, k_2) = [\partial_{k_1} P(k_1,k_2), P(k_1,k_2)]$, we obtain
\[ P(\kk) \, \partial_{k_2} G(\kk) \, P(\kk) = P(\kk) \left[ \partial^2_{k_1 \, k_2} P(\kk), P(\kk) \right] P(\kk) + P(\kk) \, [ \partial_{k_1} P(\kk), \partial_{k_2} P(\kk) ] \, P(\kk) \]
with $\kk = (k_1, k_2)$. As $P(\kk)^2 = P(\kk)$, one immediately sees that the first summand on the right-hand side of the above equality actually vanishes. We conclude that
\begin{multline*}
P(0,k_2) \, \mathcal{T}(k_2)^* \, \partial_{k_2} \mathcal{T}(k_2) \, P(0,k_2) \\ = \int_{0}^{1} P(0,k_2) \, T_{k_2}(k_1,0)^* \, P(k_1,k_2) \, [ \partial_{k_1} P(k_1,k_2), \partial_{k_2} P(k_1,k_2) ] \, P(k_1,k_2) \, T_{k_2}(k_1,0) \, P(0,k_2) \, \di k_1.
\end{multline*}

We now take the trace of both sides of the above equality. In view of the fact that $P(k_1,k_2) \, T_{k_2}(k_1,0) \, P(0,k_2)$ is a unitary transformation of $\Ran P(0, k_2)$ into $\Ran P(k_1, k_2)$, the expression simplifies to
\begin{equation} \label{eqn:BerryCurv}
\Tr \left( P(0,k_2) \, \mathcal{T}(k_2)^* \, \partial_{k_2} \mathcal{T}(k_2) \, P(0,k_2) \right) = \int_{0}^{1} \Tr \left( P(\kk) \, [ \partial_{k_1} P(\kk), \partial_{k_2} P(\kk) ] \right) \, \di k_1.
\end{equation}

We are now ready to prove the following

\begin{theorem} \label{thm:deltaGeom}
Let $\set{P(\kk)}_{\kk \in \R^2}$ be a family of projectors satisfying Assumption \ref{assum:proj}. The associated $\Z_2$ invariant $\delta = \delta(P)$ is given by
\begin{equation} \label{eqn:deltaGeom}
\delta = \frac{1}{2 \pi } \int_{0}^{1/2} \di k_2 \int_{0}^{1} \di k_1 \, \mathcal{F}(k_1, k_2) - \frac{1}{2 \pi}\left( \oint_{\Gamma_{1/2}} \A - \oint_{\Gamma_{0}} \A\right) \mod 2
\end{equation}
where:
\begin{itemize}
\item $\mathcal{F}$ is the \emph{Berry curvature}
\begin{equation} \label{eqn:BerryCurvature}
\mathcal{F}(\kk) = - \iu \, \Tr \left( P(\kk) \, [ \partial_{k_1} P(\kk), \partial_{k_2} P(\kk) ] \right);
\end{equation}
\item $\Gamma_{0}$ (respectively $\Gamma_{1/2}$) is the positively-oriented loop in the \emph{Brillouin torus} $\T^2 := \R^2/\Z^2$ given by $\set{k_2=0}$ (respectively $\set{k_2=1/2}$);
\item $\A$ is the trace of the Berry connection \eqref{eqn:ABerry}, computed with respect to a continuous and symmetric Bloch frame on $\Gamma_0$ and $\Gamma_{1/2}$. 
\end{itemize}
\end{theorem}
\begin{proof}
In view of Theorem \ref{thm:GP_vs_FMP}, we have that $\delta = \mathcal{I}(\alpha) \in \Z_2$, where $\alpha$ is as in \eqref{eqn:alpha_vs_Tk2}. We use the expression \eqref{eqn:IntegralRueda_b} for $\mathcal{I}(\alpha)$: the first summand there equals the first summand of \eqref{eqn:deltaGeom} in view of \eqref{eqn:A->S} and \eqref{eqn:BerryCurv}. 

It remains to show that also the second term on the right-hand side of \eqref{eqn:IntegralRueda_b} agrees with the second term on the right-hand side of \eqref{eqn:deltaGeom}, \ie that
\begin{equation} \label{eqn:gammaBerry} 
2 \left(\frac{1}{2 \pi \iu} \, \log \frac{\det \gamma(0)}{\det \gamma(1/2)}\right) \equiv - \frac{1}{2 \pi}\left( \oint_{\Gamma_{1/2}} \A - \oint_{\Gamma_{0}} \A\right) \bmod 2
\end{equation}
with $\gamma(0)$ and $\gamma(1/2)$ are as in \eqref{eqn:gammas}. Arguing as in Section \ref{sec:matching1d}, since $\theta \, \mathcal{T}(0) \, \theta^{-1} = \mathcal{T}(0)$ we can write $\mathcal{T}(0) = \eu^{\iu \, M}$ with $M = M^* \in \BH$ commuting with $P(0,0)$. Correspondingly, $\alpha(0) = \eu^{\iu \, h}$, with $h = h^*$ the self-adjoint matrix associated to the self-adjoint operator $P(0,0) \, M \, P(0,0)$ on $\Ran P(0,0)$. The analogous statement can be formulated for $\mathcal{T}(1/2)$, owing to \ref{item:T-periodic} and \ref{item:T-TRS}. It follows then that \eqref{eqn:gammaBerry1d} holds for $\gamma(0)$ and $\gamma(1/2)$, on the specified loops $\Gamma_0$ and $\Gamma_{1/2}$, respectively: this implies \eqref{eqn:gammaBerry}.
\end{proof}

The formula \eqref{eqn:deltaGeom} expresses the $\Z_2$ invariant $\delta$ purely in terms of geometric data, namely the Berry curvature and the trace of the Berry connection $1$-form, integrated along the loops $\set{k_2 = 0}$ and $\set{k_2 = 1/2}$ in the Brillouin torus $\T^2$. As was already mentioned in the Introduction, this formula coincides with the one appearing in \cite[Eqn.~(A8)]{FuKane06} for the Fu--Kane $\Z_2$ index associated to TRS topological insulators and quantum spin Hall systems. Thus, Theorem~\ref{thm:deltaGeom} gives a proof of the first equality in Theorem~\ref{thm:MainResults_Z2}.

\subsection{Evaluation of the $\Z_2$ invariant on TRIMs} \label{sec:TRIM}

As was early realized \cite{FuKane06}, one of the main peculiarities of the $\Z_2$ invariant of $2$-dimensional TRS topological insulators is that it can be evaluated by considering appropriate quantities which are defined only at the four \emph{time-reversal invariant momenta} (TRIMs) in a unit cell, which are fixed by the combined action of the integer shifts and of the involution $\kk \to -\kk$. These TRIMS are the points in the set $V$ appearing in \eqref{eqn:TRIM}. Geometrically, this feature was recently understood with the use of an appropriate equivariant cohomology \cite{DeNittisGomi15}.

We illustrate this characteristic of the $\Z_2$ invariant by recalling one of the possible definitions of the Fu--Kane index \cite[Eqn.~(3.26)]{FuKane06}. Given a continuous and periodic Bloch frame $\Phi$, define the \emph{sewing matrix}
\[ w(\kk)_{ab} := \scal{\Phi_a(-\kk)}{\theta \, \Phi_b(\kk)}, \quad a, b \in \set{1, \ldots, m}. \]
The matrix $w(\kk)$ is unitary, and gives a measure of the failure of the frame $\Phi$ to be TRS. Notice indeed that if $\Phi$ is a symmetric frame, then $w(\kk) \equiv \eps^{-1}$ for all $\kk \in \R^2$.

One easily realizes that $w(\kk + \mathbf{n}) = w(\kk)$ for $\kk \in \R^2$ and $\mathbf{n} \in \Z^2$, and that moreover $w(-\kk) = - w(\kk)^t$ for $\kk \in \R^2$. In particular, at the four TRIMs in \eqref{eqn:TRIM} we have that $w(\kk_*) = - w(\kk_*)^t$, $\kk_* \in V$, namely $w(\kk_*)$ is skew-symmetric. As a consequence, it has a well defined Pfaffian, squaring to its determinant. Upon an adequate choice of the branch of the square root, the Fu--Kane $\Z_2$ index can be defined as%
\footnote{Notice that we switch from the ``additive'' definition of $\Z_2 = \set{0,1}$ to the ``multiplicative'' one $\Z_2 = \set{-1, +1}$.}%
\[ (-1)^\Delta := \prod_{\kk_* \in V} \frac{\sqrt{\det w(\kk_*)}}{\Pf w(\kk_*)} \quad \in \Z_2. \]

A similar formulation holds also for the Fiorenza--Monaco--Panati invariant $\delta$: indeed
\[ (-1)^\delta = \prod_{\kk_* \in V} \frac{\sqrt{\det U\sub{obs}(\kk_*)}}{\det U(\kk_*)} \]
(compare \cite[Sec.~5]{FiorenzaMonacoPanati16}), where the obstruction matrices $U\sub{obs}(\kk_*)$ and the matrices $U(\kk_*)$ are defined in \eqref{eqn:Uobs} and \eqref{eqn:U(k*)}, respectively. This expression can be used to show that $\delta = \Delta \in \Z_2$ \cite[Thm.~5]{FiorenzaMonacoPanati16}.

We show now that, in the same spirit, also the Graf--Porta index $\mathcal{I}(\alpha) \in \Z_2$ of a family of matrices $\alpha$ as in Assumption \ref{assum:alpha} can be expressed in terms of the endpoint matrices $\alpha(0)$ and $\alpha(1/2)$. In order to do so, it will be useful to appeal to the following fact.

\begin{lemma} \label{lemma:detalpha=1}
Let $\set{\alpha(k_2)}_{k_2 \in \R}$ be as in Assumption \ref{assum:alpha}. Then there exists a continuous, $\Z$-periodic and even function $\varphi \colon \R \to \R$ such that
\[ \det \alpha(k_2) = \eu^{\iu \, \varphi(k_2)}, \quad k_2 \in \R. \]
\end{lemma}
\begin{proof}
It was already observed in the discussion at the beginning of Section \ref{sec:GrafPorta} that the TRS condition \ref{item:alpha_TRS} implies that  $k_2 \mapsto \det \alpha(k_2)$ is an even function. As such, it has a vanishing winding number, which implies the existence of a continuous, periodic and even argument $\varphi$ as required in the statement (see also \cite[Lemma~2.13]{CorneanHerbstNenciu15} for an explicit construction of $\varphi$).
\end{proof}

\begin{proposition}
Let $\set{\alpha(k_2)}_{k_2 \in \R}$ be as in Assumption \ref{assum:alpha}. Then its Graf--Porta index $\mathcal{I}(\alpha)$ can be computed by
\begin{equation} \label{eqn:-1I}
(-1)^{\mathcal{I}(\alpha)} = \frac{\sqrt{\det \alpha(0)}}{\det \gamma(0)} \, \frac{\sqrt{\det \alpha(1/2)}}{\det \gamma(1/2)},
\end{equation}
where $\gamma(0), \gamma(1/2) \in U(m)$ are as in \eqref{eqn:gammas}.
\end{proposition}
\begin{proof}
The conclusion follows from the use of formula \eqref{eqn:IntegralRueda_b} to compute $\mathcal{I}(\alpha)$. Indeed, denoting again $A(k_2) := \det \alpha(k_2)$, and letting $\varphi$ be as in Lemma \ref{lemma:detalpha=1}, we have
\[ A'(k_2) = \iu \, \varphi'(k_2) \, A(k_2) \quad \Longrightarrow \quad \frac{1}{2 \pi \iu} \int_{0}^{1/2} A(k_2)^{-1} \, A'(k_2) \, \di k_2 = \frac{\varphi(1/2) - \varphi(0)}{2 \pi}. \]
One then immediately deduces that
\[ (-1)^{\int_{0}^{1/2} A(k_2)^{-1} \, A'(k_2) \, \di k_2} = (\eu^{\iu \pi})^{\{\varphi(1/2) - \varphi(0)\}/2\pi} = \eu^{\iu \, \varphi(1/2)/2} \, \eu^{\iu \, \varphi(0)/2} = \frac{\sqrt{A(1/2)}}{\sqrt{A(0)}}. \]

On the other hand, arguing similarly using the principal branch of the complex logarithm, we have that
\[ (-1)^{2 \, \log \{\det \gamma(1/2) / \det \gamma(0)\}/2\pi \iu} = \frac{\det \gamma(1/2)}{\det \gamma(0)}. \]
Combining the above two equalities we obtain exactly \eqref{eqn:-1I}, in view of \eqref{eqn:IntegralRueda_b}.
\end{proof}

The above expression for the Graf--Porta $\Z_2$ index can be now compared to another formulation for the invariant of TRS topological insulators, proposed by Prodan in \cite{Prodan11_PRB}. We review Prodan's definition for the reader's convenience, adapting it to our notation. Using properties \ref{item:T-periodic}, \ref{item:T-TRS} and \ref{item:T-group} for the parallel transport $T_{k_2}(1,0)$, we notice that at $k_2 = 0$
\[ T_0(1,0) = T_0(1,1/2) \, T_0(1/2, 0) = T_0(0,-1/2) \, \theta \, T_0(0,-1/2)^{-1} \, \theta^{-1}, \]
and similarly at $k_2=1/2$
\[ T_{1/2}(1,0) = T_{1/2}(0,-1/2) \, \theta \, T_{1/2}(0,-1/2)^{-1} \, \theta^{-1}. \]
If $k_* \in \set{0,1/2}$, the above two equalities imply that
\begin{multline} \label{eqn:mess}
P(0,k_*) \, T_{k_*}(1,0) \, P(0,k_*) \\
= P(0,k_*) \, T_{k_*}(0,-1/2) \, P(-1/2,k_*) \, \theta \, P(-1/2,k_*) \,  T_{k_*}(0,-1/2)^{-1} \, P(0,k_*) \, \theta^{-1} \, P(0,k_*)
\end{multline}
in view of the intertwining property \ref{item:T-intertwine}. Arguing as in Section \ref{sec:d=0}, we can choose bases $\Xi(0,k_*)$ in $\Ran P(0,k_*)$ and $\Xi(-1/2,k_*)$ in $\Ran P(-1/2,k_*)$ which satisfy
\[ \theta \, \Xi(0,k_*) \act \eps = \Xi(0,k_*), \quad \theta \, \Xi(-1/2,k_*) \act \eps = \Xi(-1/2,k_*). \]
Thus, in these bases both the operators
\begin{gather*}
P(0,k_*) \, \theta \, P(0,k_*) \colon \Ran P(0,k_*) \to \Ran P(0,k_*) \quad \text{and} \\
P(-1/2,k_*) \, \theta \, P(-1/2,k_*) \colon \Ran P(-1/2,k_*) \to \Ran P(-1/2,k_*) 
\end{gather*}
act as $\eps C$, where $C$ is the complex conjugation operator. Notice also that, in view of \eqref{eqn:alpha_vs_Tk2}, the restriction of the operator on the left-hand side of \eqref{eqn:mess} is represented by the matrix $\alpha(k_*)$ in the basis $\Xi(0,k_*)$. The equality in \eqref{eqn:mess} implies then
\begin{equation} \label{eqn:gammahat}
\alpha(k_*) = \widehat{\gamma}(k_*) \, \eps \, \widehat{\gamma}(k_*)^t \, \eps^{-1}, \quad k_* \in \set{0,1/2},
\end{equation}
where $\widehat{\gamma}(k_*)$ is the matrix representing to the operator 
\[ P(0,k_*) \, T_{k_*}(0,-1/2) \, P(-1/2,k_*) \colon \Ran P(-1/2,k_*) \to \Ran P(0,k_*) \]
upon choosing the basis $\Xi(-1/2,k_*)$ in the domain and the basis $\Xi(0,k_*)$ in the image. The Prodan $\Z_2$ invariant $\xi \in \Z_2$ is then defined (see \cite[Eqn.~(33)]{Prodan11_PRB}) by the relation
\begin{equation} \label{eqn:Prodan}
(-1)^\xi := \frac{\sqrt{\det \alpha(0)}}{\det \widehat{\gamma}(0)} \, \frac{\sqrt{\det \alpha(1/2)}}{\det \widehat{\gamma}(1/2)}.
\end{equation}

It is now immediate to see that the following statement holds.
\begin{theorem} \label{thm:GP=Prodan}
Let $\set{P(\kk)}_{\kk \in \R^2}$ be a family of projectors as in Assumption \ref{assum:proj}. Let $\alpha$ be as in \eqref{eqn:alpha_vs_Tk2}. Then
\[ \mathcal{I}(\alpha) = \xi \in \Z_2, \]
where $\mathcal{I}(\alpha)$ is defined in \eqref{eqn:rueda} and $\xi$ is defined in \eqref{eqn:Prodan}.
\end{theorem}
\begin{proof}
In view of \eqref{eqn:-1I}, it suffices to notice that if $\widehat{\gamma}(k_*)$, $k_* \in \set{0,1/2}$, is as in \eqref{eqn:gammahat}, then 
\[ \gamma(k_*) := \eps^{-1} \, \widehat{\gamma}(k_*) \, \eps, \quad k_* \in \set{0,1/2}, \]
satisfies \eqref{eqn:gammas}.
\end{proof}


\appendix

\goodbreak


\section{Lifting extra degeneracies of the matching matrix} \label{sec:ExtraDegen}

In this Appendix, we prove a few auxiliary results, which were used in Section \ref{sec:Removal}. The proofs adapt the argument in \cite[Lemma~2.18]{CorneanHerbstNenciu15}.

\begin{lemma}[Local splitting lemma] \label{lemma:splitting}
Let $R>0$ and $k_0 \in \R$. Denote by $B_R(k_0)$ the interval $[k_0 - R, k_0 + R]$. Let $\set{\alpha(k) = \eu^{\iu h(k)}}_{k \in B_R(k_0)}$ be a family of $m \times m$ unitary matrices such that $k \mapsto h(k) = h(k)^*$ is smooth on $B_R(k_0)$.
\begin{enumerate}[label=$(\roman*)$,ref=$(\roman*)$]
 \item \label{item:split1} Under the above assumptions, it is possible to construct a sequence $\set{\alpha_j(k)=\eu^{\iu h_j(k)}}_{k \in B_R(k_0)}$ of families of unitary matrices, with $k \mapsto h_j(k)$ \emph{smooth} and 
 \[ \lim_{j \to \infty} \sup_{k_2 \in B_R(k_0)} \norm{\alpha(k_2) - \alpha_j(k_2)} = 0, \]
 such that the spectrum of $\alpha_j(k_0)$ is non-degenerate.
 \item \label{item:split2} Let $k_0 = 0$, and assume that $h(k)$ is \emph{time-reversal symmetric}, \ie
 \begin{equation} \label{hTRS}
 \eps \, h(k) = h(-k)^t \, \eps, \quad \text{and in particular} \quad \eps \, \alpha(k) = \alpha(-k)^t \, \eps.
 \end{equation}
 Then it is possible to construct a sequence $\set{\alpha_j(k)=\eu^{\iu h_j(k)}}_{k \in B_R(k_0)}$ of families of unitary matrices, with $k \mapsto h_j(k)$ \emph{smooth} and 
 \[ \lim_{j \to \infty} \sup_{k_2 \in B_R(0)} \norm{\alpha(k_2) - \alpha_j(k_2)} = 0, \]
 which are also time-reversal symmetric and such that all eigenvalues of $\alpha_j(0)$ are \emph{doubly} degenerate.
 \item \label{item:split3} Assume that the spectrum of $\alpha(k)$ is non-degenerate for all $k \ne k_0$. Then it is possible to construct a sequence $\set{\alpha_j(k)=\eu^{\iu h_j(k)}}_{k \in B_R(k_0)}$ of families of unitary matrices, with $k \mapsto h_j(k)$ \emph{continuous} and
 \[ \lim_{j \to \infty} \sup_{k_2 \in B_R(k_0)} \norm{\alpha(k_2) - \alpha_j(k_2)} = 0, \]
 such that the spectrum of $\alpha_j(k)$ is non-degenerate for all $k \in B_R(k_0)$.
\end{enumerate}
\end{lemma}

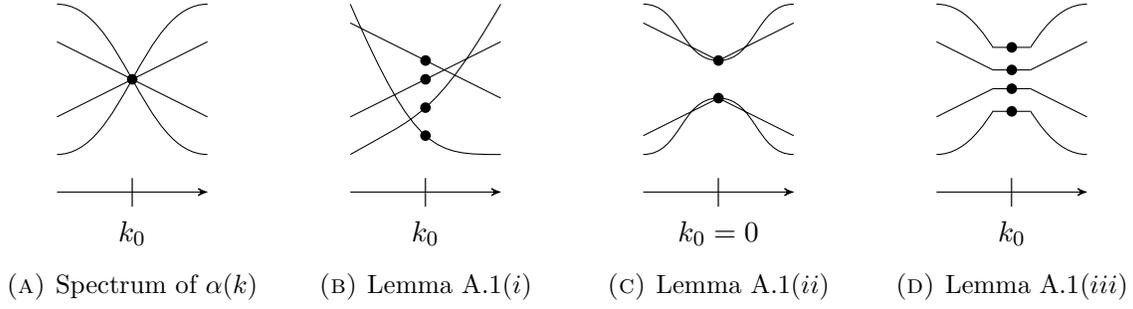
\begin{figure}[ht]
\begin{subfigure}[b]{.25\textwidth}
\centering
\begin{tikzpicture}
\draw (-1,-1) .. controls (-.125,-1) and (.125,1) .. (1,1)
      (-1,1) .. controls (-.125,1) and (.125,-1) .. (1,-1)
      (-1,.5) -- (1,-.5)
      (-1,-.5) -- (1,.5);
\fill (0,0) circle (2pt);
\draw [-stealth'] (-1,-1.5) -- (1,-1.5);
\draw (0,-1.75) node [anchor=north] {$k_0$};
\draw (0,-1.5) node {$|$};
\end{tikzpicture}
\caption{Spectrum of $\alpha(k)$}
\end{subfigure}%
\begin{subfigure}[b]{.25\textwidth}
\centering
\begin{tikzpicture}
\path [name path = line] (0,-1) -- (0,1);
\path [name path = lambda1] (-1,-1) .. controls (-.125,-.5) and (.125,-.5) .. (1,1);
\path [name path = lambda2] (-1,1) .. controls (-.125,-1) and (.125,-1) .. (1,-1);
\path [name path = lambda3] (-1,.75) -- (1,-.25);
\path [name path = lambda4] (-1,-.5) -- (1,.5);
\draw (-1,-1) .. controls (-.125,-.5) and (.125,-.5) .. (1,1)
      (-1,1) .. controls (-.125,-1) and (.125,-1) .. (1,-1)
      (-1,.75) -- (1,-.25)
      (-1,-.5) -- (1,.5);
\draw [-stealth'] (-1,-1.5) -- (1,-1.5);
\draw (0,-1.75) node [anchor=north] {$k_0$};
\draw (0,-1.5) node {$|$};
\fill [name intersections = {of = line and lambda1, by = c1}] (c1) circle (2pt);
\fill [name intersections = {of = line and lambda2, by = c2}] (c2) circle (2pt);
\fill [name intersections = {of = line and lambda3, by = c3}] (c3) circle (2pt);
\fill [name intersections = {of = line and lambda4, by = c4}] (c4) circle (2pt);
\end{tikzpicture}
\caption{Lemma \ref{lemma:splitting}\ref{item:split1}}
\end{subfigure}%
\begin{subfigure}[b]{.25\textwidth}
\centering
\begin{tikzpicture}
\draw (-1,1) .. controls (-.5,1) and (-.5,.25) .. (0,.25)
      (0,.25) .. controls (.5,.25) and (.5,1) .. (1,1)
      (-1,.75) -- (0,.25) -- (1,.75)
      (-1,-1) .. controls (-.5,-1) and (-.5,-.25) .. (0,-.25)
      (0,-.25) .. controls (.5,-.25) and (.5,-1) .. (1,-1)
      (-1,-.75) -- (0,-.25) -- (1,-.75);
\fill (0,.25) circle (2pt)
      (0,-.25) circle (2pt);
\draw [-stealth'] (-1,-1.5) -- (1,-1.5);
\draw (0,-1.75) node [anchor=north] {$k_0=0$};
\draw (0,-1.5) node {$|$};
\end{tikzpicture}
\caption{Lemma \ref{lemma:splitting}\ref{item:split2}}
\end{subfigure}%
\begin{subfigure}[b]{0.25\textwidth}
\centering
\begin{tikzpicture}
\path [name path = lambda1] (-1,-1) .. controls (-.125,-1) and (.125,1) .. (1,1);
\path [name path = lambda4] (-1,1) .. controls (-.125,1) and (.125,-1) .. (1,-1);
\path [name path = lambda3] (-1,.5) -- (1,-.5);
\path [name path = lambda2] (-1,-.5) -- (1,.5);
\path [name path = line1] (-.25,-1) -- (-.25,1);
\path [name path = line2] (.25,-1) -- (.25,1);
\draw (-1,-1) .. controls (-.125,-1) and (.125,1) .. (1,1)
      (-1,1) .. controls (-.125,1) and (.125,-1) .. (1,-1)
      (-1,.5) -- (1,-.5)
      (-1,-.5) -- (1,.5);
\fill [white] (-.25,-1) rectangle (.25,1);
\draw [name intersections = {of = lambda4 and line1, by = P41}, name intersections = {of = lambda1 and line2, by = P12}] (P41) -- (P12);
\draw [name intersections = {of = lambda3 and line1, by = P31}, name intersections = {of = lambda2 and line2, by = P22}] (P31) -- (P22);
\draw [name intersections = {of = lambda2 and line1, by = P21}, name intersections = {of = lambda3 and line2, by = P32}] (P21) -- (P32);
\draw [name intersections = {of = lambda1 and line1, by = P11}, name intersections = {of = lambda4 and line2, by = P42}] (P11) -- (P42);
\fill (P11)+(0.25,0) circle (2pt)
      (P21)+(0.25,0) circle (2pt)
      (P31)+(0.25,0) circle (2pt)
      (P41)+(0.25,0) circle (2pt);
\draw [-stealth'] (-1,-1.5) -- (1,-1.5);
\draw (0,-1.75) node [anchor=north] {$k_0$};
\draw (0,-1.5) node {$|$};
\end{tikzpicture}
\caption{Lemma \ref{lemma:splitting}\ref{item:split3}}
\end{subfigure}%
\caption{Local splitting of eigenvalues.}
\label{fig:split}
\end{figure}

\begin{proof}
By a shift $k \to k-k_0$, we can always assume $k_0 = 0$. We prove the first two statements together, assuming first time-reversal symmetry and highlighting which modifications are necessary in absence of it. The third statement requires a different proof. Figure~\ref{fig:split} illustrates the typical shape of the spectrum for the approximants we construct.

\smallskip

\noindent \textsl{Step 1: \ref{item:split1} and \ref{item:split2}.} Let us use the generic notation $\Pi$ for any of the spectral projections onto one of the eigenvalues of $\alpha(0)$; let $p$ denote the dimension of the range of $\Pi$ (\ie the degeneracy of the eigenvalue). When no time-reversal symmetry is present, the choice of any orthonormal basis for $\Pi$ gives a decomposition
\begin{equation}\label{noTRS}
\Pi=\bigoplus_{j=1}^{p} P_{j},\quad  \dim \Ran P_{j}=1,\quad P_{j}=P_{j}^*=P_{j}^2.
\end{equation}
When we have instead time-reversal symmetry, we proceed as follows. Denote by $C$ the complex conjugation with respect to the frame that makes $\eps$ into the form \eqref{eqn:eps=J}, and by $\Theta:=\eps C$; then $\Theta$ is an antiunitary operator satisfying $\Theta^2=-\Id$ in $\Ran \Pi \simeq \C^{m}$. From \eqref{hTRS} we see that we have the property
\[ \Theta \Pi =\Pi \Theta,\quad \text{or}\quad \eps \Pi= \Pi^t\eps,\quad \dim\Ran \Pi =2r. \]
We want to prove the following decomposition formula:
\begin{equation}\label{now4}
\Pi=\bigoplus_{j=1}^{r} P_{j},\quad  \dim \Ran P_{j}=2,\quad P_{j}=P_{j}^*=P_{j}^2,\quad \eps P_{j}=P_{j}^t\eps.
\end{equation}
Start by choosing an arbitrary unit vector $\mathbf{v}_1\in \Ran \Pi$. Define $\mathbf{v}_2=\Theta \mathbf{v}_1$. We have that $\Pi \mathbf{v}_2=\Pi\Theta \mathbf{v}_1=\Theta \Pi \mathbf{v}_1=\mathbf{v}_2$, hence $\mathbf{v}_2$ also belongs to $\Ran \Pi$. We also know that $\scal{\mathbf{v}_1}{\mathbf{v}_2}=0 $ and $\mathbf{v}_1=-\Theta \mathbf{v}_2$. Define 
 $$P_1=|\mathbf{v}_1\rangle \langle \mathbf{v}_1| +|\mathbf{v}_2\rangle \langle \mathbf{v}_2|.$$
Let $\mathbf{f}\in \C^{m}$. Then
\begin{align*}
\Theta P_1 \mathbf{f} &=\mathbf{v}_2 \overline{\scal{\mathbf{v}_1}{\mathbf{f}} }-\mathbf{v}_1\overline{\scal{ \mathbf{v}_2}{\mathbf{f}} }=
-\mathbf{v}_1\scal{\mathbf{f}}{\mathbf{v}_2} +\mathbf{v}_2 \scal{\mathbf{f}}{\mathbf{v}_1}=-\mathbf{v}_1\scal{ \Theta \mathbf{v}_2}{\Theta \mathbf{f}} +\mathbf{v}_2\scal{ \Theta \mathbf{v}_1}{\Theta \mathbf{f}}\\
&=\mathbf{v}_1\scal{ \mathbf{v}_1}{\Theta \mathbf{f}}+\mathbf{v}_2\scal{ \mathbf{v}_2}{\Theta \mathbf{f}}=P_1\Theta \mathbf{f}.
\end{align*}
This shows that $\Theta P_1=P_1\Theta$, or $\eps P_{1}=P_{1}^t\eps$. If $r>1$ we continue inductively. Let $\mathbf{v}_3$ be an arbitrary unit vector orthogonal to $\Ran P_1$ in $\Ran \Pi$. Define $\mathbf{v}_4=\Theta \mathbf{v}_3$. As before, we can show that $\Pi \mathbf{v}_4=\mathbf{v}_4$ and $\scal{ \mathbf{v}_3}{\mathbf{v}_4} =0$.  Moreover, using the properties of $\Theta$ we can show that $\mathbf{v}_4$ is also orthogonal on both $\mathbf{v}_1$ and $\mathbf{v}_2$. Define $P_2=|\mathbf{v}_3 \rangle \langle \mathbf{v}_3| + |\mathbf{v}_4 \rangle \langle \mathbf{v}_4|.$ The proof of $\Theta P_2=P_2\Theta$ is the same as the one for $P_1$. We now continue inductively until we exhaust the range of $\Pi$. We conclude that \eqref{now4} is proved. 

Now let us apply now \eqref{noTRS} (respectively \eqref{now4}) to each spectral projector $\Pi_j(0)$ of $\alpha(0)$. Assume that the multiplicity of the eigenvalue $\lambda_j(0)$ (\ie the dimension of $\Ran \Pi_j(0)$) equals $p_j(0)$, where $1\leq p_j(0)\leq m$. In presence of time-reversal symmetry, we have that this multiplicity is even, $p_j(0) = 2 r_j(0)$.
 
Then \eqref{noTRS} leads to
\begin{equation} \label{zumba9noTRS}
\Pi_j(0)=\bigoplus_{l_j=1}^{p_j(0)} P_{j,l_j}(0),\quad \dim\Ran P_{j,l_j}(0)=1,\quad P_{j,l_j}=P_{j,l_j}^*=P_{j,l_j}^2.
\end{equation}
Respectively \eqref{now4} leads to
\begin{align}\label{zumba9}
\Pi_j(0)=\bigoplus_{l_j=1}^{r_j(0)} P_{j,l_j}(0),\quad \dim\Ran P_{j,l_j}(0)=2,\quad P_{j,l_j}=P_{j,l_j}^*=P_{j,l_j}^2,\quad \eps P_{j,l_j}=P_{j,l_j}^t\eps.
\end{align}
Define
\[ A_j(0):=\sum_{l_j=1}^{m_j(0)} (l_j-1)\; P_{j,l_j}(0) \]
where $m_j(0) := p_j(0)$ in the non-symmetric case, and $m_j(0) := r_j(0)$ in the time-reversal symmetric case. From \ref{zumba9}, in the latter case we also have $\eps A_j(0)=A_j(0)^t\eps$. Seen as an operator acting on $\Ran\Pi_j(0)$, the spectrum of $A_j(0)$ consists of (doubly degenerate in the time-reversal symmetric case) eigenvalues given by $\{0,1,...,m_j(0)-1\}$. Of course, if $r_j(0)=1$ then $A_j(0)=0$. 

Now let $g_s\colon\R\to [0,1]$ be smooth, even, $g_s(0)=1$, and $\mathrm{supp}(g_s)\subset [-s/2,s/2]$. Define
\[ v_s(k):= s \, g_s(k) \left (\sum_{j=1}^{p_0} A_j(0) \right) \]
where $p_0$ is the total number of distinct eigenvalues of $\alpha(0)$. Then the support of $v_s$ is contained in $B_s(0)$ and obeys time-reversal symmetry. Define $\alpha_s(k)$ to be $ \eu^{\iu(h(k)+v_s(k))}$ if $k\in B_s(0)$, and let  $\alpha_s(k)=\alpha(k)$ outside $B_s(0)$. The family $\alpha_s$ obeys all the required properties, at the price of taking possibly a smaller $s$ to avoid overlapping of the eigenvalues at $0$, and converges uniformly in norm to $\alpha$ when $s\to 0$. 

\smallskip

\noindent \textsl{Step 2: \ref{item:split3}.} By assumption, the matrix $\alpha(0)$ has $p_0 \le m$ distinct eigenvalues, each having multiplicity $p_j \ge 1$. If $p_j>1$, the $j$-th eigenvalue is then associated to a \emph{cluster} of $p_j$ eigenvalues of the matrix $\alpha(k)$ for $k \ne 0$, which remain distinct due to the non-degeneracy condition and are well-separated from the others, but cross each other at $k=0$. Let $C_j$ be a simple, positively oriented contour, independent of $|k|<R$, enclosing the $j$-th cluster and no other eigenvalues. Let $\Pi_j(k)$ be the spectral projection of $\alpha(k)$ corresponding to the $j$-th cluster. We have
\begin{align}\label{zumba4}
\Pi_j(k)=\frac{\iu}{2\pi}\int_{C_j}(\alpha(k)-z)^{-1}\,\di z,\quad \eps \, \Pi_j(k)= \Pi_j(-k)^t\,\eps,\quad |k|<R.
\end{align}
These spectral projectors depend smoothly on $k$: in particular, $\lim_{k\to 0}\norm{\Pi_j(k)-\Pi_j(0)}=0$.

Let $\Pi = \Pi(k)$ denote any of the spectral projections for which $p_j > 1$, \ie which is associated to an eigenvalue cluster in which there is a crossing; if the $j$-th eigenvalue is already non-degenerate, \ie if $p_j = 1$, we do not need to operate on it. We label the eigenvalues in the cluster associated to $\Pi$ in increasing order of their principal arguments%
\footnote{If the crossing occurs at $\lambda = -1$, then one can easily adapt the following argument taking arguments for the eigenvalues in $[0, 2\pi)$.}%
: $\mathrm{Arg}(\lambda_l(k)) =: \phi_l(k) \le \phi_{l+1}(k)$ for all $l \in \set{1, \ldots, p}$ and $k$ in a possibly smaller interval $B_s(0) \subset B_R(0)$, where principal arguments of the eigenvalues of $\alpha(k)$ are well defined. We can moreover assume that $s>0$ is so small that $\phi_l(k)$ depends continuously on $k$; notice moreover that all $\phi_l(k)$ are close to a common value $\phi(0)$ (the argument of the eigenvalue at the crossing). Corresponding to this labelling, there is a splitting $\Pi(k) = \bigoplus_{l=1}^{p} P_l(k)$, with $\dim \Ran P_l = 1$, which is continuous for $k \ne 0$ (\ie the individual $P_l$'s become ill-defined at $k = 0$).
 
By taking a possibly smaller $s>0$, we may assume that
\[ \norm{\Pi(k) - \Pi(0)} < 1 \quad \text{for } |k| \le s. \]
Then the two projectors are interpolated by the Kato-Nagy unitary $W(k)$ \cite[Sec.~I.6.8]{Kato66}, \ie
\[ \Pi(k) = W(k) \, \Pi(0) \, W(k)^*, \quad W(k)^{-1} = W(k)^*. \]
Let $\set{\Psi_l(\pm s)}_{l = 1, \ldots, p}$ be orthonormal bases in $\Ran \Pi(\pm s)$, such that $P_l(\pm s) \Psi_l(\pm s) = \Psi_l(\pm s)$. Then $\set{W(\pm s)^* \, \Psi_l(\pm s)}_{l = 1 \ldots, p}$ determine two bases in $\Ran \Pi(0)$. Denote by $V_s$ the change-of-basis matrix associated to the unitary transformation of $\Ran \Pi(0)$ which maps one into the other:
\[ W(s)^* \, \Psi(s) =\left( W(-s)^* \, \Psi(-s) \right) \act V_s, \quad V_s \in U(p). \]
Write $V_s = \eu^{\iu v_s}$, with respect to any choice of the logarithm. Let also $\Phi^{(s)}(x)$ be defined by
\[ \Phi^{(s)}(x) :=  \left( W(-s)^* \, \Psi(-s) \right) \act \eu^{\iu \, x \, v_s}, \] 
\ie $\Phi^{(s)}(x)$ is the basis in $\Ran \Pi(0)$ whose coefficients in the basis $W(-s)^* \, \Psi(-s)$ are the entries of the matrix $\eu^{\iu \, x \, v_s}$. Set finally
\[ A^{(s)}(k) = \sum_{l=1}^{p} \phi_l^{(s)}(k) \, P_l^{(s)}(k), \]
where
\begin{gather*}
\phi_l^{(s)}(k) := \phi_l(-s) + \frac{k+s}{2s} \left( \phi_l(s) - \phi_l(-s) \right), \\
P_l^{(s)}(k) := W(k) \left|\Phi_l^{(s)} \left( \frac{k+s}{2s} \right) \right\rangle \left\langle \Phi_l^{(s)} \left( \frac{k+s}{2s} \right) \right| W(k)^*.
\end{gather*}
Notice that $\sum_{l=1}^{p} P_l^s(k) = \Pi(k)$ for all $k \in B_s(0)$. The operator $A^{(s)}(k)$ is then self-adjoint, and has non-degenerate spectrum for all $|k|<s$, as its eigenvalues interpolate linearly the arguments of the ones at $k=-s$ and at $k=s$ (which are assumed to be distinct and labelled in increasing order).

We apply this procedure to any of the spectral projections $\Pi_j(k)$ of $\alpha(k)$ which have $p_j > 1$. In a self-explanatory notation, set for a sufficiently small $s>0$
\[ h_s(k) := \sum_{j=1}^{p_0} A_j^{(s)}(k), \quad h_s(\pm s) = h(\pm s). \]
Define finally $\alpha_s(k) = \eu^{\iu \, h_s(k)}$ for $|k| < s$, and $\alpha_s(k) = \alpha(k)$ elsewhere. Observe that, since all the eigenvalues in a cluster are close to the common value at the crossing, $\norm{h_s(k) - h(k)} \xrightarrow[s \to 0]{} 0$ uniformly in $k$. The families $\alpha_s$ define the required continuous approximants of the family $\alpha$, which have non-degenerate spectrum in $B_R(k_0)$.
\end{proof}

\begin{lemma}[Analytic approximation] \label{lemma:ApproxAnalytic}
Let $\alpha = \set{\alpha(k)}_{k \in \R}$ be a continuous family of $m\times m$ unitary matrices satisfying \ref{item:alpha_periodic} and \ref{item:alpha_TRS}, \ie periodicity and TRS.
Then it is possible to construct a sequence $\alpha_j$ of \emph{real analytic} families of unitary matrices also satisfying the above properties, and such that
\[ \lim_{j\to\infty}\sup_{k\in \R}\norm{\alpha_j(k)-\alpha(k)}=0. \]
\end{lemma}
\begin{proof}
Let $f(k)=\pi^{-1}(1+k^2)^{-1}$. If $\nu>0$ we define $f_\nu(k):=\nu^{-1}f(k/\nu)$. The function $f$ is positive for $k \in \R$, even, $f_\nu$ is an approximation of the Dirac distribution and is analytic on the strip $\{|{\rm Im}(k)|<\nu\}$. Then the family of matrices
\begin{equation}\label{zumba12}
\mu_\nu(k):=\int_\R f_\nu (k-k') \alpha(k') \, \di k'
\end{equation}
is $\Z$-periodic and real analytic%
\footnote{%
Due to periodicity, it is enough to find an analytic extension to a bounded open neighbourhood $O$ of $[0,1]$ of the form
\[ O := \set{z = k + \iu h \in \C : |k|<2, \: |h|< \nu' < \nu}. \]
The existence of such an extension in $O$ then follows from Fubini's theorem and the Cauchy integral formula for $f_\nu$, in view of the estimate 
\[ |f_\nu(z-k')| \le \frac{\nu}{\pi} \, \frac{1}{\nu^2 - (\nu')^2 + (k-k')^2}, \quad z = k + \iu h \in O, \: k' \in \R. \]
}%
. {\it A priori} $\mu_\nu(k)$ is neither self-adjoint nor unitary, but it obeys  $\eps \, \mu_\nu(k)=\mu_\nu(-k)^t \, \eps$. Also, $\mu_\nu$ converges uniformly to $\alpha$ when $\nu\to 0$, which implies
\begin{equation}\label{zumba13}
\lim_{\nu\to 0}\sup_{k\in\R}\norm{\mu_\nu(k)\mu_\nu(k)^*-{\Id}}=0.
\end{equation}
If $\nu$ is small enough, the operator $\mu_\nu(k)\mu_\nu(k)^*$ is self-adjoint and positive, hence we can define
\begin{equation}\label{zumba14}
\alpha_\nu(k):=\{\mu_\nu(k)\mu_\nu(k)^*\}^{-1/2} \mu_\nu(k),\quad \alpha_\nu(k)\alpha_\nu(k)^*=\Id.
\end{equation}
Clearly, $\alpha_\nu$ is a continuous $\Z$-periodic family of unitary matrices, and converges uniformly to $\alpha$ when $\nu\to 0$. We still need to prove that $\alpha_\nu$ satisfies TRS and is real analytic.

If $\nu$ is small enough we can write
\[ \alpha_\nu(k)=\frac{\iu}{2\pi} \int_{|z-1|=1/2} z^{-1/2}(\mu_\nu(k)\mu_\nu(k)^*-z)^{-1}\mu_\nu(k)\,\di z \]
Now the family of matrices
$$\gamma_\nu(k):=\mu_\nu(k)\mu_\nu(k)^*=\int_{\R^2} f_\nu (k-k')f_\nu (k-k'') \alpha(k')\alpha(k'')^* \, \di k' \di k''$$
has a holomorphic extension to the strip $\{|{\rm Im}(k)|<\nu\}$. If $\nu$ is small enough, due to \eqref{zumba13} it follows that $(\mu_\nu(k)\mu_\nu(k)^*-z)^{-1}$ exists for all $|z-1|=1/2$ and is uniformly bounded on $\{|{\rm Im}(k)|<\eta\}$ with $0<\eta\ll \nu$. Together with \eqref{zumba12} this proves that $\alpha_\nu$ has a holomorphic extension to $\{|{\rm Im}(k)|<\eta\}$.

The last thing we need to prove is $\eps \, \alpha_\nu(k)=\alpha_\nu(-k)^t \, \eps$, namely TRS. From \eqref{zumba14} we see that if $\nu$ is small enough one can write
$$ \alpha_\nu(k)=\sum_{l\geq 0}a_l [\mu_\nu(k)\mu_\nu(k)^*-\Id]^l\mu_\nu(k),\quad (1+x)^{-1/2}=\sum_{l\geq 0} a_lx^l.$$
We have already seen that $\eps \, \mu_\nu(k)=\mu_\nu(-k)^t \, \eps$ and this leads to $\eps \, \mu_\nu(k)^*= \overline{\mu_\nu(-k)} \, \eps$. Then
\begin{align*}
\eps \, [\mu_\nu(k)\mu_\nu(k)^*-\Id]^l\mu_\nu(k)&=[\mu_\nu(-k)^t\overline{\mu_\nu(-k)}-\Id]^l \mu_\nu(-k)^t \, \eps \\
& = \mu_\nu(-k)^t[\overline{\mu_\nu(-k)}\mu_\nu(-k)^t-\Id]^l \, \eps
\end{align*}
where the second identity can be proved by induction with respect to $l$. The proof of the Lemma is over.
\end{proof}

After the above auxiliary Lemmas, we come to the main technical results of the Appendix.

\begin{proposition}\label{propo-dec-1}
Let $\alpha$ be a continuous family of $m\times m$ matrices, $m=2n$, satisfying \ref{item:alpha_periodic} and \ref{item:alpha_TRS}. Then it is possible to construct a sequence $\alpha_j$ of families of matrices for which the following hold:
\begin{itemize}
\item they depend \emph{continuously} on $k$, and are \emph{real analytic} in small balls of fixed ($j$-dependent) radius around every half-integer;
\item they satisfy \ref{item:alpha_periodic} and \ref{item:alpha_TRS};
\item they have \emph{doubly} degenerate eigenvalues at half-integer values of $k$, and \emph{non-degenerate} eigenvalues away from half-integer values of $k$;
\item they approximate $\alpha$ in the sense that
\[ \lim_{j\to\infty}\sup_{k\in \R}\norm{\alpha_j(k)-\alpha(k)}=0. \] 
\end{itemize}
\end{proposition}
\begin{proof} 
\noindent \textsl{Step 0.} Using Lemma \ref{lemma:ApproxAnalytic}, we may assume without loss of generality that $\alpha$ is actually a real analytic family, and thus satisfies Assumption \ref{assum:alpha}.

\noindent  \textsl{Step 1.} We begin by fixing the spectrum at $0$ and $1/2$, and by periodicity at all half-integer values of $k$. 

Assume that the spectrum of $\alpha(0)$ consists of $1\leq p_0\leq  n$ even degenerate eigenvalues labeled as $\{\lambda_1(0),\ldots,\lambda_{p_0}(0)\}$ in the increasing order of their principal arguments. The even degeneracy is due to Kramers degeneracy, compare Lemma \ref{lemma:Kramers}. If $s>0$ is sufficiently small and $|k|<s$, due to the continuity of $k \mapsto \alpha(k)$ we know that the spectrum of $\alpha(k)$ will also consist of well separated clusters of eigenvalues. Let $\Pi_j(k)$ be the spectral projection of $\alpha(k)$ corresponding to the $j$-th cluster, as in \eqref{zumba4}. The matrix $\alpha(k)$ is block diagonal with respect to the decomposition $\mathbb{C}^{m}=\bigoplus_{j=1}^{p_0} \Pi_j(k)\mathbb{C}^{m}$, \ie $\alpha(k)=\sum_{j=1}^{p_0} \Pi_j(k)\alpha(k)\Pi_j(k)$, $|k|<s$. 

Define $$\widetilde{\alpha}(k):=\sum_{j=1}^{p_0} \lambda_j(0)\Pi_j(k)=\eu^{\iu\sum_{j=1}^{p_0} {\rm Arg}(\lambda_j(0))\Pi_j(k)},\quad |k|<s.$$
The matrix $\widetilde{\alpha}(k)$ is unitary, commutes with $\alpha(k)$, and $\eps \, \widetilde{\alpha}(k) = \widetilde{\alpha}(-k)^t \,\eps$ if $|k|<s$. Define $\gamma(k):=\widetilde{\alpha}^{-1}(k)\alpha(k)$; we have that $\gamma(k)$ is unitary, commutes with $\alpha(k)$,  $\eps \, \gamma(k) = \gamma(-k)^t \, \eps$ if $|k|<s$ and $\lim_{k\to 0}\gamma(k)=\Id$. In particular, $-1$ is never in the spectrum of $\gamma(k)$. Going through the Cayley transform (Proposition \ref{prop:Cayley}) we can find a self-adjoint matrix $\widetilde{h}(k)$ such that 
\[ \gamma(k)=\eu^{\iu \widetilde{h}(k)},\quad \widetilde{h}(0)=0,\quad \eps \, \widetilde{h}(k) = \widetilde{h}(-k)^t \, \eps,\quad [\Pi_j(k),\widetilde{h}(k)]=0,\quad |k|<s. \]
Using that $\alpha(k)=\alpha(k+r)$ for all $r\in\Z$ we obtain
\begin{align}\label{zumba6}
\alpha(k+r)=\eu^{\iu\left(\sum_{j=1}^{p_0} {\rm Arg}(\lambda_j(0))\Pi_j(k) +\widetilde{h}(k)\right)},\quad |k|<s,\quad \forall r\in\Z.
\end{align}
We proceed similarly for $k\in\{|k\pm 1/2|<s\}$. Remember that due to the periodicity of $\alpha$, the spectrum will also have Kramers degeneracy at $k=\pm 1/2$. Furthermore, there will be the same number of distinct eigenvalues at $k=+ 1/2$ and $k=-1/2$, and we can choose the same ordering for them, \ie $\lambda_j(1/2)=\lambda_j(-1/2)$ for all $1\leq j\leq p_{1/2}\leq n$; note that there is no connection with the numbering at $k=0$. From the Riesz integral formula, choosing a contour $\tilde{C}_j$ around each $\lambda_j(1/2)=\lambda_j(-1/2)$ and using that $\alpha(k)=\alpha(k+1)$ we obtain $\Pi_j(-1/2)=\Pi_j(1/2)$ and moreover
\[ \eps \Pi_j(k)=\Pi_j(-k)^t\eps,\quad \eps \Pi_j(-1/2)=\eps \Pi_j(1/2) = \Pi_j(1/2)^t\eps,\qquad  |k\pm 1/2|<s. \]
As in \eqref{zumba6} we get
\begin{align}\label{zumba8}
\alpha(k+r)=\eu^{\iu\left(\sum_{j=1}^{p_{1/2}} {\rm Arg}(\lambda_j(1/2))\Pi_j(k) +\widetilde{h}(k)\right)},\quad |k\pm 1/2|<s,\quad \forall r\in\Z.
\end{align}

From \eqref{zumba6} and \eqref{zumba8} we see that $\alpha(k)=\eu^{\iu h(k)}$ on the set
$$\Omega_s:=\bigcup_{r\in\Z} \{|k-r|<s\}\cup \{|k-r-1/2|<s\}.$$
We see that $\Omega_s$ is symmetric with respect to the origin and consists of a union of disjoint open intervals centred around all integers and half-integers. We apply the local splitting Lemma \ref{lemma:splitting}\ref{item:split2} in each of these intervals, obtaining an approximant $\widehat{\alpha}_s$ with the required spectral degeneracies. Define $\alpha_s(k)$ to be $\widehat{\alpha}_s(k)$ if $k\in\Omega_s$, and let  $\alpha_s(k)=\alpha(k)$ outside $\Omega_s$. The family $\alpha_s$ obeys all the required properties and converges uniformly in norm to $\alpha$ when $s\to 0$. 

\smallskip

\noindent  \textsl{Step 2.} Now we will show how we can make $\alpha_s$ to also have non-degenerate spectrum in the interval $(0,1/2)$. In order to simplify notation, we drop the subscript $s$ and we assume that $\alpha$ is a $\Z$-periodic TRS family of unitary matrices and has doubly degenerate spectrum at $0$ and $1/2$ (hence at all integers and half-integers). By applying the analytic approximation Lemma \ref{lemma:ApproxAnalytic}, we may moreover assume that $\alpha(k)$ depends analytically on $k$.

By the Analytic Rellich Theorem (compare \cite[Sec.~2.7.4]{CorneanHerbstNenciu15}), one can always choose the eigenvalues of $\alpha(k)$ to also depend analytically on $k$. These eigenvalues can cross, but by analyticity they either cross at isolated points, or they coincide throughout $[0,1/2]$. In particular, there might be eigenvalues which stay doubly degenerate after stemming from one of the doubly degenerate eigenvalues at $k=0$.

We first lift this degeneracy, as follows. Let $a>0$ be such that all eigenvalue clusters of $\alpha(k)$ are well-separated for $k \in [0,2a]$. In particular, we may assume that for $k \in [0,2a]$ there exists a $\phi_0$ such that $-1$ is always in the resolvent of $\alpha(k) \, \eu^{\iu (\phi_0 + \pi)}$. By going through the Cayley transform (Proposition \ref{prop:Cayley}), we can then write $\alpha(k) = \eu^{\iu h(k)}$ in this interval. We apply the local splitting Lemma \ref{lemma:splitting}\ref{item:split1} to separate all eigenvalues at $k=a$, and then the analytic approximation argument of Lemma \ref{lemma:ApproxAnalytic}. The resulting approximating family $\widetilde{\alpha}$ will have non-degenerate spectrum at $k=a$, and hence cannot have eigenvalues which stay doubly degenerate throughout $[0,1/2]$. 

By compactness of the interval $[0, 1/2]$ and analyticity of the eigenvalues of $\widetilde{\alpha}(k)$, there can only by a finite number of points $\set{k_1, \ldots, k_N} \subset (0,1/2)$ where eigenvalues intersect. Let $s>0$ be so small that the intervals $B_s(k_i) = [k_i - s, k_i+s]$, $i \in \set{1, \ldots, N}$, do not overlap. We also assume that $k_1 -s > 0$ and $k_N + s < 1/2$. Applying the local splitting Lemma \ref{lemma:splitting}\ref{item:split3} in each of these intervals, we obtain there a smooth approximant $\widetilde{\alpha}_s(k)$ for $k \in B:= \bigcup_{i=1}^{N} B_s(k_i)$ with the desired degeneracies. Set finally
\[ \alpha_s(k) := \begin{cases} 
\widetilde{\alpha}_s(k) & \text{if } k \in B, \\
\alpha(k) & \text{if } k \in [0,1/2] \setminus B,
\end{cases}  \quad k \in [0,1/2]. \]
Extend the definition of $\alpha_s(k)$ for $k \in [-1/2,0]$ by setting $\alpha_s(k) := (\eps \, \alpha_s(-k) \, \eps^{-1})^t$, and then to the whole $\R$ by periodicity, namely $\alpha_s(k+r) := \alpha_s(k)$ for all $k \in [-1/2, 1/2]$ and $r \in \Z$. Then $\alpha_s$ satisfies all the required properties. This concludes the proof of the Proposition.
\end{proof}

\begin{proposition}\label{propo-dec-2}
Let $\alpha$ be a continuous family of $m\times m$ matrices, $m=2n$, satisfying \ref{item:alpha_periodic} and \ref{item:alpha_TRS}. Then it is possible to construct a sequence $\alpha_j$ of families of unitary matrices for which the following hold:
\begin{itemize}
\item they depend \emph{continuously} on $k$, and are \emph{real analytic} in small balls of fixed ($j$-dependent) radius around every half-integer;
\item they have \emph{non-degenerate} eigenvalues for all $k \in \R$, and the spectrum is \emph{symmetric} with respect to the exchange $k \to -k$;
\item they approximate $\alpha$ in the sense that
\[ \lim_{j\to\infty}\sup_{k\in \R}\norm{\alpha_j(k)-\alpha(k)}=0. \] 
\end{itemize}
\end{proposition}
\begin{proof}
Arguing as in the proof of the above Proposition \ref{propo-dec-1}, we obtain a family $\alpha_j(k)$ which satisfies all the required properties, but has (doubly) degenerate spectrum at half-integer values of $k \in \R$. Applying Lemma \ref{lemma:splitting}\ref{item:split3} to small symmetric intervals of width $2s>0$ around such values, we can open gaps between these eigenvalues. Since the procedure of Lemma \ref{lemma:splitting}\ref{item:split3} interpolates linearly the arguments of the eigenvalues between $n/2 - s$ and $n/2 + s$, $n \in \Z$, the symmetry of the spectrum of $\alpha_j$, which is implied by TRS, is preserved, since the linear interpolation is then constant.
\end{proof}


\goodbreak


\bigskip \bigskip

{\footnotesize

\begin{tabular}{rl}
(H.D. Cornean) & \textsc{Department of Mathematical Sciences, Aalborg University} \\
 &  Fredrik Bajers Vej 7G, 9220 Aalborg, Denmark \\
 &  \textsl{E-mail address}: \href{mailto:cornean@math.aau.dk}{\texttt{cornean@math.aau.dk}} \\
 \\
(D. Monaco) & \textsc{Fachbereich Mathematik, Eberhard Karls Universit\"{a}t T\"{u}bingen} \\
 &  Auf der Morgenstelle 10, 72076 T\"{u}bingen, Germany \\
 &  \textsl{E-mail address}: \href{mailto:domenico.monaco@uni-tuebingen.de}{\texttt{domenico.monaco@uni-tuebingen.de}} \\
 \\
(S. Teufel) & \textsc{Fachbereich Mathematik, Eberhard Karls Universit\"{a}t T\"{u}bingen} \\
 &  Auf der Morgenstelle 10, 72076 T\"{u}bingen, Germany \\
 &  \textsl{E-mail address}: \href{mailto:stefan.teufel@uni-tuebingen.de}{\texttt{stefan.teufel@uni-tuebingen.de}} \\
\end{tabular}

}

\end{document}